\documentclass{article}
\usepackage[utf8]{inputenc}

\usepackage{fullpage}
\usepackage{amsfonts}
\usepackage{amsmath}
\usepackage{amssymb}
\usepackage{amsthm}
\usepackage{cite}
\usepackage{subfigure}
\usepackage{bm}
\usepackage{graphics}
\usepackage{graphicx}

\usepackage{subcaption}

\usepackage{latexsym,amsfonts,amssymb,amsmath,amsthm}

\usepackage{booktabs} 
\usepackage{graphicx}
\usepackage{pgfplots}
\usepackage[all]{nowidow}
\usepackage{soul}
\usepackage[utf8]{inputenc}
\usepackage{tikz}
\usetikzlibrary{automata}
\usepackage{multicol,comment}
\usepackage{algpseudocode,algorithm,algorithmicx}

\usepackage{cite}
\usepackage[inline]{enumitem} 

\usepackage{ulem}

\definecolor{blue}{HTML}{1F77B4}
\definecolor{orange}{HTML}{FF7F0E}
\definecolor{green}{HTML}{2CA02C}

\pgfplotsset{compat=1.14}

\newtheorem{tm}{Theorem}[section]
\newtheorem{prop}[tm]{Proposition}

\newtheorem{lem}[tm]{Lemma}

\newtheorem{rk}[tm]{Remark}

\numberwithin{equation}{section}
\numberwithin{tm}{section}

\usepackage[makeroom]{cancel}

\title{Dynamics of solutions to a multi-patch epidemic model with a saturation incidence mechanism}

\author{ Yawo Ezunkpe\footnote{yawo.ezunkpe@sjsu.edu; Department of Aerospace Engineering, San Jose State University, San Jose, California, USA},\quad   Cynthia T. Nnolum\footnote{cynthia.nnolum@unlv.edu; Department of Mathematical Sciences, University of Nevada Las Vegas, Las Vegas, USA}, \quad  Rachidi B. Salako\footnote{rachidi.salako@unlv.edu; Department of Mathematical Sciences, University of Nevada Las Vegas, Las Vegas, USA},  \quad and\quad Shuwen Xue\footnote{sxue@niu.edu; Department of Mathematical Sciences, Northern Illinois University, Dekalb, IL 60115, USA. } }
\date{}

\begin{document}

\maketitle

\begin{abstract} 
This study examines the behavior of solutions in a multi-patch epidemic model that includes a saturation incidence mechanism. When the fatality rate due to the disease is not null, our findings show that the solutions of the model tend to stabilize at disease-free equilibria. Conversely, when the disease-induced fatality rate is null, the dynamics of the model become more intricate. Notably, in this scenario, while the saturation effect reduces the basic reproduction number $\mathcal{R}_0$, it can also lead to a backward bifurcation of the endemic equilibria curve at $\mathcal{R}_0=1$. Provided certain fundamental assumptions are satisfied, we offer a detailed analysis of the global dynamics of  solutions based on the value of $\mathcal{R}_0$. Additionally, we investigate  the asymptotic profiles of  endemic equilibria as  population dispersal rates tend to zero. To support and illustrate our theoretical findings, we conduct numerical simulations.




 
\end{abstract}

\noindent{\bf Keywords}: Patch model; Epidemic model;
Asymptotic Behavior; Persistence.

\smallskip

{
\noindent{\bf 2020 Mathematics Subject Classification}: 34D05, 34D23, 92D25, 92D30, 37N25}

\section{Introduction} 
Over the past few decades, numerous epidemic models have been proposed and analyzed \cite{1991book,arino2003multi,arino2009diseases,2008book,brauer2019mathematical}. The predictions about disease dynamics derived from both theoretical and numerical studies of these models have proven essential for devising and implementing effective disease control strategies \cite{Diekmann2000,hethcote1976qualitative}. In most of these works, selecting  appropriate incidence mechanism in epidemic modeling plays essential role on the dynamics of solutions. Indeed, as strongly highlighted by the works \cite{AM2004, DD2003,GM2003,HD1991,LHL1987,LLL1986}, a simple change in the incidence mechanism of an epidemic model may lead to substantial changes in dynamical behaviors of solutions to the model.   Additionally, factors such as environmental variability and population movements  significantly influence the spread of diseases within populations.

\medskip

\noindent In the influential work \cite{allen2007asymptotic}, the authors introduce and analyze the following multi-patch epidemic model:
\begin{equation}\label{standard-incidence}
\begin{cases}
\displaystyle\frac{d S_i}{dt}=d_S\sum_{j\in\Omega}L_{ij} S_j- \frac{\beta_{i}S_iI_i}{S_i+I_i}+\gamma_{i} I_{i},  & i\in\Omega,\ t>0, \cr 
\displaystyle\frac{d I_{i}}{dt}=d_I\sum_{j\in\Omega}L_{ij} I_{j}+ \frac{\beta_{i}S_iI_i}{S_i+I_i}-\gamma_{i} I_{i}, & i\in\Omega,\ t>0.
\end{cases}
\end{equation}
This model explores how population movement and spatial heterogeneity affect disease dynamics. It represents a population distributed across a discrete network $\Omega$, consisting of a finite number $|\Omega|=n$ of patches (or cities). For each patch $i \in \Omega$, $S_i(t)$ and $I_i(t)$ denote the number of susceptible and infected individuals at time $t>0$ on patch-i, respectively. The parameters $L_{ij} \geq 0$ for $i \neq j\in\Omega$ represent the degree of movements  from patch $j$ to patch $i$. For $i\in\Omega$,  $L_{ii} = -\sum_{j \neq i} L_{ji}$ is  the total degree of movement out from patch $i$. The disease-specific parameters $\beta_i$ and $\gamma_i$ denote the local transmission and recovery rates on patch $i$, respectively, while the positive numbers $d_S>0$ and $d_I>0$ are the dispersal rates for susceptible and infected individuals, respectively.  An important fact about system \eqref{standard-incidence} is that the total population size is constant over time since the model does not account for changes in population demographics. Under the assumption that the connectivity matrix $L = (L_{ij})$ is symmetric and irreducible, \cite{allen2007asymptotic} demonstrates that, when the total initial population size $N>0$ is given, the basic reproduction number (BRN) $\hat{\mathcal{R}}_0$ (as defined in formula \eqref{R-hat-0} below) serves as a critical threshold for determining disease persistence. Specifically, if $\hat{\mathcal{R}}_0\le 1$, the model \eqref{standard-incidence} predicts eventual disease extinction. Conversely, if $\hat{\mathcal{R}}_0>1$, the model \eqref{standard-incidence} predicts disease persistence and the existence of a unique endemic equilibrium (EE) solution. Additionally, their study reveals that as $d_S$ approaches zero, the profiles of the EEs indicate that if there is at least one ``low risk" patch (that is a patch where the disease transmission rate is less than the recovery rate), the infected component of the EEs will approach zero across all patches. Biologically, this suggests that reducing the dispersal rate of the susceptible population can significantly mitigate the disease's impact. For further insights into system \eqref{standard-incidence}, interested readers can consult \cite{chen2020asymptotic, gao2021impact, gao2020fast, gao2020does, li2019dynamics, SW2024a}. For some recent studies on continuous time and space related epidemic models to \eqref{standard-incidence}, we refer to \cite{Allen, CuiLamLou, CuiLou, li2020dynamics, LSS2023, LouSalako2021, Peng2013, PengZhao, salako2024} and the references therein.

   The disease standard-incidence mechanism, given by $\beta_iS_iI_i/(S_i+I_i)$, is employed in modeling system \eqref{standard-incidence}. This incidence rate, as introduced by \cite{Jong}, is based on the random-mixing assumption, where the probability of a susceptible individual $S_i$ contracting the infection is proportional to the encounter rate with infected individuals, represented by $I_i/(S_i + I_i)$. In contrast, the mass-action transmission mechanism, originating from \cite{kermack1927contribution}, assumes that the rate of new infections per unit area and time is directly proportional to the product of the numbers of infected and susceptible individuals. Consequently, the incidence function $\beta_iS_iI_i$ is used in the mathematical modeling. Studies such as \cite{li2023sis,SW2024c,SW2024a,salako2024} analyze system \eqref{standard-incidence} with the mass-action transmission rate described by
\begin{equation}\label{mass-action-incidence}
\begin{cases}
\displaystyle\frac{d S_i}{dt}=d_S\sum_{j\in\Omega}L_{ij} S_j- \beta_{i}S_iI_i+\gamma_{i} I_{i},  & i\in\Omega,\ t>0, \cr 
\displaystyle\frac{d I_{i}}{dt}=d_I\sum_{j\in\Omega}L_{ij} I_{j}+ \beta_{i}S_iI_i-\gamma_iI_i, & i\in\Omega,\ t>0,
\end{cases}
\end{equation}
and investigate the global dynamics of its solutions. The parameters in system \eqref{mass-action-incidence} carry the same meanings as those in system \eqref{standard-incidence}.  When the total population size $N>0$ is given, both systems \eqref{standard-incidence} and \eqref{mass-action-incidence} have the same  (unique) disease free equilibrium (DFE). However, they have different BRNs as the BRN of system \eqref{standard-incidence} is independent of $N$ while the BRN of system \eqref{mass-action-incidence} depends linearly on $N$. Moreover, under appropriate hypotheses, it was established in \cite{SW2024a} that system \eqref{mass-action-incidence} may have at least two EEs for a range of its BRN less than one.  The latter result strongly highlights the effect of incidence mechanism on the dynamics of these simple multiple patches epidemics models. It also illustrates how population movements may complicate disease dynamics because such interesting multiplicity result of EEs does not hold for the single-strain model \eqref{mass-action-incidence}. For related results on the PDE analogue of system \eqref{mass-action-incidence}, we refer interested readers to \cite{CS2023b,castellano2022effect,DengWu,Li2018,peng2023novel,peng2021global,salako2023,salako2024,tao2023analysis,wen2018asymptotic,WuZou} and the references cited therein.

  In the current work, we consider the saturated-incidence function, represented by $\beta_iS_iI_i/(\zeta_i+S_i+I_i)$, and investigate the dynamics of solutions to the multiple patch epidemic system
\begin{equation}\label{model}
\begin{cases}
\displaystyle\frac{d S_i}{dt}=d_S\sum_{j\in\Omega}L_{ij} S_j- \frac{\beta_{i}S_iI_i}{ \zeta_i+S_i+I_i}+\gamma_{i} I_{i},  & i\in\Omega,\ t>0, \cr 
\displaystyle\frac{d I_{i}}{dt}=d_I\sum_{j\in\Omega}L_{ij} I_{j}+ \frac{\beta_{i}S_iI_i}{ \zeta_i+S_i+I_i}-\gamma_{i} I_{i}-\mu_iI_i, & i\in\Omega,\ t>0,
\end{cases}
\end{equation}
where $\mu_i\ge 0$, $i\in\Omega$, is the disease induced fatality rate on the patch-i.  For  $i\in\Omega$, $\zeta_i>0$ accounts for the  saturation effect of the population during the mixing of the infected population with the susceptible population on the patch-i. Following \cite{GLPZ2023}, $\zeta_i$, $i\in\Omega$, may also be viewed as a portion of the population on patch-i that is naturally resistant to  infection.  When $\zeta_i=0$ and $\mu_{i}=0$ for all $i \in \Omega$, system \eqref{model} simplifies to system \eqref{standard-incidence}. In this study, we focus on the scenario where $\zeta_i > 0$ for all patches $i\in\Omega$. A PDE version of system \eqref{model}, which involves populations engaging in local and random movements   within spatially and temporally varying environments, was recently analyzed in \cite{GLPZ2023}. Additionally, \cite{FLY2024} explored system \eqref{model} with $\bm\mu:=(\mu_i)^T_{i\in\Omega} = \bm0$ in continuous space environments, considering populations that employ nonlocal dispersal movements. Our current work builds upon these studies by examining the dynamics of solutions to system \eqref{model}, which is a space-discrete and time-continuous model. Notably, some of our key findings are novel even in the context of the continuous models discussed in \cite{GLPZ2023, FLY2024}. Specifically, for $\bm\mu=\bm 0$:  Theorems \ref{TH2} and \ref{TH3} establish the global stability of the DFE under certain general conditions; Theorem \ref{TH6} explores the structure of the set of the  EE solutions   as  BRN varies; and Theorem \ref{TH5} confirms the uniqueness of the EE solution under specific assumptions about the model parameters. When the disease induced fatality rate is positive on at least one patch, as mentioned above, Theorem \ref{TH1} shows that the disease will be eventually eradicated.

 When the disease induced fatality rate is negligible, i.e., $\bm \mu = \bm 0$, the BRN $\mathcal{R}_0$ of system \eqref{model} is strictly decreasing in positive $\bm \zeta =(\zeta_i)^T_{i \in \Omega}$  and strictly increasing with respect to the total population size (see Proposition \ref{prop2}). Additionally, Proposition \ref{prop2}-(iii)-(iv) demonstrate the existence of a critical total population size $\mathcal{N}_0$, which increases with respect to the infected population dispersal rate and saturation incidence $\bm\zeta$, respectively, and is independent of the susceptible population rate. The BRN of system \eqref{model} exceeds unity if and only if the total population size is greater than this critical threshold and the dispersal rate $d_{I}$ of the infected population is small. Moreover, Proposition \ref{prop2}-(iv-3) shows that a large saturation incidence can significantly lower the BRN $\mathcal{R}_0$. These findings underscore significant differences compared to the dynamics of solutions in the multiple patch epidemic model \eqref{standard-incidence}, where the BRN is unaffected by the total population size.

 There are several interesting studies on continuous space-time epidemic models. For some recent studies on PDE epidemic models, we refer interested readers to \cite{2008book,CuiLamLou,CuiLou,DNS2023,KuoPeng,li2020dynamics,Li2018,LouSalako2021,Peng2013,Tuncer2012,SLX2019,SS2024,WWK2022,SL2020}.

 The organization of the manuscript is as follows. Section \ref{Sec2} contains four subsections: The first subsection  provides the basic notations, assumptions, and definitions used throughout the work. The second subsection presents the main results along with their relevant biological implications. The third subsection includes extensive numerical simulations that illustrate these theoretical results. The final subsection  offers discussions and comparisons with previous findings. Section \ref{Sec3} contains preliminary results, and the proofs of the main results are detailed in Section \ref{Sec4}.

 \section{Notations, Assumptions, Definitions, and Main Results }\label{Sec2} 
 \subsection{Notations, Assumptions, and Definitions}
Throughout the paper, 
a bold letter always represents  a column vector in $\mathbb{R}^n$, and its no-bold form with a numeric subscript will be a component of it. For example, for any $\bm Z\in\mathbb{R}^n$, one has $\bm Z=(Z_1, \dots, Z_n)^T$, where $Z_j\in\mathbb{R}$ for  $j\in\Omega:=\{1,2,\cdots,n\}$. 
We write $\bm 0=(0, \dots, 0)^T$  and $\bm 1=(1, \dots, 1)^T$. For $\bm Z\in\mathbb{R}^n$, define 
 $$
{\bm Z}_{m}:=\min_{j=1,\cdots, n} Z_j,\quad  {\bm Z}_{M}:=\max_{j=1\cdots,n} Z_j,\quad 
 \|\bm Z\|_1:=\sum_{j=1}^n|Z_{j}|,\quad \text{and}\quad \|\bm Z\|_{\infty}:=\max_{j=1,\cdots, n}| Z_j|.
$$
We denote by ${\rm diag}(\bm Z)$ the diagonal matrix with diagonal entries $[{\rm diag}(\bm Z)]_{ii}= Z_i$ for all $i=1,\cdots,n$. Let $\mathbb{R}_+$ denote the set of nonnegative real numbers. Given $\bm Z, \bm Y\in\mathbb{R}^n$, we write:  $\bm Z\ge  \bm Y$ or $\bm Y\le \bm Z$ if $\bm Z-\bm Y  \in \mathbb{R}^n_+$ ; $\bm Z> \bm Y$ or $\bm Y<\bm Z$ if $\bm Z-\bm Y\in \mathbb{R}^n_+\setminus\{\bm 0\}$; and $\bm Z\gg \bm Y$ or $\bm Y\ll \bm Z$ if $ Z_i>Y_i$  for all $i=1,\cdots,n$. 
 Next, for  $\bm Z, \bm Y\in \mathbb{R}^n$, define the Hadamard product $ \bm Z\circ \bm Y :=(Z_1Y_1,\cdots, Z_n Y_n)^T$, 
 and set  ${\bm Z}/{\bm Y}=(Z_1/Y_1,\cdots,$
 $ Z_n/ Y_n)^T$ if $Y_i\neq 0$ for all $i\in\Omega$.
Adopting these notations, system \eqref{model} can be rewritten as 
 \begin{equation*}
     \begin{cases}
         \bm S'=d_S\mathcal{L}\bm S+(\bm\gamma-\bm\beta\circ\bm S/(\bm\zeta+\bm S+\bm I))\circ\bm I, & t>0,\cr 
        \bm I'=d_I\mathcal{L}\bm I+(\bm\beta\circ \bm S/(\bm\zeta+\bm S+\bm I)-\bm\gamma)\circ\bm I-\bm\mu \circ\bm I, & t>0.
     \end{cases}
 \end{equation*}

\medskip

\noindent Throughout this work, we make the following assumptions on the parameters of the epidemic system \eqref{model}:
\begin{itemize}
\item[{\bf (A1)}] $L_{ii}=-\sum_{j\ne i}L_{ji}$ for $i=1,\cdots,n$,  $\mathcal{L}=(L_{ij})_{i,j=1}^n$ is quasipositive (i.e., $L_{ij}\ge 0$ for any $i\neq j$) and irreducible.
\item[{\bf (A2)}] $\bm\zeta, \bm\beta, \bm\gamma\gg \bm 0$, and $d_S, d_I>0$.
\end{itemize}

  Biologically, assumption {\bf(A1)}  means that the patches are fully connected, allowing individuals to move directly or indirectly between any two patches. Assumption {\bf (A2)} indicates that all members of the population have positive dispersal rates and that individuals can both contract and recover from the disease on any patch. 
 Due to biological interpretations of the vectors $\bm S$ and $\bm I$, we will only be interested in nonnegative solutions of \eqref{model}. Hence, the initial data of system \eqref{model} will always satisfy the standing assumption:
\begin{itemize}
\item[{\bf (A3)}]  $\bm S^0\ge\bm 0$, $ \bm I^0> \bm 0$. 
\end{itemize}

 Assumption {\bf (A3)}  implies that the initial total number of infected individuals is positive.  For any initial data $(\bm S(0), \bm I(0))=(\bm S^0, \bm I^0)\in\mathbb{R}_+^n\times \mathbb{R}_+^n$,  \eqref{model} has a unique nonnegative solution $(\bm S(t),\bm I(t))$ defined on a maximal interval of existence $[0,T_{\max})$.  Since $\bm\mu\ge\bm 0$, summing up all the equations in \eqref{model}, we find that 
\begin{equation}\label{Eq1:1}
\frac{d}{dt}\sum_{j\in\Omega} (S_j+I_j)=-\sum_{j\in\Omega}\mu_iI_j(t)\le 0 \quad 0<t<T_{\max},
\end{equation}which means that the total population is non-increasing. 
 Therefore, for any initial data $(\bm S^0, \bm I^0)$ satisfying {\bf (A3)}, the solution satisfies 
\begin{equation*}
 \sum_{j\in\Omega}( S_{j}(t)+ I_{j}(t))\le \sum_{j\in\Omega}(S_j^0+I_j^0),\quad \forall\ 0\le t<T_{\max}.
\end{equation*}
This means that the solution of  \eqref{model} exists globally and $T_{\max}=\infty$. Note that when $\bm\mu=\bm 0$, equality holds in \eqref{Eq1:1} for all $t\ge 0$.  It is easy to see that if $\bm I^0=\bm 0$ then $\bm I(t)=\bm 0$ for all $t\ge 0$. 
By {\bf (A1)}, $\mathcal{L}$ induces a  strongly positive matrix-semigroup $\{e^{t\mathcal{L}}\}_{t>0}$. Hence, if $(\bm S^0,\bm I^0)$ satisfies {\bf (A3)},  then $\bm S(t)\gg \bm 0$ and $\bm I(t)\gg \bm 0$ for all $t>0$. 

 For  $n\times n$ real-valued square matrix $M$, let $\sigma(M)$ be the set of eigenvalues of $M$,  $\sigma_{*}(M)$ be the spectral bound, i.e.,
\begin{equation*}
    \sigma_{*}(M):=\max\{\mathfrak{R}e(\lambda)\ :\ \lambda\in\sigma(M)\},
\end{equation*}
where $\mathfrak{R}e(\lambda)$ is the real part of $\lambda\in \mathbb{C}$, 
and $\rho(M)$ be the spectral radius, i.e.,
$$
\rho(M):=\max\{|\lambda|\ :\ \lambda\in\sigma(M)\}.
$$
Since $\mathcal{L}$ is quasi-positive and irreducible, it generates a strongly-positive matrix-semigroup $\{e^{t\mathcal{L}}\}_{t\ge 0}$ on $\mathbb{R}^n$.  Moreover, since $\sum_{i\in\Omega}{L}_{ij}=0$  for each $j\in\Omega$,  by the Perron-Frobenius theorem, $\sigma_*(\mathcal{L})=0$ is a simple eigenvalue of $\mathcal{L}$.  Furthermore, there is an  eigenvector $\bm\alpha$ associated with $\sigma_*(\mathcal{L})$ satisfying 
\begin{equation}\label{alpha-eq}
\mathcal{L}\bm\alpha=\bm 0,\quad \sum_{j\in\Omega}{\alpha}_j=1,\quad \text{and} \quad {\alpha}_j>0,\ \forall \ j\in\Omega,
\end{equation}  
and $\bm\alpha$ is the unique nonnegative eigenvalue of $\mathcal{L}$ with $\sum_{j\in\Omega}{\alpha}_j=1$. 

 An equilibrium solution  $(\bm S, \bm I)$ of  \eqref{model}  is a nonngative solution of the system of algebraic equations  
\begin{equation}\label{EE-system-1}
    \begin{cases}0=d_S\mathcal{L}\bm S +(\bm \gamma-\bm\beta\circ\bm S/(\bm\zeta+\bm S+\bm I))\circ\bm I\cr 
    0=d_I\mathcal{L}\bm I+(\bm\beta\circ\bm S/(\bm\zeta+\bm S+\bm I)-\bm\gamma)\circ \bm I-\bm \mu\circ\bm I.
    \end{cases}
\end{equation}
An equilibrium solution of system \eqref{model} of the form $(\bm S,\bm 0)$ is called a \textit{disease free equilibrium} (DFE). Since $\sigma_*(\mathcal{L})=0$ is a simple eigenvalue of $\mathcal{L}$, then $(\bm S,\bm 0)$ is a DFE of system \eqref{model} if and only if
\begin{equation}\label{DFE-eq1}
    \bm S=\|\bm S\|_1\bm\alpha.
\end{equation}
where $\bm\alpha$ is given by \eqref{alpha-eq}.

 Any equilibrium solution $(\bm S,\bm I)$ of \eqref{model} satisfying $\bm I>\bm 0$ and $\bm S>0$  will be called an \textit{endemic equilibrium} (EE)  solution. Since  {\bf (A1)} holds, then $\bm S\gg \bm 0$ and $\bm I\gg 0$ for any EE solution $(\bm S,\bm I)$ of \eqref{model}.  As we shall soon see from Theorem \ref{TH1} below, system \eqref{model} has no EE solution whenever $\bm\mu>\bm 0$.

\subsection{Main Results}
Next, we state our main results. To this end, we first consider the case of $\bm \mu>\bm 0$, and then discuss the case of $\bm\mu=\bm 0$.  {Throughout the paper, $\bm\alpha$ is fixed and satisfies \eqref{alpha-eq}.}


\subsubsection{Large-time behavior of solutions of system \eqref{model} when $\bm\mu> {\bm0}$.}

Our main result on system \eqref{model} when $\bm\mu>\bm0$ reads as follows.

\begin{tm}\label{TH1}  Suppose that {\bf (A1)-(A3)} holds.  Suppose also that $\bm\mu> {\bm 0}.$ Then, $\|\bm S^0+\bm I^0\|_1> \int_0^{\infty}\sum_{j\in\Omega}\mu_j I_j(t)dt$ and $(\bm S(t), \bm I(t))\to \big(\big(\|\bm S^0+\bm I^0\|_1-\int_0^{\infty}\sum_{j\in\Omega}\mu_j I_j(t)dt\big)\bm\alpha,{\bm 0}\big)$ as $t\to\infty$. 
    
\end{tm}

 When only disease induced death rate is taken into account by ignoring other factors that may impact population demographics, Theorem \ref{TH1} suggests that the disease will always be contained.  It would be of important biological interest to examine the global dynamics of solutions to system \eqref{model} by incorporating  population's natural birth and death rates. In general, it is a challenging task to establish an explicit formula for limit of solutions in Theorem \ref{TH1}. Nonetheless, explicit formulas may derived in some specific cases as detailed in the next remark.
 
 \begin{rk}Assume  $|\Omega|=1$ and $\mu>0$. 
 \begin{itemize}
 \item[\rm(i)] Assume in addition that $\zeta>0$. Then, the explicit formula for $\mu\int_0^{\infty}I(t)dt$,  hence for the limit of the susceptible population,  in terms of the initial data can be written  if $\beta=\mu+\gamma$  (see Theorem \ref{TH9}-{\rm (i)}). When $\beta\ge \mu+\gamma$, it always holds that $\int_0^{\infty}I(t)dt\to(S^0+I^0)/\mu$ as $\zeta\to0^+$ (see Theorem \ref{TH9}-{\rm(i)}).  
 
 \item[\rm (ii)] If  $\zeta=0$ in \eqref{model},   explicit formula of the unique solution of \eqref{model} is given by 
 \begin{equation}\label{solution-formula}
 S(t)=(S^0+I^0)Z^{\frac{\mu}{\beta}}(t)-I^0Z(t)e^{(\beta-\mu-\gamma)t} \quad \text{and} \quad I(t)=I^0Z(t)e^{(\beta-\mu-\gamma)t} \quad \forall\ t\ge 0,
 \end{equation} where   $Z(t)$ is given by \eqref{Z-formula}. As a consequence of \eqref{solution-formula},  Theorem \ref{TH9}-{\rm(ii)} below gives explicit formula for the limit of $(S(t),I(t))$ as $t$ tends to infinity in terms of the initial data.
 \end{itemize}
 \end{rk}

\subsubsection{Large-time behavior of solutions of system \eqref{model} when $\bm \mu=\bm0$.} 
Throughout this subsection, we assume that $\bm\mu=\bm0$. Thanks to the first equality in \eqref{Eq1:1}, for every positive number $N$ is fixed, the semiflow generated by  solutions of \eqref{model} leaves invariant the set 
$$
\mathcal{E}:=\Big\{(\bm S,\bm I)\in \mathbb{R}_+^n\times \mathbb{R}_+^n\ : \ \sum_{j\in\Omega}( S_j+ I_{j})= N\Big\}.
$$  
In this section, unless stated otherwise, we  fix $N>0$ and assume that our initial data is in the compact set \(\mathcal{E}\). First, thanks to \eqref{DFE-eq1}, $(N\bm\alpha,\bm 0)$ is the unique DFE of system \eqref{model} in $\mathcal{E}$. Note also from \eqref{EE-system-1} that an EE solution of system \eqref{model} in $\mathcal{E}$ is a positive solution of 
\begin{equation}\label{EE-system}
    \begin{cases}0=d_S\mathcal{L}\bm S +(\bm \gamma-\bm\beta\circ\bm S/(\bm\zeta+\bm S+\bm I))\circ\bm I\cr 
    0=d_I\mathcal{L}\bm I+(\bm\beta\circ\bm S/(\bm\zeta+\bm S+\bm I)-\bm\gamma)\circ \bm I,\cr
    N=\sum_{j\in\Omega}(S_j+I_j).
    \end{cases}
\end{equation}

\medskip

\noindent Linearizing system \eqref{model} at the DFE $(N\bm\alpha,\bm 0)\in\mathcal{E}$ when $\bm\mu=\bm0$ with respect to initial perturbations in $\mathcal{E}$ gives rise to the ODE-system 
\begin{equation}\label{linear-model}
\begin{cases}
    \frac{d\tilde{\bm S}}{dt}=d_S\mathcal{L}\tilde{\bm S} + \big(\bm\gamma-N\bm\beta\circ\bm\alpha/(\bm\zeta+N\bm\alpha) \big)\circ\tilde{\bm I} & t>0,\cr
    \frac{d\tilde{\bm I}}{dt}=d_I\mathcal{L}\tilde{\bm I}+\big(N\bm\beta\circ\bm\alpha/(\bm\zeta+N\bm\alpha) -\bm\gamma\big)\circ\tilde{\bm I} & t>0,\cr 
    0=\sum_{j\in\Omega}(\tilde{S}_j+ \tilde{I}_j).
    \end{cases}
\end{equation}
Hence, when $\bm \mu=\bm 0$, the stability of the null solution $\bm 0$ of system \eqref{linear-model} determines the local stability of the DFE $(N\bm\alpha,\bm 0)\in\mathcal{E}$ of system \eqref{model} with respect to initial perturbations in $\mathcal{E}$.  Now,  define 
 \begin{equation}\label{F-definition}
   {V}:={\rm diag}(\bm \gamma)-d_I\mathcal{L}.
 \end{equation} 
Note that  ${V}$ is invertible since $\mathcal{L}$ satisfies {\bf (A1)}, $\sigma_*(\mathcal{L})=0$ and $\bm \gamma>\bm 0$. Following the next generation matrix approach \cite{DW2002, Diekmann}, the BRN $\mathcal{R}_{0}$ of \eqref{model}  is 
 \begin{equation}\label{R-0}
      \mathcal{R}_{0}:=\rho({F}{V}^{-1})  \quad \text{where}\quad  {F}={\rm diag}(N\bm\alpha\circ \bm\beta/(\bm\zeta+N\bm\alpha)).
 \end{equation}
Note that $F$ depends on $N$ while $V$ depends on $d_I>0$. Hence, $\mathcal{R}_0$ depends both on $N>0$ and $d_I>0$, while it is independent of $d_S>0$. Thanks to \cite{allen2007asymptotic},  when $\bm\zeta =\bm0$ in \eqref{R-0}, we obtain the BRN $\hat{\mathcal{R}}_0$ of the epidemic model \eqref{standard-incidence}   
 \begin{equation}\label{R-hat-0}
  \hat{\mathcal{R}}_0=\rho(\hat{F}V^{-1})    \quad \text{where}\quad \hat{F}={\rm diag}(\bm\beta).
 \end{equation}
 Note also that $\hat{\mathcal{R}}_0$ depends on $d_I>0$, but is independent of $N>0$ and $d_S>0$. 

 Finally, when $\|\bm\beta/\bm\gamma\|_{\infty}>1$, or equivalently the set $\tilde{\Omega}:=\{j\in\Omega : \beta_j>\gamma_j\}$ is nonempty, we introduce the positive quantities 
 \begin{equation}\label{N-star-def}
     \mathcal{N}_{\rm low}^*=\min_{j\in\tilde{\Omega}}\frac{\gamma_j\zeta_j}{(\beta_j-\gamma_j)}\quad \text{and}\quad \mathcal{N}^*_{\rm up}=\min_{j\in\tilde{\Omega}}\frac{\gamma_j\zeta_j}{(\beta_j-\gamma_j)\alpha_j}.
 \end{equation}
 It is easy to see that $ \mathcal{N}^*_{\rm low}\le \mathcal{N}^*_{\rm up}$, with strict inequality if $|\Omega|\ge 2$. As shall be shown below (see Remark \ref{RK1}-{\rm (iv)}), the quantities $\mathcal{N}^*_{\rm up}$ and $\mathcal{N}^*_{\rm low}$ serve as important threshold numbers for the total size $N$ of the population when $\|\bm\beta/\bm\gamma\|_{\infty}>1$. The following result collects some important properties of $\mathcal{R}_0$.

 \begin{prop}\label{prop2} Let $\mathcal{R}_0$ and $\hat{\mathcal{R}}_0$ be defined by \eqref{R-0} and \eqref{R-hat-0}, respectively.
 \begin{itemize}
     \item[\rm (i)] $\mathcal{R}_{0}-1$ and $\sigma_*(F-{V})=\sigma_*(d_I\mathcal{L}+{\rm diag}(N\bm\alpha\circ\bm\beta/(\bm\zeta+N\bm\alpha)-\bm\gamma))$ have the same sign.

     \item[\rm (ii)] If $\bm\alpha\circ\bm\beta/(\bm\zeta+N\bm\alpha)=m\bm\gamma$ for some $m>0$, then $\mathcal{R}_{0}=mN$ for all $d_I>0$. However, if $\bm\alpha\circ\bm\beta/(\bm\zeta+N\bm\alpha)\notin{\rm span}(\bm\gamma)$, then $\mathcal{R}_{0}$ is  strictly decreasing in $d_I$, 
     \begin{equation}\label{R-0-limit}
\lim_{d_I\to 0^+}\mathcal{R}_0= \max_{i\in\Omega}\frac{N\alpha_i\beta_i}{\gamma_i(\zeta_i+N\alpha_i)}\qquad  \text{and}\quad \ \lim_{d_I\to \infty}\mathcal{R}_0=\frac{\sum_{i\in\Omega}\frac{N\beta_i\alpha_i^2}{(\zeta_i+N\alpha_i)}}{\sum_{i\in\Omega}\alpha_i\gamma_i} .
\end{equation}
\item[\rm(iii)] Fix $d_I>0$. Then $\mathcal{R}_0$ is strictly increasing in $N>0$, 
\begin{equation}\label{R-0-limit-in-N}
    \lim_{N\to 0^+}\mathcal{R}_0=0 \qquad \text{and}\qquad \lim_{N\to\infty}\mathcal{R}_0=\hat{\mathcal{R}}_0.
\end{equation}
Hence, if $\hat{\mathcal{R}}_0\le 1$, then $\mathcal{R}_0<1$ for all $N>0$. However, if $\hat{\mathcal{R}}_0>1$, then there is a unique $\mathcal{N}_0=\mathcal{N}_0(d_I,{\bm\zeta})>0$ such that $\mathcal{R}_0<1$ if $0<N<\mathcal{N}_0$; $\mathcal{R}_0=1$ if $N=\mathcal{N}_0$; and $\mathcal{R}_0>1$ if $N>\mathcal{N}_0$.

\item[\rm (iv)] If $\|\bm\beta/\bm \gamma\|_{\infty}>1$, then there is $d_*\in(0,\infty]$, independent of $\bm\zeta$, such that  $\mathcal{N}_0(d_I,\bm\zeta)$ is defined if and only if $0<d_I<d_*$. In addition, the following conclusions hold.
\begin{itemize}
    \item[\rm(iv-1)] {Fix $\bm\zeta\gg \bm 0$.} If  $ (N^*\bm \alpha\circ\bm \beta)/(\bm\zeta+N^*\bm\alpha)=\bm\gamma$ for some $N^*>0$, then $d_*=\infty$ and  $\mathcal{N}_0=N^*$ for all $ d_I>0$.
    
    \item[\rm (iv-2)] {Fix $\bm\zeta\gg \bm 0$.} If $ (N\bm \alpha\circ\bm \beta)/(\bm\zeta+N\bm\alpha)\ne\bm\gamma$ for all $N>0$, then $\mathcal{N}_0$ is strictly increasing in $d_I$ and $\mathcal{N}_0(d_I)\to\mathcal{N}^*_{\rm up}$ as $d_I\to 0^+$, where $\mathcal{N}^*_{\rm up}$ is defined by \eqref{N-star-def}.
    
    \item[\rm(iv-3)] {Fix $0<d_I<d_*$. $\mathcal{N}_0$ is strictly increasing in $\bm \zeta\gg\bm0$ and $\mathcal{N}_0(d_I,\tau\bm\zeta)=\tau\mathcal{N}_0(d_I,\bm \zeta)$ for all $\tau>0$ and $\bm\zeta\gg 0$. In particular, $\bm\zeta_m\mathcal{N}_0(d_I,\bm 1)\le\mathcal{N}_0(d_I,\bm\zeta)\le \bm\zeta_M\mathcal{N}_0(d_I,\bm 1)$ for all $\bm\zeta\gg \bm0$. Hence, $\mathcal{N}_0(d_I,\bm\zeta)\to 0$ as $\bm\zeta\to \bm 0^{+}$ and $\mathcal{N}_0(d_I,\bm\zeta)\to \infty$ as $\bm\zeta_m\to\infty$. }
\end{itemize}

 \end{itemize}
     
 \end{prop}
 
\begin{rk}\label{Rk0} Assume that $\|\bm\beta/\bm\gamma\|_{\infty}>1 $ and let $\mathcal{N}_0$ be given by Proposition \ref{prop2}-{\rm (iii)}. Then, by Proposition \ref{prop2}-{\rm (iv-1)-(iv-2)}, $\mathcal{N}_0$ is constant in $d_I\in(0,d_*)$ if and only if $(\bm\beta/\bm\gamma)_m>1$ and $\bm\zeta\circ\bm\gamma/((\bm\beta-\bm\gamma)\circ\bm\alpha)\in {\rm span}(\bm 1)$. It also follows from Proposition \ref{prop2}-{\rm (iv-3)} that large saturation incidence is helpful to lower the BRN.
    
\end{rk}
  
\noindent The next result concerns the local stability of the DFE and the existence of EE solution of system \eqref{model}. 

  \begin{tm}\label{TH1-2}  Suppose that {\bf (A1)-(A3)} holds and $\bm \mu=\bm 0$. Then the following conclusions hold.
  \begin{itemize}
      \item[\rm (i)] If $\mathcal{R}_0<1$, then the DFE is locally asymptotically stable in $\mathcal{E}$.

      \item[\rm (ii)] If $\mathcal{R}_0>1$, then the disease is uniformly persistent in the sense that there is $m_*>0$ such that 
      \begin{equation}\label{TH1-2-eq1}
          \liminf_{t\to\infty}\min_{j\in\Omega}I_j(t)\ge m_*
      \end{equation}
      for any solution $(\bm S(t), \bm I(t))$ of system \eqref{model} with initial data $(\bm S^{0},\bm I^{0})\in\mathcal{E}$ satisfying $\bm I^{0}>\bm 0$. Furthermore, system \eqref{model} has at least one EE solution.
  \end{itemize}
      
  \end{tm}

  \noindent It is evident that Theorem \ref{TH1-2}-{\rm (i)} follows directly from Proposition \ref{prop2}-{\rm (i)}. Additionally, Theorem \ref{TH1-2}-{\rm (ii)} can be established with some modifications to the proof provided in \cite[Theorem 2.3]{wang2004epidemic}. According to Proposition \ref{prop2}-{\rm (iii)} and Theorem \ref{TH1-2}-{\rm (i)}, if the set $\tilde{\Omega}$ is empty, a mild outbreak of the disease will ultimately be controlled regardless of the dispersal rates  of the population and the population size. However, when $\tilde{\Omega}$ is not empty, Theorem \ref{TH1-2} suggests that the disease is more likely to persist if the dispersal rate of the infected population is low and the total population size exceeds the critical number $\mathcal{N}_{\rm up}^*$. In the following three results, we will determine sufficient conditions on the model parameters that ensure all solutions will eventually stabilize. 

 Our next result concerns the global stability of the DFE. To state this result, we first 
 define
  \begin{equation}  \label{tilde-R-0}
   \tilde{F}={\rm diag}(N\circ\bm\beta/(\bm\zeta+N\bm 1))\quad \text{and}\quad \tilde{\mathcal{R}}_0=\rho(\tilde{F}V^{-1})
  \end{equation} 
  where $V$ is defined as in \eqref{F-definition}. It is clear that $ \mathcal{R}_0\le \tilde{\mathcal{R}}_0$, where the equality holds if and only if $|\Omega|=1$. In the latter scenario, that is $|\Omega|=1$,  we have that $\tilde{\mathcal{R}}_0=\mathcal{R}_0=\frac{N\beta}{(\zeta+N)\gamma}$. 
The next result asserts the global stability of the DFE in $\mathcal{E}$ whenever $\tilde{\mathcal{R}}_0\le 1$ and reads as follows.

\begin{tm}\label{TH2}  Suppose that {\bf (A1)-(A3)} holds. Assume also that $\bm\mu=\bm0$ and $\tilde{\mathcal{R}}_0\le1$ where $\tilde{\mathcal{R}}_0$ is defined by \eqref{tilde-R-0}. Then $(\bm S(t),\bm I(t))\to (N\bm\alpha,\bm 0)$ as $t\to\infty$ for any solution of \eqref{model} with initial $(\bm S^0,\bm I^0)\in\mathcal{E}$.  
    
\end{tm}

\begin{rk}\label{RK1}
\begin{itemize}
    \item[\rm (i)] Thanks to Theorems \ref{TH1-2} and \ref{TH2}, for the single-patch model \eqref{model} with $\bm \mu=0$, the DFE is globally stable if  $\mathcal{R}_0\le 1$, while the disease becomes endemic and there is at least one EE solution if $\mathcal{R}_0>1$. The uniqueness and global stability of the EE solution in this case will be established in  Theorem \ref{TH3} below.
    
    \item[\rm (ii)]Similarly to $\mathcal{R}_0$ as in Proposition \ref{prop2}, for every $N>0$,  we have that $\tilde{\mathcal{R}}_0$ is nonincreasing in $d_I$, 
\begin{equation}\label{RK1-eq1}
    \lim_{d_I\to0^+}\tilde{\mathcal{R}}_0=\max_{i\in \Omega}\frac{N\beta_i}{\gamma_i(\zeta_i+N)}\qquad \text{and}\qquad \lim_{d_I\to\infty}\tilde{\mathcal{R}}_0=\frac{\sum_{i\in\Omega}\frac{N\beta_i\alpha_i}{\zeta_i+N}}{\sum_{i\in\Omega}\alpha_i\gamma_i}.
\end{equation}


Moreover, for every $d_I>0$, it also holds that $\tilde{\mathcal{R}}_0$ is strictly increasing in $N>0$, $\tilde{\mathcal{R}}_0\to 0$ as $N\to 0^+$, and $ \tilde{\mathcal{R}}_0\to \hat{\mathcal{R}}_0$ as $N\to \infty$ where $\hat{\mathcal{R}}_0$ is defined by \eqref{R-hat-0}. In particular, $\tilde{\mathcal{R}}_0<\hat{\mathcal{R}}_0$ for all $N>0$.

\item[\rm (iii)] Assume $\bm\mu=\bm 0$. If
 $\hat{\mathcal{R}}_0\le 1$, it follows from Theorem \ref{TH2} that the disease will  be eventually eradicated for any total population size $N>0$ and dispersal rate $d_S$ of the susceptible population. Hence, observing that $\hat{\mathcal{R}}_0\le \|\bm \beta/\bm\gamma\|_{\infty}$, then if $\|\bm \beta/\bm\gamma\|_{\infty}\le 1$,  the disease will be eventually eradicated for any total population size $N>0$ and population dispersal rates $d_I$ and $d_S$. {As a consequence of Proposition \ref{prop3}-{\rm (ii)} and Remark \ref{Rk3} below, there are some range of the parameters of system \eqref{model} satisfying $\mathcal{R}_0<1<\tilde{\mathcal{R}}_0$ such that the DFE is not globally stable. }
 
 \item[\rm(iv)] Assume that $\bm\mu = \bm{0}$ and  $\|\bm\beta/\bm\gamma\|_{\infty} > 1$, let $\mathcal{N}^*_{\rm low}$  be defined by \eqref{N-star-def}. 
 Thanks to \eqref{RK1-eq1} and Proposition \ref{prop2}-{\rm (iv-2)}, if   $0 < N \le \mathcal{N}^*_{\rm low}$, the DFE is globally stable for the system \eqref{model} regardless of the population dispersal rates. A stronger version will be established in Theorem \ref{TH10} below. 

\end{itemize}
    
\end{rk} 

\begin{figure}
    \centering
    \includegraphics[width=0.8\textwidth,height=0.2\textheight]{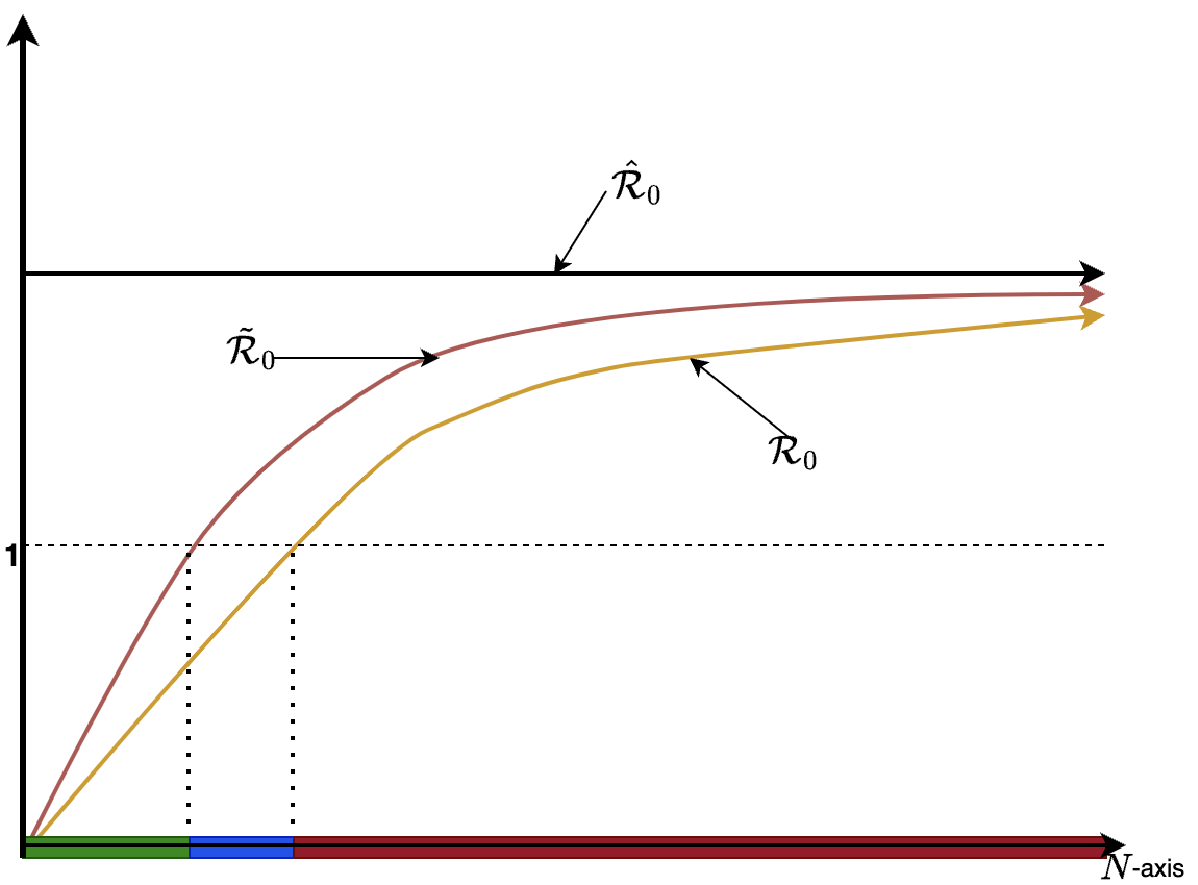}
    \caption{Schematic illustration of the curves $\mathcal{R}_0$, $\tilde{\mathcal{R}}_0$, and $\hat{\mathcal{R}}_0$ with respect to $N>0$ for a fixed value of $d_I$ and $|\Omega|\ge 2$. {Assume $\bm\mu=\bm 0$. {\bf (a)} When $N$ falls in the green interval, then the DFE is both locally asymptotically stable and globally stable for system \eqref{model}. {\bf (b)} When $N$ lies in the blue interval, then the DFE is locally asymptotically stable and possibly not globally stable. Moreover, system \eqref{model} may have at least two EEs. {\bf (c)} When $N$ lies in the red interval, then the DFE is unstable and system \eqref{model} has at least one EE solution. 
    }
    }
    \label{Fig1_1}
\end{figure}

\begin{tm}\label{TH10} Suppose that $\|\bm \beta/\bm\gamma\|_{\infty}>1$ and {\bf (A1)-(A3)} hold. Assume also that $\bm \mu=\bm 0$ and $0<N\le \mathcal{N}^*_{\rm up}$, where $\mathcal{N}^*_{\rm up}$ is defined by \eqref{N-star-def}. Then $(\bm S(t), \bm I(t))\to (N\bm\alpha,\bm 0)$ as $t\to\infty$ for any solution of system \eqref{model} with initial $(\bm S^0,\bm I^0)\in\mathcal{E}$.
    
\end{tm}

 Thanks to Theorem \ref{TH10}, if the total population size is below the threshold number $\mathcal{N}^*_{\rm up}$, then the disease will be eradicated irrespective of the population dispersal rates. We point out that this is a sharp result since if $N>\mathcal{N}^*_{\rm up}$, then $\mathcal{R}_0>1$ for small dispersal rate $d_I$ of infected population.

 Due to the multiple patches and the patch-heterogeneity of the parameters of system \eqref{model}, it is a challenging question to investigate the uniqueness and/or stability of EE solution when $\mathcal{R}_0>1$. In the next two results, we identify some practical scenarios where solutions of \eqref{model} with positive initial data eventually stabilize at the EE whenever $\mathcal{R}_0>1$.  For convenience, we define $\bm r=(r_1, \dots, r_n)^T\in\mathbb{R}^n$ with  
 \begin{equation*}
    { r}_{j}
    :=\frac{ \gamma_{j}}{ \beta_{j}}, \quad \forall\; j\in\Omega.
\end{equation*}
When $\bm r$ is patch-homogeneous, that is the infection and recovery rates are proportional, our next result asserts the global dynamics of solutions of system \eqref{model} under some additional hypothesis. 

\begin{tm} \label{TH3} Suppose that {\bf (A1)-(A3)} holds and $\bm \mu=\bm 0$. Suppose also that $\bm r\in{\rm span}(\bm 1)$ and $\bm \zeta\in{\rm span}(\bm \alpha)$. 
\begin{itemize}
    \item[\rm (i)] If $\mathcal{R}_0\le 1$, then the DFE is globally stable with respect to perturbations with initial data in $\mathcal{E}$.

    \item[\rm (ii)] If $\mathcal{R}_0>1$, then system \eqref{model} has a unique EE solution ${\bm E}^*$ in $\mathcal{E}$. Moreover, ${\bm E}^*$ is globally stable with respect to perturbations with initial data in $\mathcal{E}$.
\end{itemize}
    
\end{tm}

\begin{rk}\label{RK2} Suppose that the hypotheses of Theorem \ref{TH3} hold. Hence, there exist $\tau>0$ and $m>0$ such that $\bm r=\tau\bm 1$ and $\bm\zeta =m\bm\alpha$. As a result, we have from \eqref{R-0-limit} that $\mathcal{R}_0=\frac{N}{\tau(m+N)}$ is independent of the population dispersal rates. Clearly $\mathcal{R}_0$ is strictly increasing in $N$ and strictly decreasing in $\tau$. Furthermore, thanks to Theorem \ref{TH3}, the followings hold.
\begin{itemize}
    \item[\rm (i)] If $\tau\ge 1$, we always have that $\mathcal{R}_0<1$ and the disease will be eventually eradicated.
    \item[\rm (ii)] If $0<\tau<1$ and  $N\le \frac{\tau }{(1-\tau)}m$, then $\mathcal{R}_0\le 1$  and  the disease eventually goes extinct.
    \item[\rm (iii)] If $0<\tau<1$ and $N>\frac{\tau}{1-\tau}m$, then $\mathcal{R}_0>1$, the disease is endemic, and positive solutions of \eqref{model} with initial data in $\mathcal{E}$ eventually stabilize at the unique EE solution given by  ${\bm E}^*:=(\tau(N+m)\bm \alpha,((1-\tau)N-\tau m)\bm \alpha)$. Noting that ${\bm E}^*\to (0,N\bm\alpha)$ as $\tau\to 0^+$, for every $N>0$,  then high disease infection rate significantly decreases the total size of the susceptible population at the EE solutions. 
\end{itemize} 
    
\end{rk}

\noindent Theorem \ref{TH3} shows that for single patch model, every solution eventually stabilizes at an equilibrium solution. Moreover, the EE is unique and globally stable if $\mathcal{R}_0>1$ under the hypothesis of the theorem. In particular, the global dynamics of solutions of \eqref{model} is well understood for the single patch model. In the remainder of this section, we shall always suppose that $|\Omega|\ge 2$, that is there are at least two patches in the network epidemic model \eqref{model}. {As shall be seen  from Proposition \ref{prop3}-{\rm(ii)} and Remark \ref{Rk3} below, when $\bm\mu=\bm 0$, $\bm\zeta\in{\rm span}(\bm \alpha)$, and $\bm r\notin {\rm span}(\bm 1)$, the conclusions of Theorem \ref{TH3}-{\rm(i)} may fail.} When the population disperses uniformly, the next result asserts  the global dynamics of solutions of \eqref{model}, and shows that solutions always eventually stabilize at an equilibrium solution.

\begin{tm}\label{TH4} Suppose that {\bf (A1)-(A3)} hold and $\bm \mu=\bm 0$. Suppose also that $d_S=d_I$. Then the following conclusions hold.
\begin{itemize}
    \item[\rm (i)] If $\mathcal{R}_0\le 1$, then the DFE is globally stable with respect to perturbations with initial data in $\mathcal{E}$.

    \item[\rm (ii)] If $\mathcal{R}_0>1$, then system \eqref{model} has a unique EE solution $(\bm S^*, \bm I^*)$ in $\mathcal{E}$. Furthermore, $(\bm S^*,\bm I^*)$ is globally stable with respect to perturbations with initial data in $\mathcal{E}$.
\end{itemize}
    
\end{tm}

\noindent Theorems \ref{TH3} and \ref{TH4} provides sufficient conditions on the model parameters under which the EE is unique and globally stable whenever it exists. The next result examines the uniqueness of the  EE solution  under some  hypotheses.

\begin{tm}\label{TH5}  Fix $N>0$, $d_S>0$ and $d_I>0$. Assume that $\bm \mu=\bm0$, {\bf (A1)-(A2)} hold, and $\hat{\mathcal{R}}_0>1$. 
     \item[\rm (i)]    If   $d_S\ge d_I$, then system \eqref{model} has no EE solution in $\mathcal{E}$ if $ \mathcal{R}_0\le 1$, and has a unique EE in $\mathcal{E}$ if $ \mathcal{R}_0>1$.

     \item[\rm (ii)] If   $N(\bm 1-2\bm r)\circ\bm\alpha\ge \bm r\circ\bm \zeta$, then $\mathcal{R}_0>1$ and    system \eqref{model}  has a unique EE solution in $\mathcal{E}$.
    
\end{tm}

\noindent We complement Theorem \ref{TH5} with the following result on the global structure of the set of EE solutions of system \eqref{model} as $\mathcal{R}_0$ varies.

\begin{tm}\label{TH6}  Fix $d_I>0$, $d_S>0$, and suppose that $\hat{\mathcal{R}}_0>1$. Then there is $0<\mathcal{R}_{\min}\le 1$ such that system \eqref{model} has no EE for $\mathcal{R}_0<\mathcal{R}_{\min}$ and at least one EE solution if $\mathcal{R}_{\min}<\mathcal{R}_0<\hat{\mathcal{R}}_0$. Moreover,  as $\mathcal{R}_0$ varies from $\mathcal{R}_{\min}$ to $\hat{\mathcal{R}}_0$, the set of EE solutions of \eqref{model} forms a simple  curve $\mathcal{C}_*:=\{(\mathcal{R}_0,\bm S,\bm I)=(f(l),\bm S(\cdot;l),\bm I(\cdot;l)) : l>\mathcal{N}_0\}$, where $\mathcal{N}_0$ is as in Proposition \eqref{prop2}-{\rm (iii)}. $(f(l),\bm S(\cdot; l), \bm I(\cdot;l))$ is analytic function of $l>\mathcal{N}_0$ and satisfies $\bm 0\ll \bm I(\cdot;l)\ll \bm I(\cdot;\tilde{l})$ for all $\tilde{l}>l>\mathcal{N}_0$,
\begin{equation}\label{Th6-eq1}
    \lim_{l\to \mathcal{N}_0^+}(f(l),\bm S(\cdot;l),\bm I(\cdot;l))=(1,\mathcal{N}_0\bm\alpha, \bm 0), \, \lim_{l\to\infty}f(l)=\hat{\mathcal{R}}_0, \, \lim_{l\to\infty}\sum_{j\in\Omega}S_j(\cdot;l)=\infty, \, \text{and}\, \lim_{l\to\infty}\sum_{j\in\Omega}I_j(\cdot;l)=\infty.
\end{equation} 
Furthermore,  
\begin{itemize}
    \item[\rm (i)] 
${\mathcal{R}}_0=1$ is a forward  transcritical bifurcation point if 
\begin{equation}\label{TH6-eq1}
    d_I\sum_{j\in\Omega}\frac{\zeta_j\eta_j\eta_j^*\beta_j(\alpha_j-\eta_j)}{(\zeta_j+\mathcal{N}_0\alpha_j)^2}<d_S\sum_{j\in\Omega}\frac{\beta_j\eta_j\eta_j^*\alpha_j(\mathcal{N}_0\eta_j+\zeta_j)}{(\zeta_j+\mathcal{N}_0\alpha_j)^2};
\end{equation}

\item[\rm (ii)] ${\mathcal{R}}_0=1$ is a backward  transcritical bifurcation point if 
\begin{equation}\label{TH6-eq2}
    d_I\sum_{j\in\Omega}\frac{\zeta_j\eta_j\eta_j^*\beta_j(\alpha_j-\eta_j)}{(\zeta_j+\mathcal{N}_0\alpha_j)^2}>d_S\sum_{j\in\Omega}\frac{\beta_j\eta_j\eta_j^*\alpha_j(\mathcal{N}_0\eta_j+\zeta_j)}{(\zeta_j+\mathcal{N}_0\alpha_j)^2},
\end{equation}
\end{itemize}
where $\bm \eta\gg \bm 0$ and $\bm\eta^*\gg\bm 0$ are the right and left positive eigenvectors associated with $\sigma_*\big(d_I\mathcal{L}+{\rm diag}\big(\big(\mathcal{N}_0\bm\beta\circ\bm\alpha/(\bm \zeta+\mathcal{N}_0\bm\alpha )-\bm\gamma\big)\big)\big)$  satisfying $\|\bm\eta\|_1=\|\bm\eta^*\|_1=1$, respectively.
In particular,  $\mathcal{R}_0=1$ is  always  a forward transcritical bifurcation point if 
$d_S\ge d_I$.



\begin{figure}
    \centering
    \includegraphics[width=01\textwidth,height=0.2\textheight]{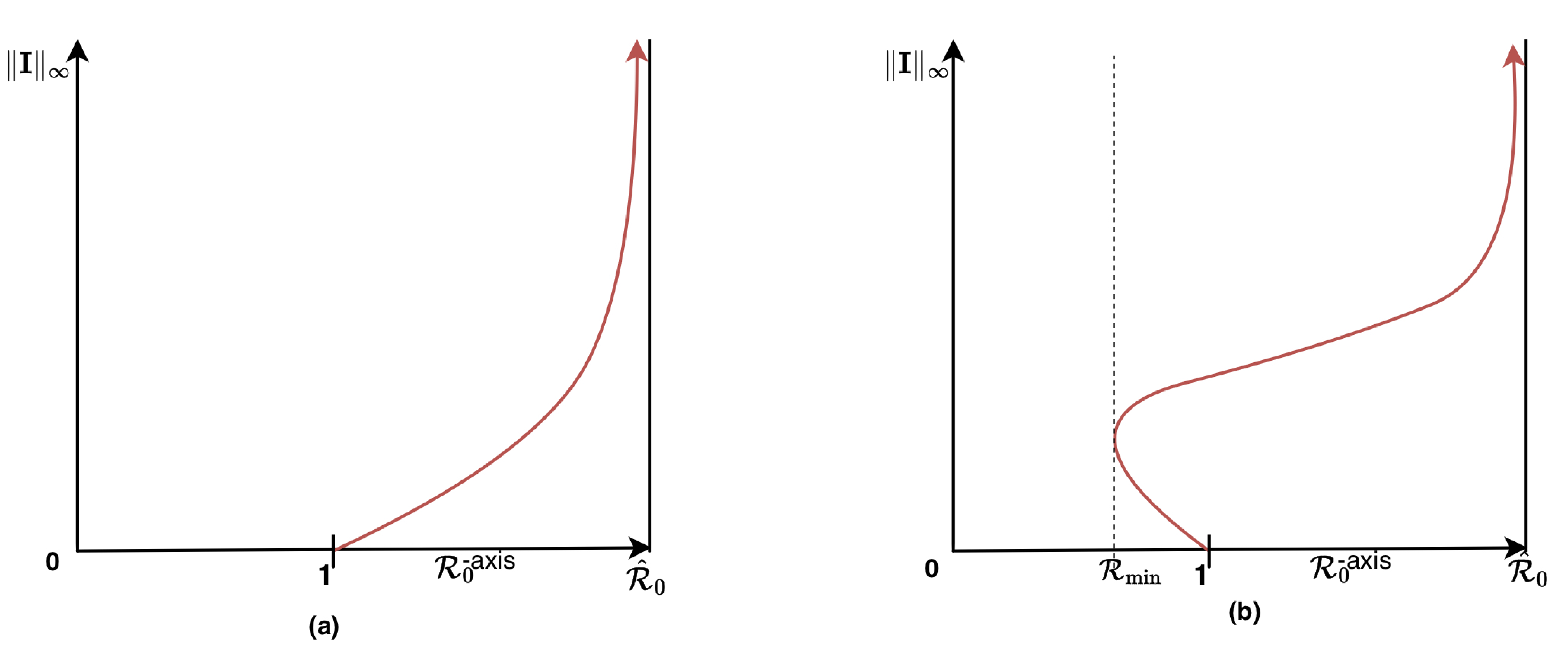}
    \caption{Schematic illustration of   bifurcation curve of $\|\bm I\|_{\infty}$ at EEs as $\mathcal{R}_0$ varies from $0<\mathcal{R}_{\min}\le 1$ to $\hat{\mathcal{R}}_0$. Figure {\bf (a)} corresponds to the case of fixed $d_S\ge d_I>0$  as described in Theorem \ref{TH6}.  Figure {\bf (b)} corresponds to the case of fixed $d_I>0$ and $0<d_S<d_{\rm up}^*$   as described in Theorem \ref{TH6} and Proposition \ref{prop3}-{\rm(ii)}.
    }
    \label{fig2}
\end{figure}

\end{tm}


\noindent An immediate consequence of Theorem \ref{TH6} is that when all parameters of the system \eqref{model} are fixed, there is at most a finite number of EE solutions. Moreover, the $\bm I$-components at the EE solutions can be totally ordered. Theorem \ref{TH6} also shows that when all the parameters are fixed, but the total population size and hence $\mathcal{R}_0$ varies, then $\mathcal{R}_0=1$ is a forward transcritical bifurcation point for the set of EE solutions if 
the susceptible population disperses faster than the infected population. 
{Our next result identifies sufficient conditions on the parameters of the system \eqref{model} under which a backward bifurcation occurs at $\mathcal{R}_0=1$.

\medskip

\begin{prop}\label{prop3} Suppose that $|\Omega|=2$ and $\mathcal{L}$ is symmetric. Suppose also that $\bm\zeta\in{\rm span}(\bm 1)$, $\bm \zeta\gg \bm 0$, $\bm r\notin{\rm span}(\bm 1)$, $\gamma_1<\beta_1$, $\|\bm\gamma\|_1<\|\bm\beta\|_1$, and $\mathcal{N}_{\rm up}^*=\frac{\gamma_1\zeta_1}{(\beta_1-\gamma_1)\alpha_1}$, where $\mathcal{N}_{\rm up}^*$ is defined by \eqref{N-star-def}. Then $d_*=\infty$ in Proposition \ref{prop2}-{\rm (iv)}. Moreover, for every $d_I>0$, the following conclusions hold.

\begin{itemize}
    \item[\rm (i)] If  $\eta_2\sqrt{\beta_2}\le\eta_1\sqrt{\beta_1}$,  then $\mathcal{R}_0=1$ is a forward transcritical bifurcation point for any $d_S>0$. 
    \item[\rm (ii)] If $\eta_2\sqrt{\beta_2}>\eta_1\sqrt{\beta_1}$,  then there is $0<d_{\rm up}^*<d_I$ such that $\mathcal{R}_0=1$ is a backward transcritical bifurcation point for every $0<d_S<d^*_{\rm up}$, while it is a forward transcritical bifurcation point for every $d_S>d^*_{\rm up}$. 
\end{itemize}   
    Furthermore, it holds that 
    \begin{equation}\label{prop3-eq1}
        \Big(1-\frac{(\beta_1-\gamma_1)}{d_IL}\Big)_+<\frac{\eta_2}{\eta_1}<1 \quad \forall\ d_I>0,
    \end{equation}
    where $L:=L_{12}=L_{21}>0$ and $\bm\eta$ is the positive eigenvector associated with $\sigma_*\big(d_I\mathcal{L}+{\rm diag}\big(\big(\mathcal{N}_0\bm\beta\circ\bm\alpha/(\bm \zeta+\mathcal{N}_0\bm\alpha )-\bm\gamma\big)\big)\big)$  satisfying $\|\bm\eta\|_1=1$.
\end{prop}
\begin{rk}\label{Rk3} Assume  the hypotheses of Proposition \ref{prop3} and set  $L:=L_{12}=L_{21}>0$. In addition, if  $\beta_1<\beta_2$,   then thanks to \eqref{prop3-eq1} and Proposition \ref{prop3}-{\rm (ii)}, for every $d_I>d_0^*:=\frac{\beta_1-\gamma_1}{L(1-\sqrt{\beta_1}/\sqrt{\beta_2})}$,  there is $d^*_{\rm up}=d^*_{\rm up}(d_I)>0$ such that $\mathcal{R}_0=1$ is a transcritical backward bifurcation point for every $0<d_S<d^*_{\rm up}$. It then follows from Theorem \ref{TH6} that for every fixed $d_I>d_0^*$ and $0<d_S<d^*_{\rm up}$, the epidemic model \eqref{model} has at least two EE solutions 
for some range of the value $N>0$ corresponding to $\mathcal{R}_0<1$. This is strongly in contrast with the dynamics of solutions of \eqref{standard-incidence}, since the latter has no EE solution when its BRN $\hat{\mathcal{R}}_0$ is less than or equal to one. {Table \ref{Ex8-T} gives numerical simulations for the existence of EEs when $\mathcal{R}_0<1$ under the hypotheses of Proposition \ref{prop3}.}
\end{rk}
}

\subsubsection{Asymptotic profiles of EEs of system \eqref{model} when $\bm\mu=\bm 0$.}

We investigate the profiles of the EE solutions as either $d_S$ or $d_I$ becomes significantly small. Our first result concerns the case of $d_S$ tending to zero while $d_I>0$ is fixed. In the subsequent results, recall that $\bm r=\bm\gamma/\bm\beta$.

\begin{tm}\label{TH7}  Suppose that $\bm\mu=\bm0$. Fix $ d_I>0$ and $N>0$, and suppose that $\mathcal{R}_0>1$. For every $d_S>0$, let $(\bm S,\bm I)$ be an EE solution of \eqref{model} in $\mathcal{E}$.  Then $\bm I-(\sum_{j\in\Omega}I_j)\bm\alpha\to \bm 0$ as $d_S\to 0^+$. Furthermore, the following conclusions hold.
\begin{itemize}
    \item[\rm (i)] If either $\bm r_{M}\ge 1$ or $\bm r_M<1$ and $N\le \|\bm\zeta\circ {\bm r}/(\bm 1-{\bm r})\|_1$, then $\|\bm I\|_1\to 0$ and $\|\bm S\|_1\to N$  as $d_S\to 0^+$. 
    \begin{itemize}
        \item[\rm (i-1)] If either $\bm r_M\ge 1$ or $\bm r_M<1$ and $N<\|\bm\zeta\circ\bm r/(\bm 1-\bm r)\|_1$, then, up to a subsequence, as $d_S$ tends to zero, $(\bm S,\frac{1}{d_S}\bm I)\to (l^*(\bm\alpha-d_I\bm P^*),l^*\bm P^*)$ where $l^*>\mathcal{N}_0$  and  $\bm 0\ll \bm P^*\ll \frac{1}{d_I}\bm\alpha$ satisfy
    \begin{equation}\label{Eq1-TH7}
       \begin{cases} 0=d_I\mathcal{L}\bm P^* +\bm\beta\circ(l^*(\bm\alpha-d_I\bm P^*)/(\bm\zeta +l^*(\bm\alpha-d_I\bm P^*))-\bm r)\circ\bm P^*\cr 
       N=l^*\sum_{j\in\Omega}(\alpha_j-d_IP^*_j).
       \end{cases}
    \end{equation}
    Here $\mathcal{N}_0$ is given by Proposition \ref{prop2}-{\rm (iii)}.
    \item[\rm (i-2)] If $\bm r_M<1$ and $ N=\|\bm\zeta\circ\bm r/(\bm 1-\bm r)\|_1$, then, up to a subsequence, as $d_S\to 0^{+}$, either $(\bm S,\frac{1}{d_S}\bm I)$ has the asymptotic profiles described in {\rm(i-1)}, or $\bm S\to \bm\zeta\circ\bm r/(\bm 1-\bm r)$.
    \end{itemize}

    \item[\rm (ii)] If ${\bm r}_M<1$ and $N>\|\bm\zeta\circ {\bm r}/(\bm 1-{\bm r})\|_1$, then, up to a subsequence, as $d_S\to 0^+$, one of the following holds. 
    \begin{itemize}
    \item[\rm (ii-1)] $(\bm S,\bm I)\to (\bm S^*,\bm I^*)$ where 
    \begin{equation}\label{Eq2-TH7}
\bm S^*:= \Big(\bm\zeta+\frac{(N-\|\bm\zeta\circ\bm r/(\bm 1-\bm r)\|_1)}{(1+\|\bm\alpha\circ\bm r/(\bm 1-\bm r)\|_1)}\bm\alpha\Big)\circ(\bm r/(\bm 1-\bm r))\quad  
\text{and} 
\quad 
\bm I^*:= \frac{(N-\|\bm\zeta\circ\bm r/(\bm 1-\bm r)\|_1)}{(1+\|\bm\alpha\circ\bm r/(\bm 1-\bm r)\|_1)}\bm\alpha.
\end{equation}
 
 \item[\rm (ii-2)]   $(\bm S,\frac{1}{d_S}\bm I)\to (l^*(\bm\alpha-d_I\bm P^*),l^*\bm P^*)$ where $l^*>\mathcal{N}_0$ and $\bm0\ll \bm P^*\ll \frac{1}{d_I}\bm\alpha$ solve \eqref{Eq1-TH7}. 
 \end{itemize}
 Furthermore, {\rm (ii-1)} always holds {if either $N> \|\bm \zeta\circ\bm r/((\bm 1-\bm r )\circ\bm\alpha)\|_{\infty}$ or $N= \|\bm \zeta\circ\bm r/((\bm 1-\bm r )\circ\bm\alpha)\|_{\infty} $ and $\bm \zeta\circ\bm r/((\bm 1-\bm r )\circ\bm\alpha){\notin}{\rm span}(\bm 1)$.}
\end{itemize}
    
\end{tm}

\begin{rk} { Assume that $\bm\mu=\bm\zeta=\bm 0$ so that system \eqref{model} reduces to system \eqref{standard-incidence}. In addition, if  $\bm r_M=1$ and $\hat{\mathcal{R}}_0>1$,  then it follows from the proof of the first assertion of Theorem \ref{TH7}-{\rm (i)} that at the EEs, the total infected population tends to zero as $d_S$ tends to zero. This complements the results of \cite{allen2007asymptotic,chen2020asymptotic,li2019dynamics} on the profiles of EEs of \eqref{standard-incidence} as $d_S$ tends to zero, where it is assumed that $\bm r_m<1<\bm r_M$.  }
    
\end{rk}

 When $\bm{r}_M \ge 1$, or equivalently $\beta_i \le \gamma_i$ for some $i \in \Omega$, Theorem \ref{TH7}-{\rm (i)} suggests that reducing the dispersal rate of the susceptible population can significantly diminish the disease's impact. This conclusion also holds if the total population size $N$ is less than or equal to the threshold $ \|\bm{\zeta} \circ \bm{r} / (\bm{1} - \bm{r})\|_1$ when $\bm{r}_M < 1$. However, if $\bm{r}_M <1$ and $N$ exceeds this threshold, Theorem \ref{TH7}-{\rm (ii)} indicates that the disease may still persist even if the movement of the susceptible population is entirely restricted.  Our next result concerns the profiles of EE solutions of \eqref{model} as the dispersal rate of the infected population becomes very small.

\begin{tm}\label{TH8}  Suppose that $\bm\mu=\bm 0$. Fix $d_S>0$ and $N>0$.  If $\|N\bm\alpha/(\bm r\circ(\bm\zeta+N\bm\alpha))\|_{\infty}>1$, then there is $d_0>0$ such that system \eqref{model} has a unique EE solution $(\bm S,\bm I)$ in $\mathcal{E}$ for every $0<d_I<d_0$.  Furthermore, for every $j\in\Omega$,
\begin{equation}\label{N-star}
    \lim_{d_I\to0^+}( S_j, I_j)=\Big(N^*\alpha_j, \frac{(N^*(1- r_j)\alpha_j- r_j\zeta_j)_+}{ r_j} \Big),
\end{equation}
where $0<N^*<N$ is  uniquely determined by  the algebraic equation
\begin{equation}\label{N-star-2}
    N=N^*+\sum_{j\in\Omega}\frac{(N^*(1-r_j)\alpha_j-r_j\zeta_j)_+}{r_j}.
\end{equation}

\end{tm}

\begin{rk} Assume that the hypotheses of Theorem \ref{TH8} hold. Then there is some $i\in\Omega$ such that $N\alpha_i>r_i(\zeta_i+N\alpha_i)$, which implies that $\tilde{\Omega}=\{j\in\Omega : \beta_j>\gamma_j\}$ is not empty.
\begin{itemize}
\item[\rm (i)]If $\Omega\setminus\tilde{\Omega}\ne \emptyset$, then by \eqref{N-star}, the infected populations at EEs residing on the patches of $\Omega\setminus\tilde{\Omega}$ converge to zero as $d_I$ becomes very small. Note also from the fact that $N>N^*$ in Theorem \ref{TH8}, the infected populations at the EEs persist on some of the patches of $\tilde{\Omega}$ as $d_I$ gets very small. In particular, if $\Omega$ consists of only two patches, say $\Omega=\{1,2\}$, and $\tilde{\Omega}=\{2\}$, then as the dispersal rate $d_I$ of the infected population approaches zero, we have that at the EE solution, the infected population living on patch 2 persist while those living on patch 1 die out (see Numerical Experiment 13).

\item[\rm (ii)] Set $N_{\rm critical}:=\Big(\max_{j\in\tilde{\Omega}}\frac{r_j\zeta_j}{(1-r_j)\alpha_j}\Big)\Big(1+\sum_{j\in\tilde{\Omega}}\frac{(1-r_j)\alpha_j}{r_j}\Big)-\sum_{j\in\tilde{\Omega}}\zeta_j$. It follows from \eqref{N-star} and \eqref{N-star-2} that the infected populations at the EEs persist exactly on all patches of $\tilde{\Omega}$ as $d_I$ tends to zero  if and only if $N>N_{\rm critical}$.  Indeed, consider the function 
$$
g(N^*)=N^*+\sum_{j\in\Omega}\frac{(N^*(1-r_j)\alpha_j-r_j\zeta_j)_+}{r_j}=N^*+\sum_{j\in\tilde{\Omega}}\frac{(N^*(1-r_j)\alpha_j-r_j\zeta_j)_+}{r_j}\quad N^*\ge 0,
$$
$g$ is strictly increasing and continuous, $g(0)=0$, and $g(N^*)\to\infty$ as $N^*\to\infty$.  Note also that for $\underline{N}^*=\max_{i\in\tilde{\Omega}}\frac{\zeta_ir_i}{(1-r_i)\alpha_i}$, we have 
$$ 
g(\underline{N}^*)=\Big(1+\sum_{i\in\tilde{\Omega}}\frac{(1-r_i)\alpha_i}{r_i}\Big)\Big(\max_{i\in\tilde{\Omega}}\frac{\zeta_ir_i}{(1-r_i)\alpha_i}\Big)-\sum_{i\in\tilde{\Omega}}\zeta_i.
$$
Thus, if $N>N_{\rm critical}=g(\underline{N}^*)$, by the intermediate value theorem and the strict monotonicity of $g$, there is a unique $N^*>\underline{N}^*$ such that $g(N^*)=N$. Since $N^*>\underline{N}^*$, then $N^*(1-r_i)\alpha_i>r_i\zeta_i$ for all $i\in\tilde{\Omega}$. However, if $N{\leq}N_{\rm critical}$, then the unique positive number $N^*$ satisfying $g(N^*)={N}$ must be less than or equal to $ \underline{N}^*$, in which case the set $\{i \in\Omega : N^*(1-r_i)\alpha_i\le r_i\zeta_i\}$ is not empty.

\item[\rm (iii)] If $\bm r_M<1$ and $N>\max\{\|\bm\zeta\circ \bm r/((\bm 1-\bm r)\circ\bm\alpha)\|_{\infty},\|{\bm\zeta\circ\bm r/}{((\bm 1-\bm r)\circ\bm\alpha)}\|_{\infty}\Big(1+\|{(\bm 1-\bm r)\circ\bm \alpha/}{\bm r}\|_1\Big)-\|\bm \zeta\|_1\}$, it follows from Theorems \ref{TH7}-{\rm (ii)} and \ref{TH8} that, as either the dispersal rate of susceptible or infected population becomes very small, the disease will persist on all patches. 

\end{itemize}
\end{rk}

\subsection{Numerical Simulations}
In this section, we carry out some numerical simulations to illustrate our theoretical results. For all the simulations, we consider  two patches, that is $\Omega=\{1,2\}$, and take $L_{12}=0.4$, $L_{21}=0.1$. So $L_{11}=-0.1$, $L_{22}=-0.4$ and $\bm\alpha=(0.8, 0.2)^{T}$. We also fix $N=4$ in Experiment 1 through Experiment 7. We simulate two scenarios: $\bm\mu>0$ and $\bm\mu=0$.

\medskip

\subsubsection{Case of $\bm\mu>0$}

In this subsection, we simulate the large-time behavior of solutions of system \eqref{model} when $\bm\mu>0$. 
We fix parameters $d_{S}=1$, $d_{I}=1$, $\bm\beta=(1, 1)^{T}$, $\bm\gamma=(1, 1)^T$, $\bm\zeta=(0.5, 0.5)^T$. We vary the values of $\bm\mu$ and $(\bm S_{0}, \bm I_{0})$ to see how the long-time behavior of solutions of \eqref{model} changes. Experiment 1 concerns the case of $\bm\mu\gg\bm 0$, Experiment 2 focuses on the case of $\mu_1>0$ and $\mu_2=0$, while Experiment 3 is for the case of $\mu_1=0$ and $\mu_2>0$. These three simulations are consistent with Theorem \ref{TH1}.

\medskip

\noindent{\bf Experiment 1.}
Let $\bm\mu=(0.1, 0.1)^T$. We take $(\bm S^{0}, \bm I^{0})=((1,1)^T,(1,1)^T)$. Numerically, we observe that $(\bm S(t), \bm I(t))\to ((2.8635,0.7159)^T,\bm 0)\approx \big(\big(N-\int_0^{\infty}\sum_{j\in\Omega}\mu_j I_j(t)dt\big)\bm\alpha,{\bm 0}\big)$ as $t$ becomes large (see Figure \ref{Ex1}(a)). We then take different initials, we observe the same phenomenon (see Figure \ref{Ex1}(b) for $(\bm S^{0}, \bm I^{0})=((1.5,0.1)^T,(0.5,1.9)^T)$ and Figure \ref{Ex1}(c) for $(\bm S^{0}, \bm I^{0})=((0.1,1.9)^T,(0.5,1.5)^T)$).

\begin{figure}[!ht]
\begin{center}
\subfigure[]{
\resizebox*{0.300\linewidth}{!}{\includegraphics{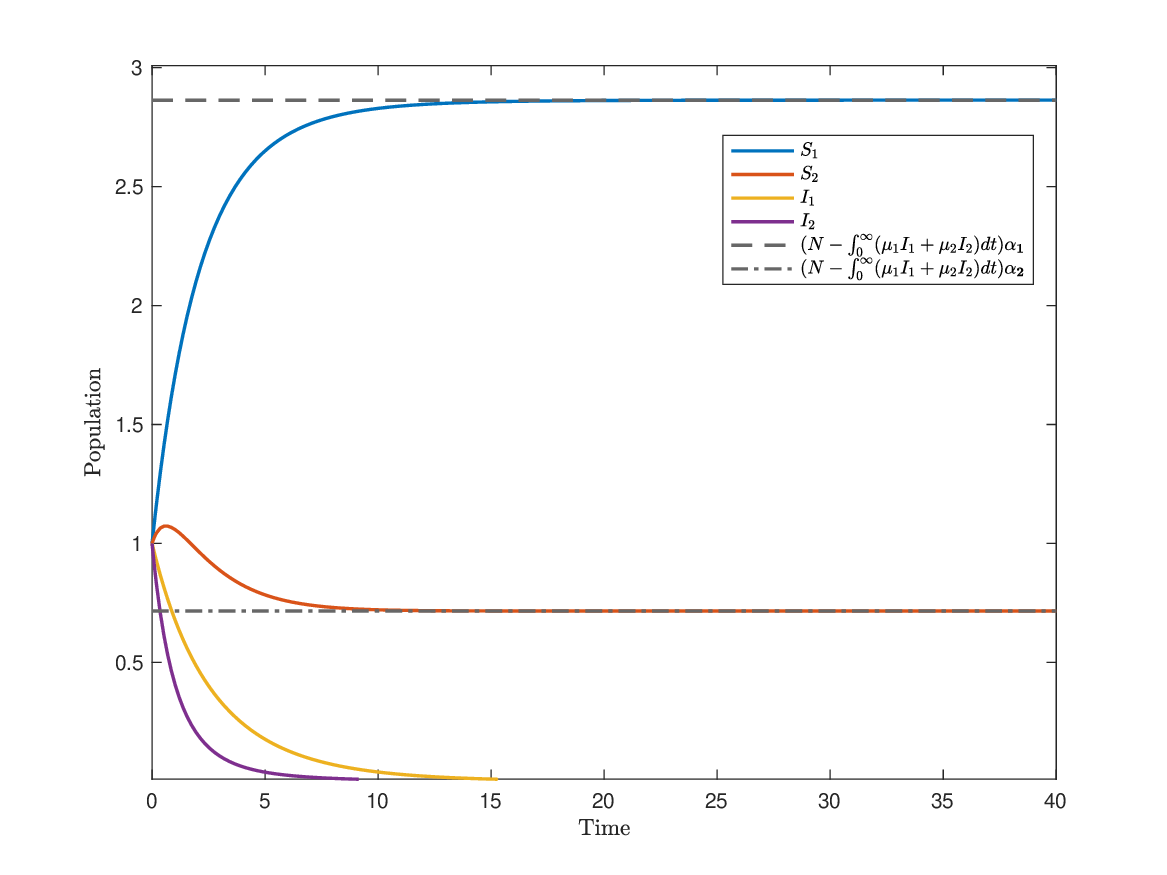}}}
\subfigure[]{\resizebox*{0.300\linewidth}{!}{\includegraphics{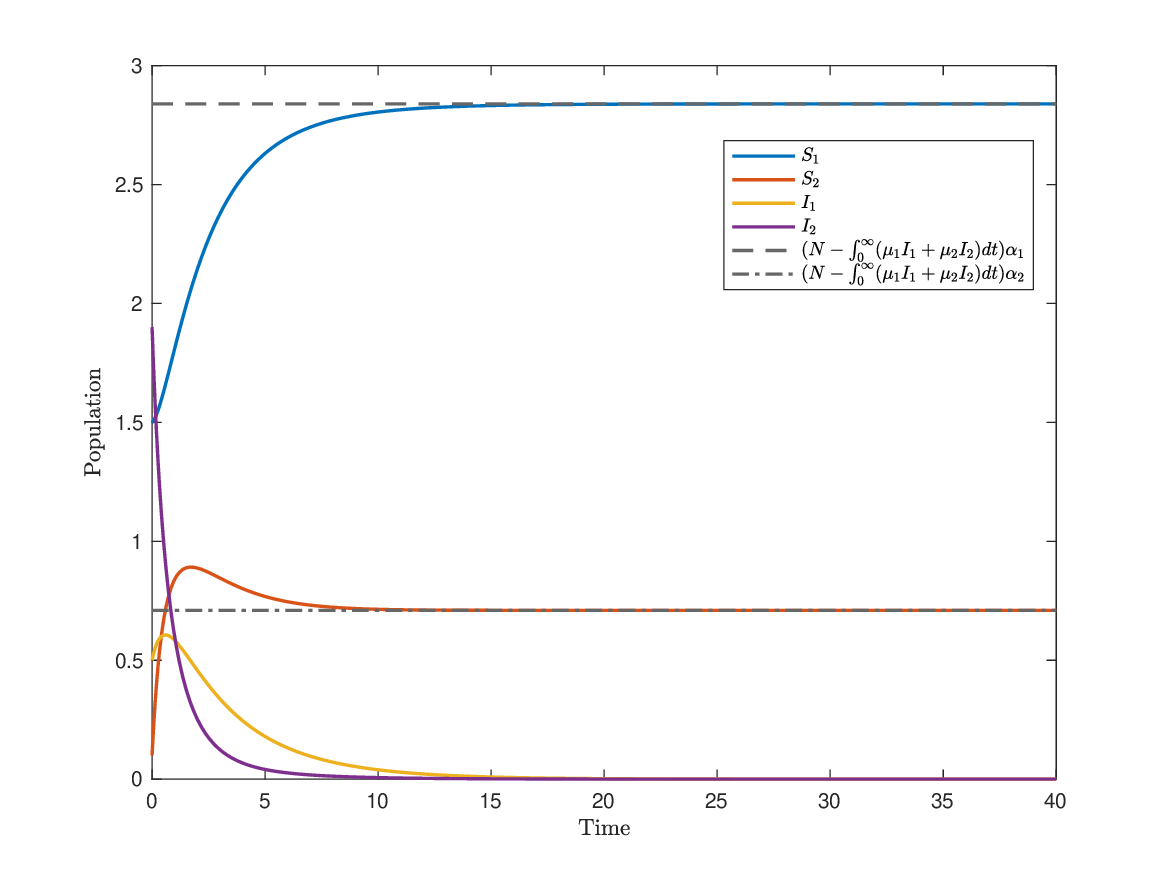}}}
\subfigure[]{\resizebox*{0.300\linewidth}{!}{\includegraphics{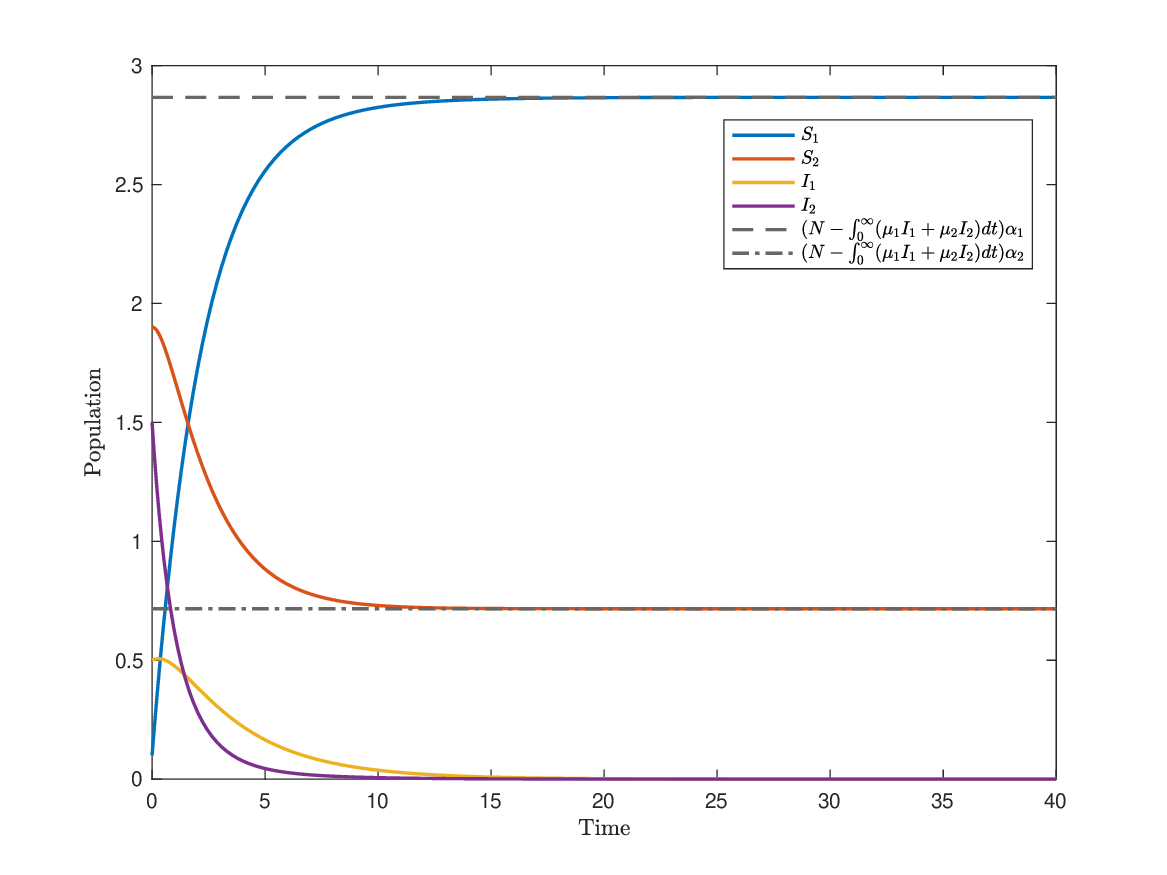}}}
\caption{Numerical simulations illustrating global dynamics of \eqref{model} when $\bm\mu=(0.1, 0.1)^T\gg \bm 0$. }.
\label{Ex1}
\end{center}
\end{figure}

\medskip

\noindent{\bf Experiment 2.} 
Let $\bm\mu=(0.1, 0)^T$. We take $(\bm S^{0}, \bm I^{0})=((1,1)^T,(1,1)^T)$. Numerically, we observe that $(\bm S(t), \bm I(t))\to ((2.9562,0.7391)^T,\bm 0)\approx \big(\big(N-\int_0^{\infty}\sum_{j\in\Omega}\mu_j I_j(t)dt\big)\bm\alpha,{\bm 0}\big)$ as $t$ becomes large (see Figure \ref{Ex2}(a)). We then take different initials, we observe the same phenomenon (see Figure \ref{Ex2}(b) for $(\bm S^{0}, \bm I^{0})=((1.5,0.1)^T,(0.5,1.9)^T)$ and Figure \ref{Ex2}(c) for $(\bm S^{0}, \bm I^{0})=((0.1,1.9)^T,(0.5,1.5)^T)$). 

\begin{figure}[!ht]
\begin{center}
\subfigure[]{
\resizebox*{0.300\linewidth}{!}{\includegraphics{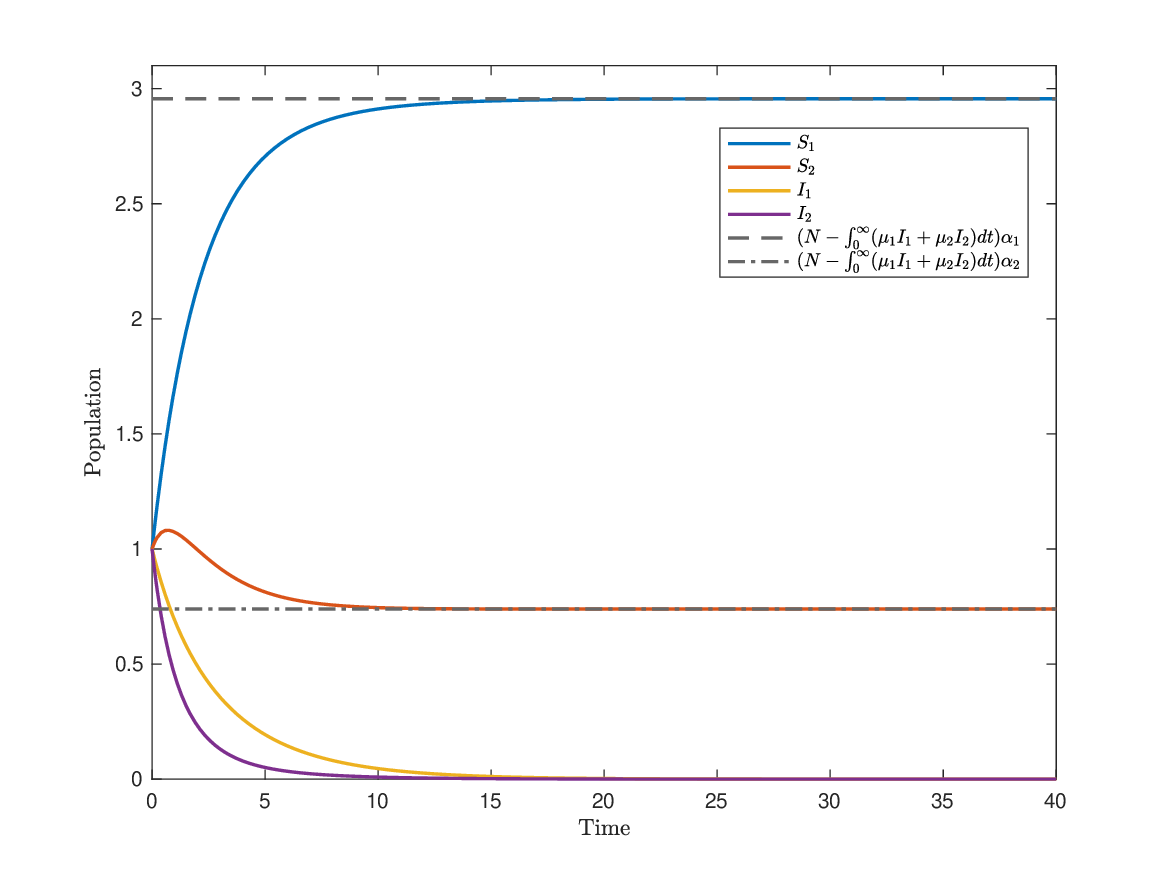}}}
\subfigure[]{\resizebox*{0.300\linewidth}{!}{\includegraphics{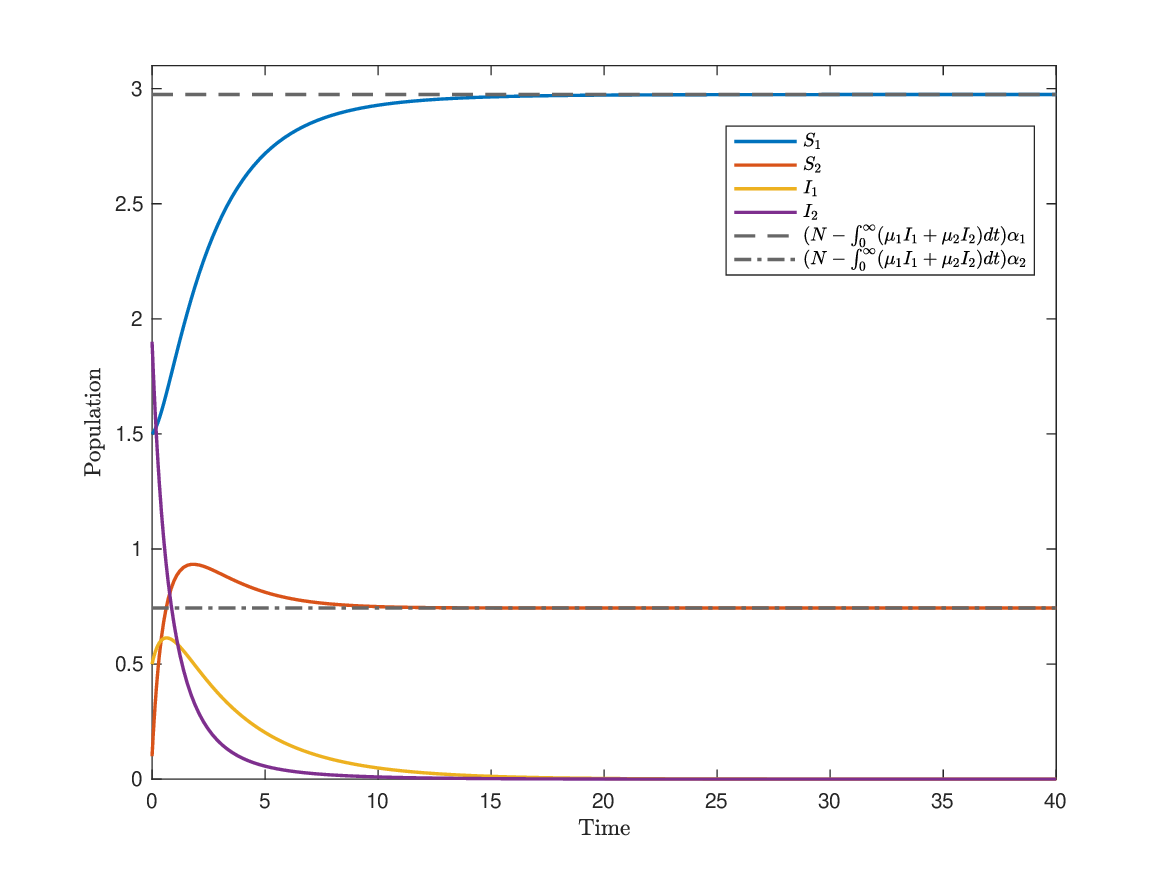}}}
\subfigure[]{\resizebox*{0.300\linewidth}{!}{\includegraphics{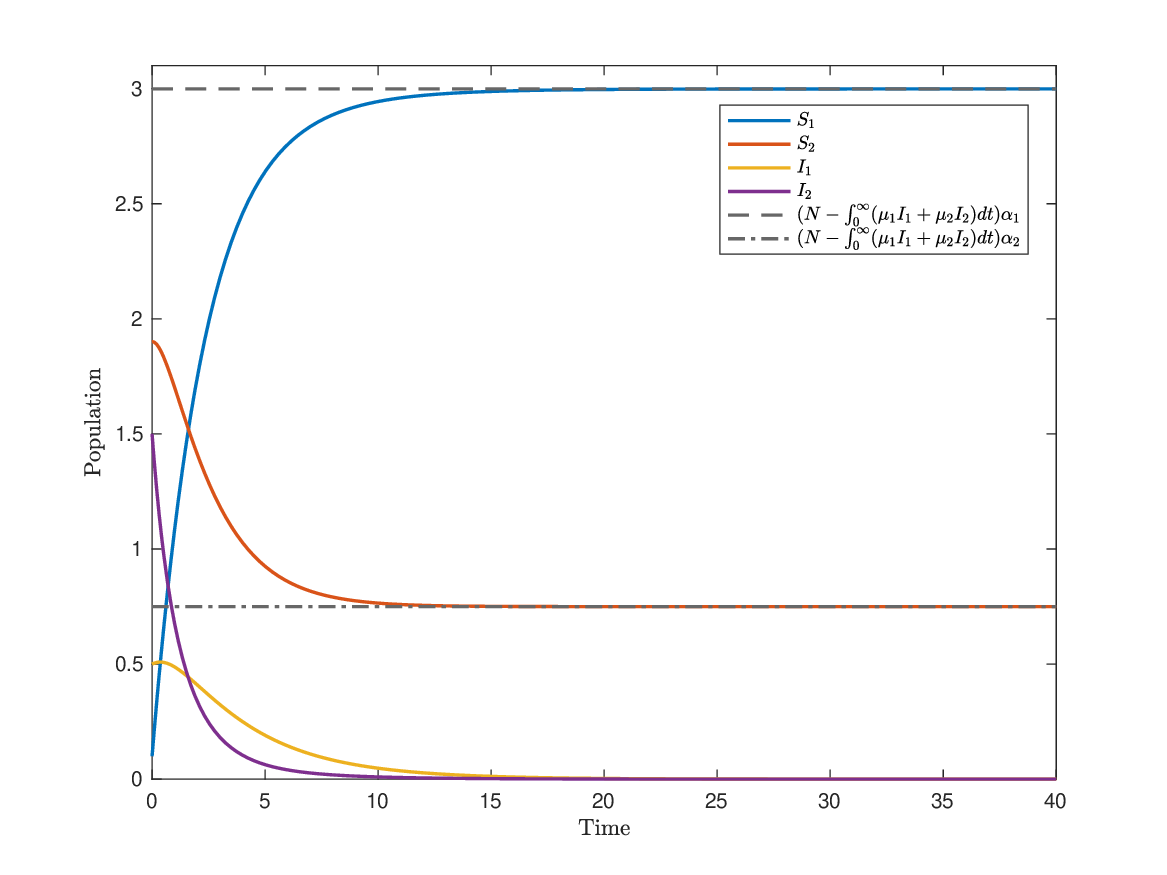}}}
\caption{Numerical simulations illustrating global dynamics of \eqref{model} when $\bm\mu=(0.1, 0)^T$}.
\label{Ex2}
\end{center}
\end{figure}

\medskip

\noindent{\bf Experiment 3.}  
Let $\bm\mu=(0,0.1)^T$. We take $(\bm S^{0}, \bm I^{0})=((1,1)^T,(1,1)^T)$. Numerically, we observe that $(\bm S(t), \bm I(t))\to ((3.0854,0.7713)^T,\bm 0)\approx \big(\big(N-\int_0^{\infty}\sum_{j\in\Omega}\mu_j I_j(t)dt\big)\bm\alpha,{\bm 0}\big)$ as $t$ becomes large (see Figure \ref{Ex3}(a)). We then take different initials, we observe the same phenomenon (see Figure \ref{Ex3}(b) for $(\bm S^{0}, \bm I^{0})=((1.5,0.1)^T,(0.5,1.9)^T)$ and Figure \ref{Ex3}(c) for $(\bm S^{0}, \bm I^{0})=((0.1,1.9)^T,(0.5,1.5)^T)$). 

\begin{figure}[!ht]
\begin{center}
\subfigure[]{
\resizebox*{0.300\linewidth}{!}{\includegraphics{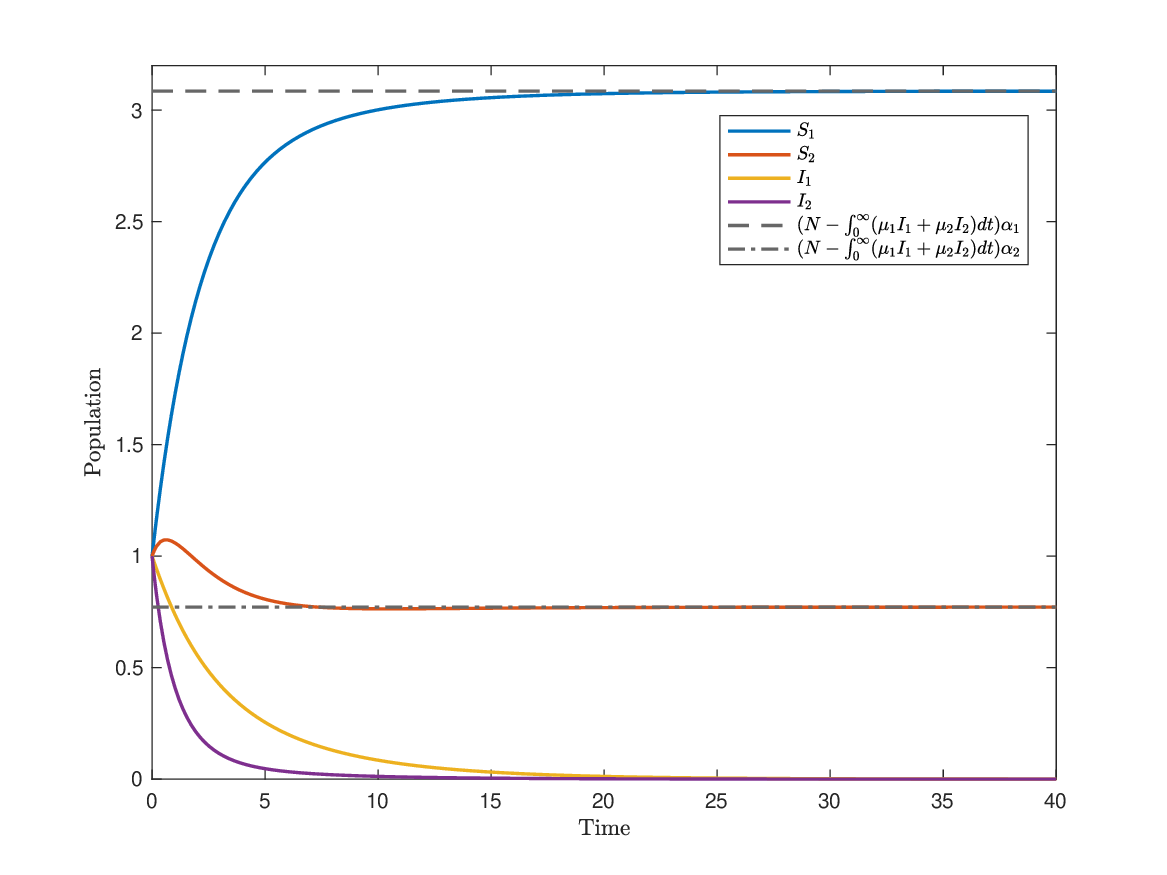}}}
\subfigure[]{\resizebox*{0.300\linewidth}{!}{\includegraphics{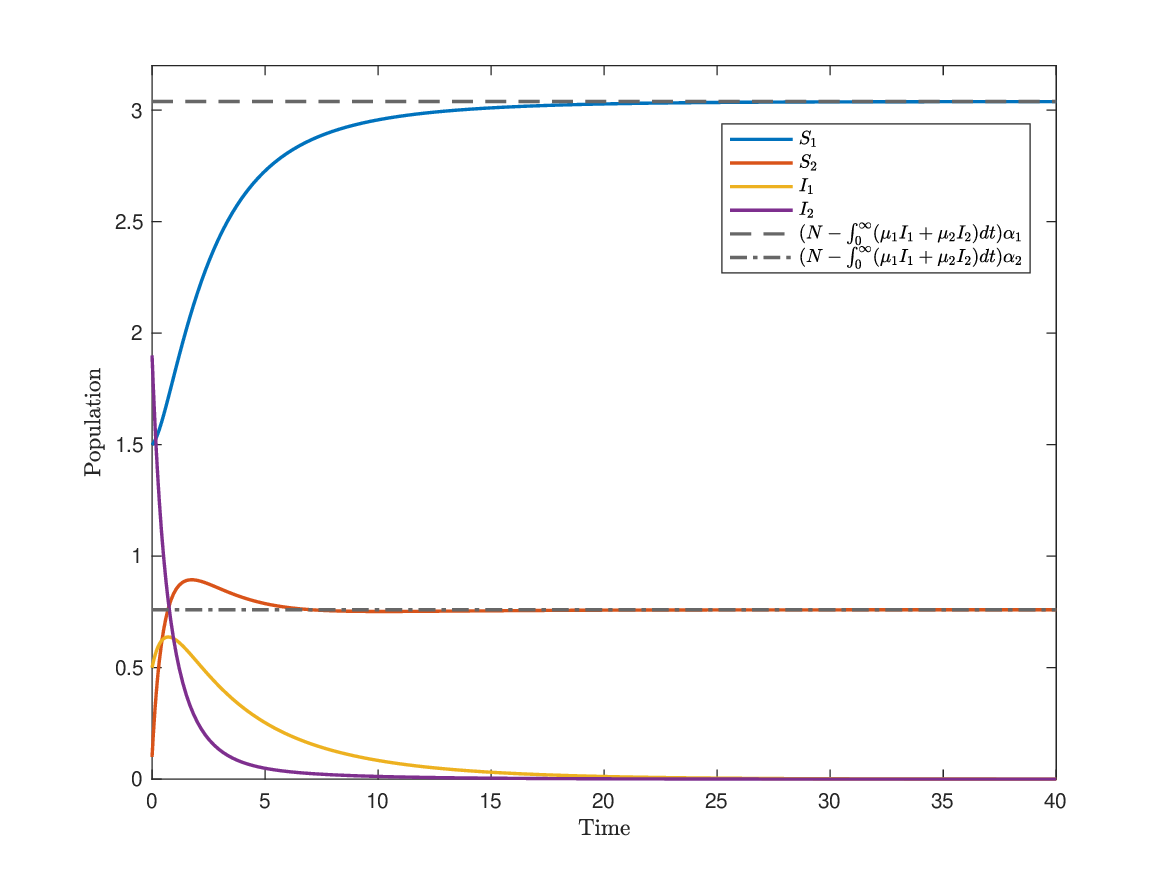}}}
\subfigure[]{\resizebox*{0.300\linewidth}{!}{\includegraphics{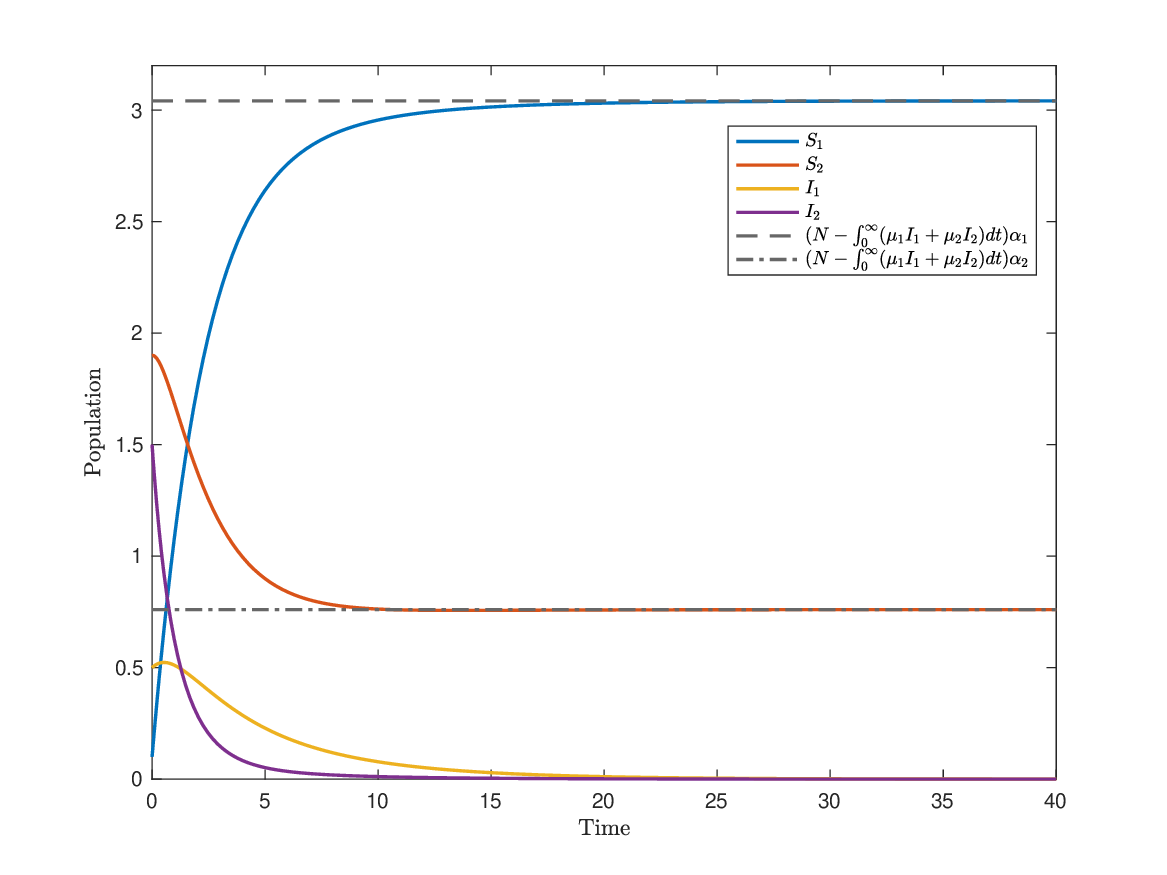}}}
\caption{Numerical simulations illustrating global dynamics of \eqref{model} when $\bm\mu=(0, 0.1)^T$}.
\label{Ex3}
\end{center}
\end{figure}

\subsubsection{Case of $\bm \mu=\bm 0$}

In this subsection, we simulate the global dynamics of \eqref{model} when $\bm\mu=0$. 
We vary the values of $\bm\beta$, $\bm\gamma$, $\bm\zeta$ and $(\bm S^{0}, \bm I^{0})$ to see how these parameters affect the global dynamics of \eqref{model}.

\medskip

\noindent{\bf Experiment 4.} 
Let $\bm\beta=(1, 1)^{T}$, $\bm\gamma=(1, 1)^T$, $\bm\zeta=(0.5, 0.5)^T$, and $d_I=1$. Then we have $\tilde{\mathcal{R}}_0=0.8889<1$. Take $(\bm S^{0}, \bm I^{0})=((1,1)^T,(1,1)^T)$ and $d_S=1$. As time becomes large, we observe that $(\bm S(t),\bm I(t))$ goes to $ (N\bm\alpha,\bm 0)=((3.2,0.8)^{T}, \bm 0$) (see Figure \ref{Ex4}(a)). Taking different initial data, we observe the same phenomenon (see Figure \ref{Ex4}(b) for $(\bm S^{0}, \bm I^{0})=((1.5,0.1)^T,(0.5,1.9)^T)$ and Figure \ref{Ex4}(c) for $(\bm S^{0}, \bm I^{0})=((0.1,1.9)^T,(0.5,1.5)^T)$). This simulation indicates that $(N\bm\alpha,\bm 0)$ is global asymptotically stable, which is consistent with theorem \ref{TH2}. 
Next, we vary the dispersal rate $d_S$ of the susceptible population: First, let $d_{S}=2$, $(\bm S^{0}, \bm I^{0})=((1,1)^T,(1,1)^T)$ and keep the other parameters the same as before. We observe that $(\bm S(t),\bm I(t))$ still goes to $(N\bm\alpha,\bm 0)=((3.2,0.8)^{T}, \bm 0$) (see Figure \ref{Ex4-2}(a)). Next, we decrease the values of $d_{S}$, we observe the same phenomenon (see Figure \ref{Ex4-2}(b) for $d_{S}=0.5$ and Figure \ref{Ex4-2}(c) for $d_{S}=10^{-5}$). For each $d_{S}$, if we choose different initial data, we also observe the convergence of $(\bm S(t),\bm I(t))$ to 
$(N\bm\alpha,\bm 0)=((3.2,0.8)^{T}, \bm 0)$. These simulations are consistent with Theorem \ref{TH2}. Moreover, the simulations indicate that when $d_{S}$ becomes smaller, it takes a longer time for the solution to stabilize at the DFE. This strongly highlights the effect of the susceptible population on the dynamics of the disease.

\begin{figure}[!ht]
\begin{center}
\subfigure[]{
\resizebox*{0.300\linewidth}{!}{\includegraphics{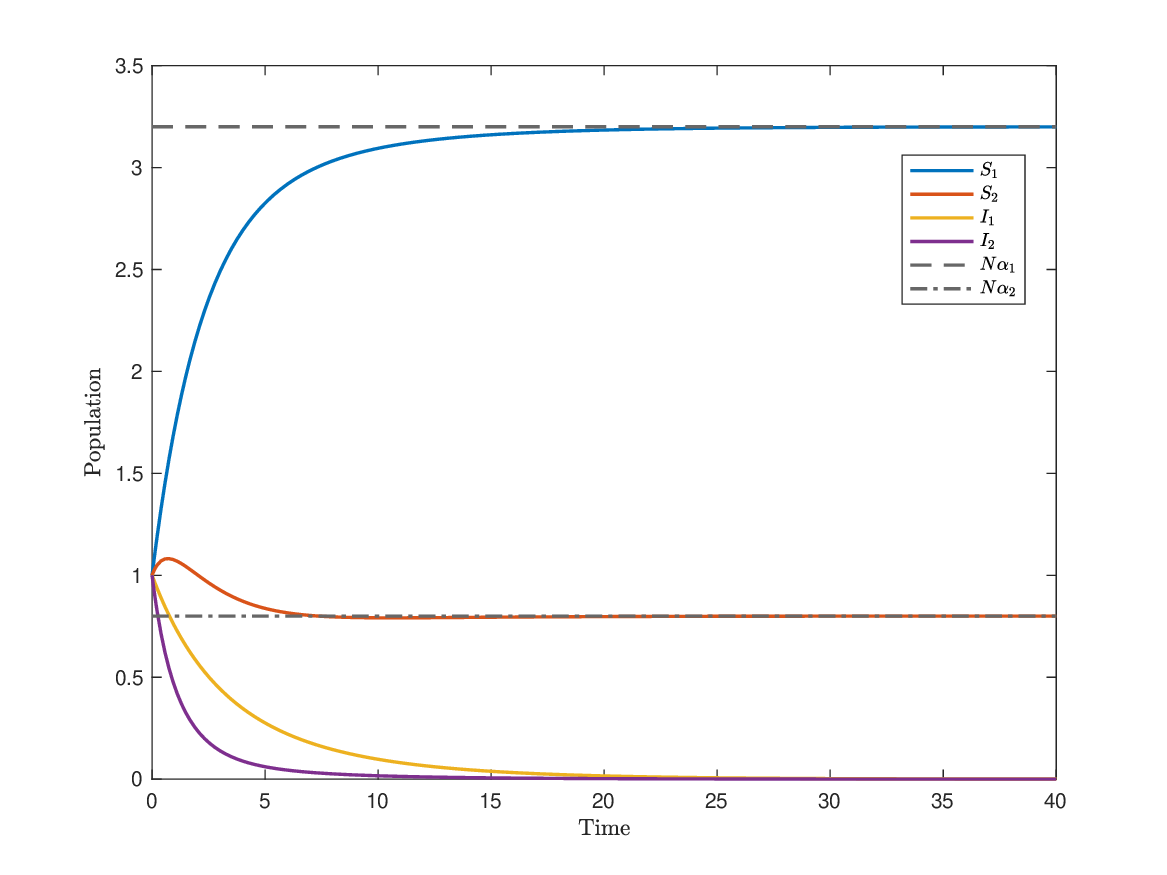}}}
\subfigure[]{\resizebox*{0.300\linewidth}{!}{\includegraphics{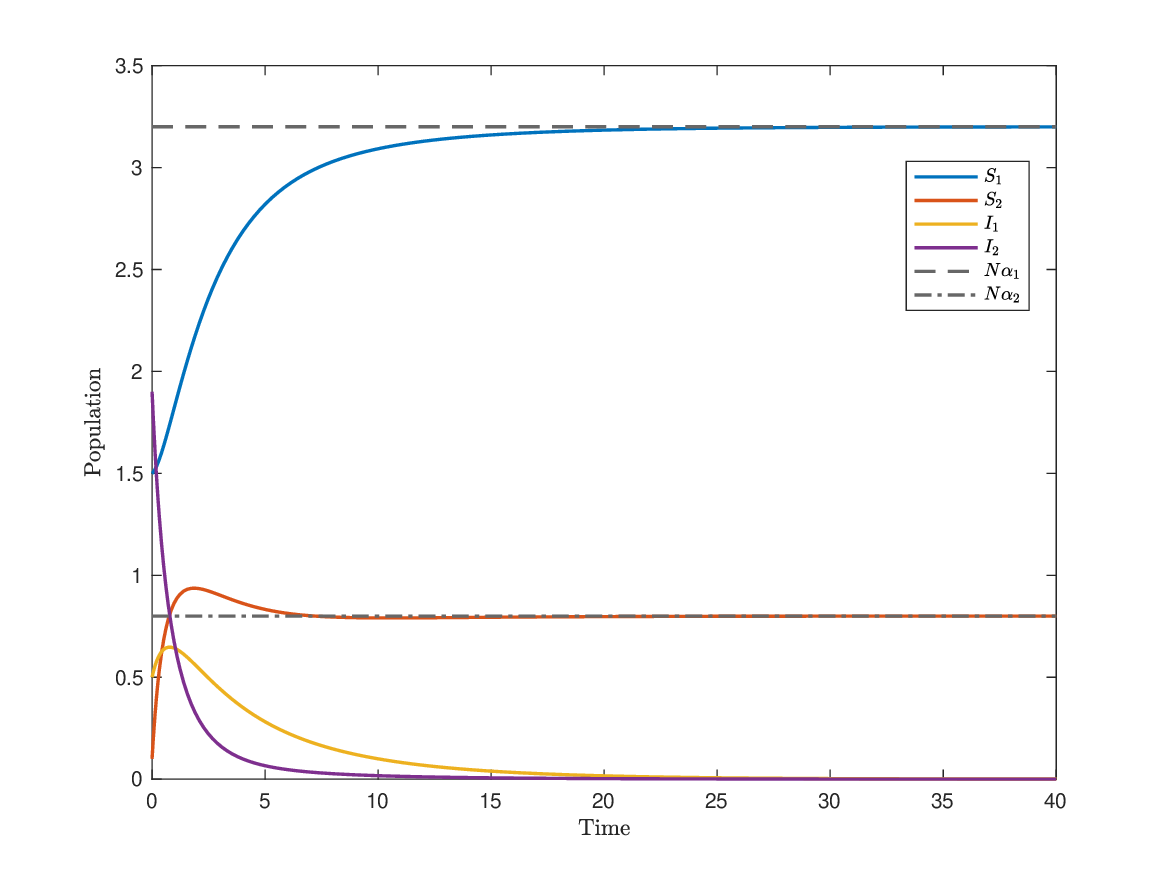}}}
\subfigure[]{\resizebox*{0.300\linewidth}{!}{\includegraphics{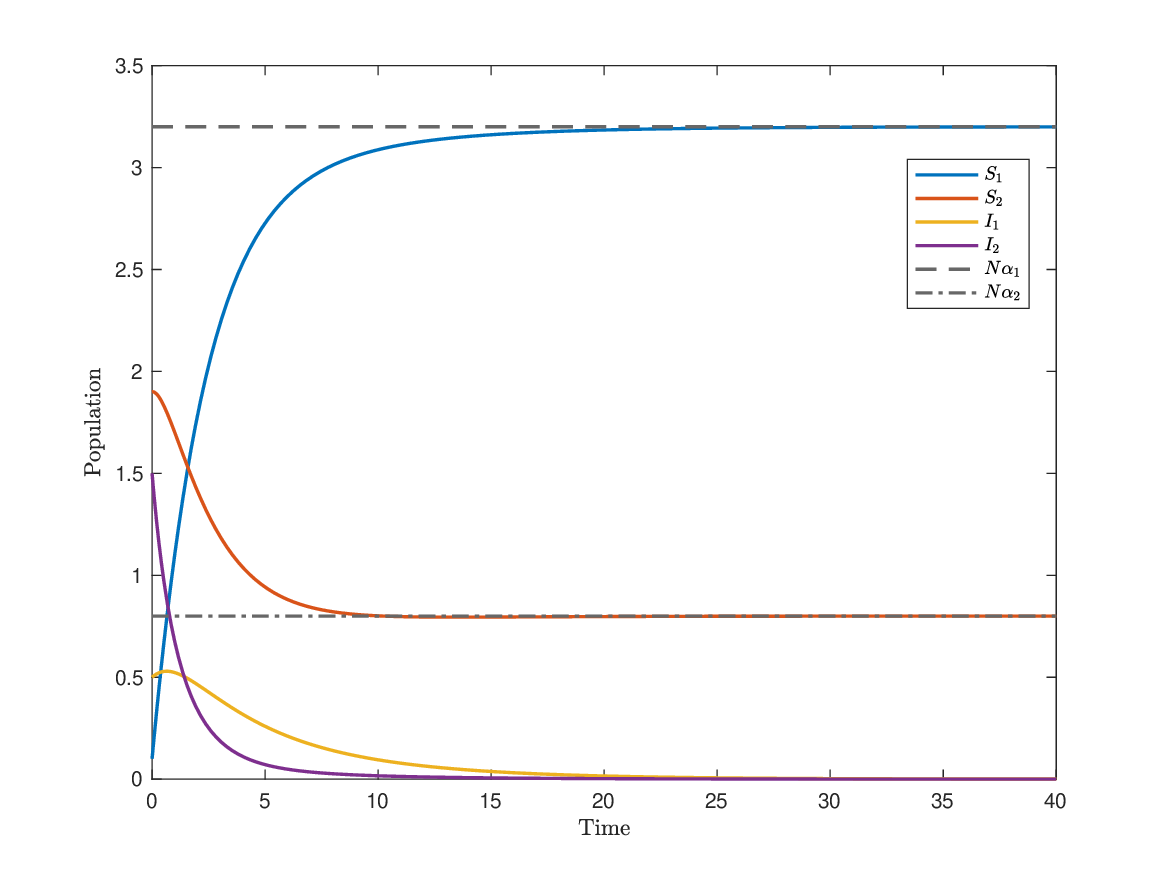}}}
\caption{Numerical simulations illustrating global dynamics of \eqref{model} when $\bm\mu=(0, 0)^T$ and the hypotheses of Theorem \ref{TH2} are satisfied for the same  population dispersal rates with three different initial data.}
\label{Ex4}
\end{center}
\end{figure}

\begin{figure}[!ht]
\begin{center}
\subfigure[]{
\resizebox*{0.300\linewidth}{!}{\includegraphics{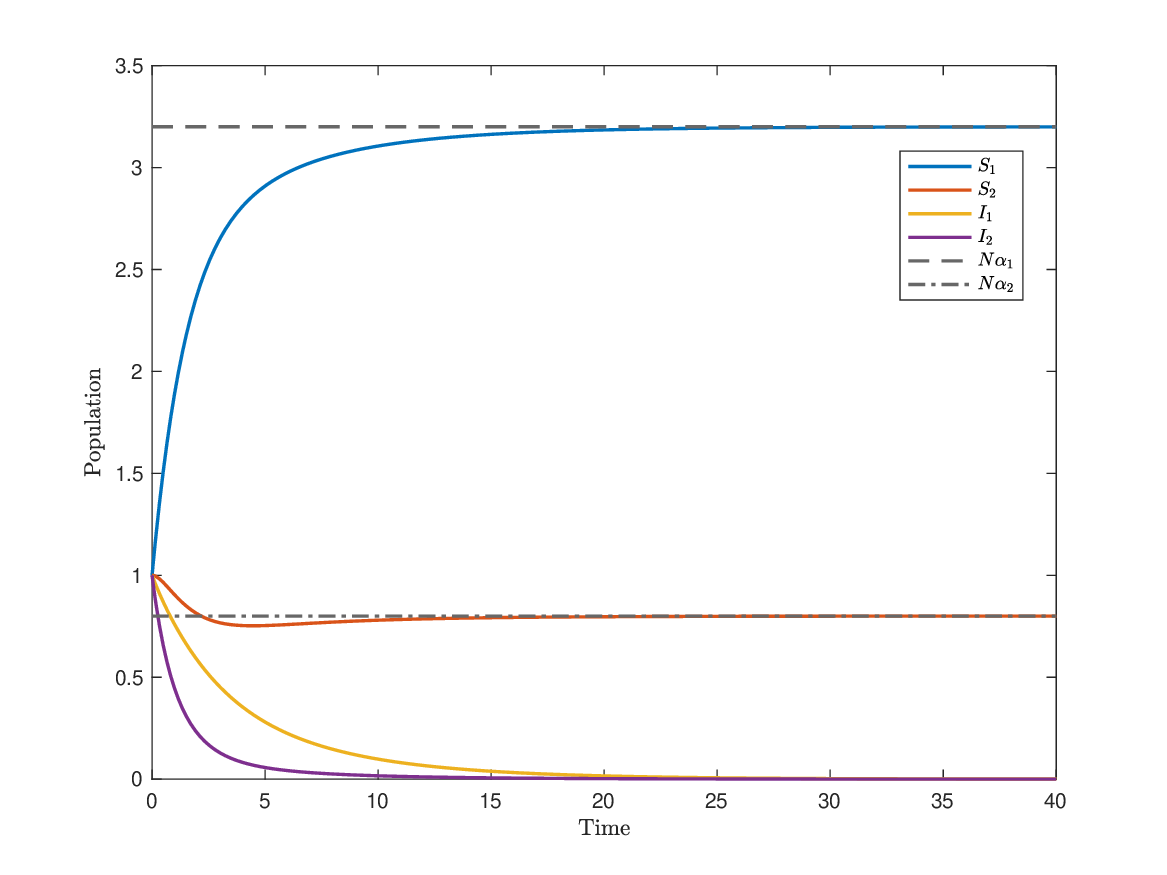}}}
\subfigure[]{\resizebox*{0.300\linewidth}{!}{\includegraphics{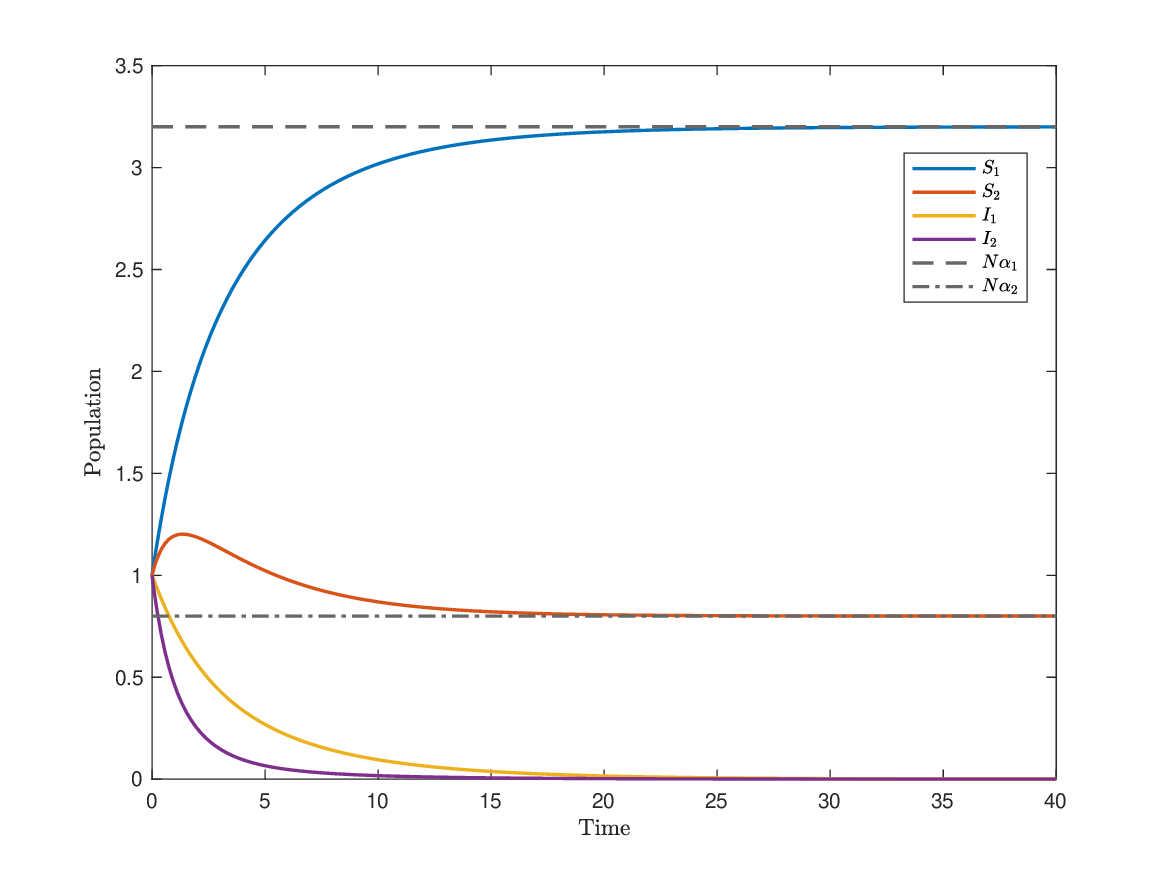}}}
\subfigure[]{\resizebox*{0.300\linewidth}{!}{\includegraphics{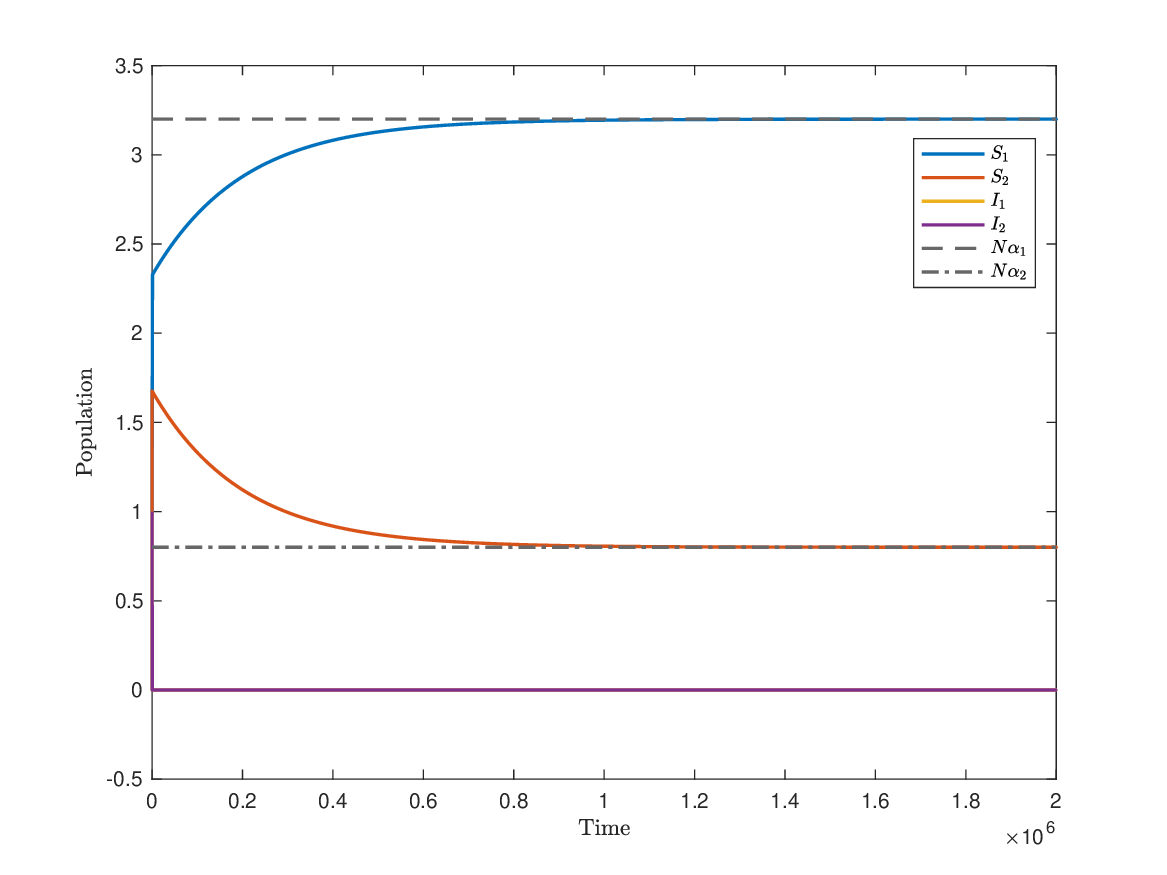}}}
\caption{Numerical simulations illustrating global dynamics of \eqref{model} when $\bm\mu=(0, 0)^T$ under the hypotheses of Theorem \ref{TH2} with the same initial data but different population dispersal rates. (a): $d_{S}=2$ and $d_I=1$, (b): $d_{S}=0.5$ and $d_I=1$, (c): $d_{S}=10^{-5}$ and $d_I=1$.}
\label{Ex4-2}
\end{center}
\end{figure}

\medskip

\noindent{\bf Experiment 5.} 
Let $\bm\beta=\bm\gamma=(1.5, 0.5)^{T}$, $\bm\zeta=(0.8,0.2)^{T}$.  Hence ${\bm r}=\bm 1\in{\rm span}(\bm 1)$ and $\bm\zeta=\bm\alpha\in{\rm span}(\bm\alpha)$, so that the hypotheses of Theorem \ref{TH3} hold. With these choices, we have that $\tau=1$ and $m=1$ in Remark \ref{RK2}, and hence $\mathcal{R}_0=N/(\tau(m+N))=\frac{4}{5}<1$. Let $d_S=0.5$ and $d_I=2$. 
Then, we subsequently run our numerical simulations  for initial data  $(\bm S^{0}, \bm I^{0})=((1,1)^T,(1,1)^T)$ (see Figure \ref{Ex5}(a)), $(\bm S^{0}, \bm I^{0})=((1.5,0.1)^T,(0.5,1.9)^T)$ (see Figure \ref{Ex5}(b)), and $(\bm S^{0}, \bm I^{0})=((0.1,1.9)^T,(0.5,1.5)^T)$) (see Figure \ref{Ex5}(c)). As time becomes larger and larger, we observe numerically that $(\bm S(t),\bm I(t))$ goes to $(N\bm\alpha,\bm 0)=((3.2,0.8)^{T}, \bm 0$), which agrees with the conclusions of Theorem \ref{TH3} (i) and Remark \ref{RK2}-{\rm (i)}.

\begin{figure}[!ht]
\begin{center}
\subfigure[]{
\resizebox*{0.300\linewidth}{!}{\includegraphics{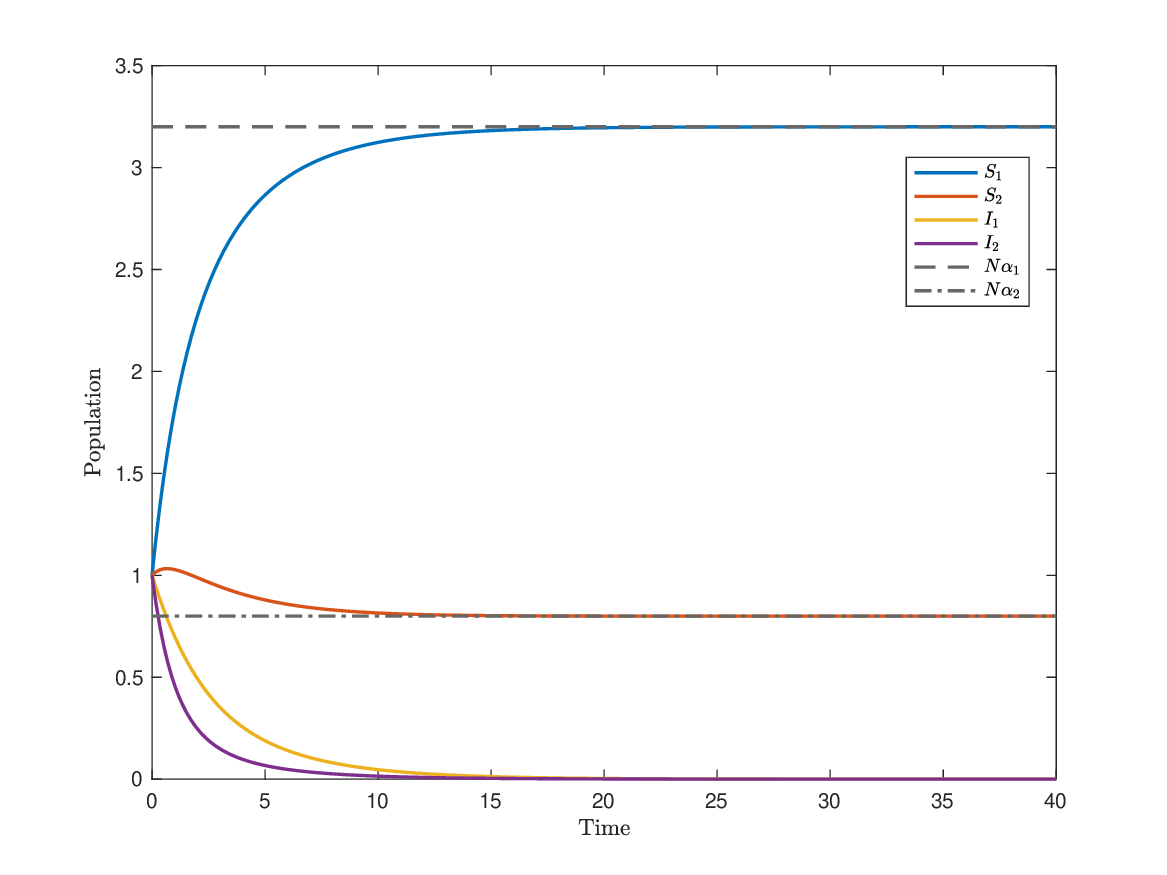}}}
\subfigure[]{\resizebox*{0.300\linewidth}{!}{\includegraphics{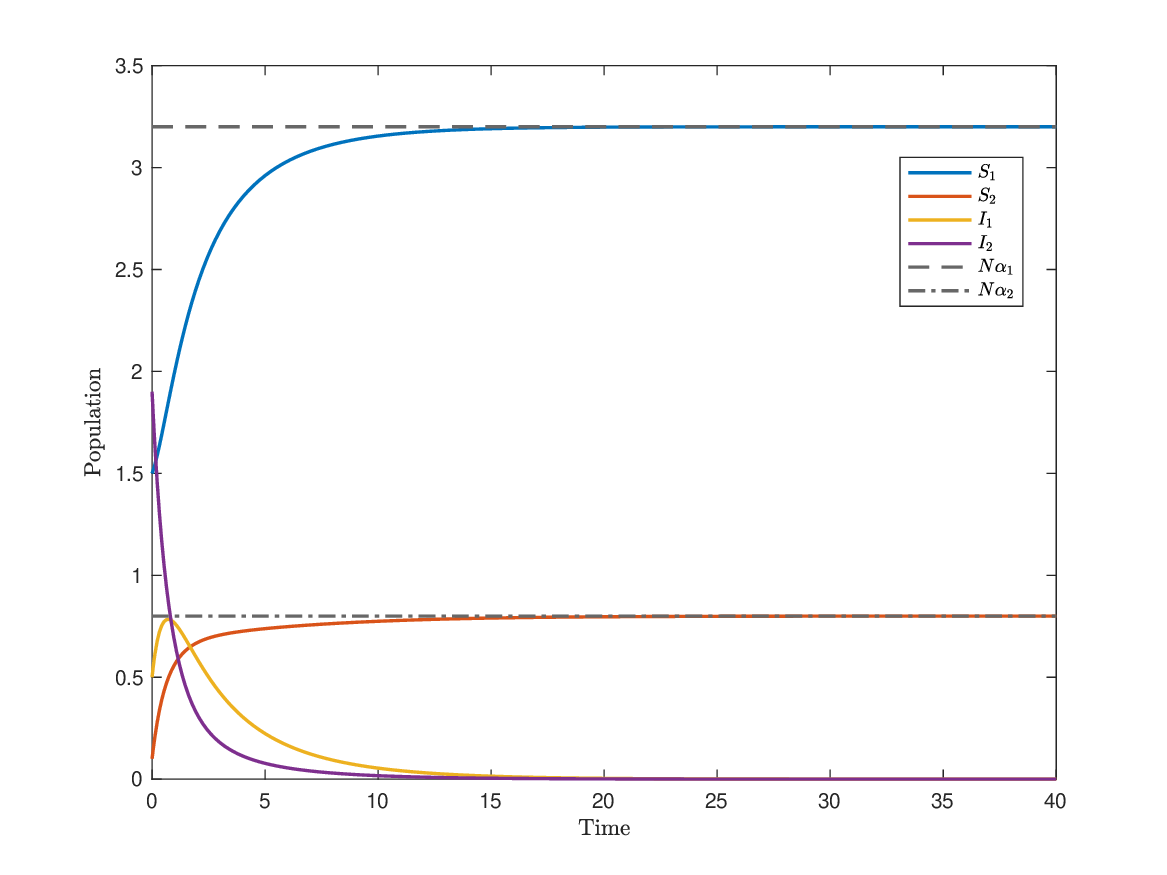}}}
\subfigure[]{\resizebox*{0.300\linewidth}{!}{\includegraphics{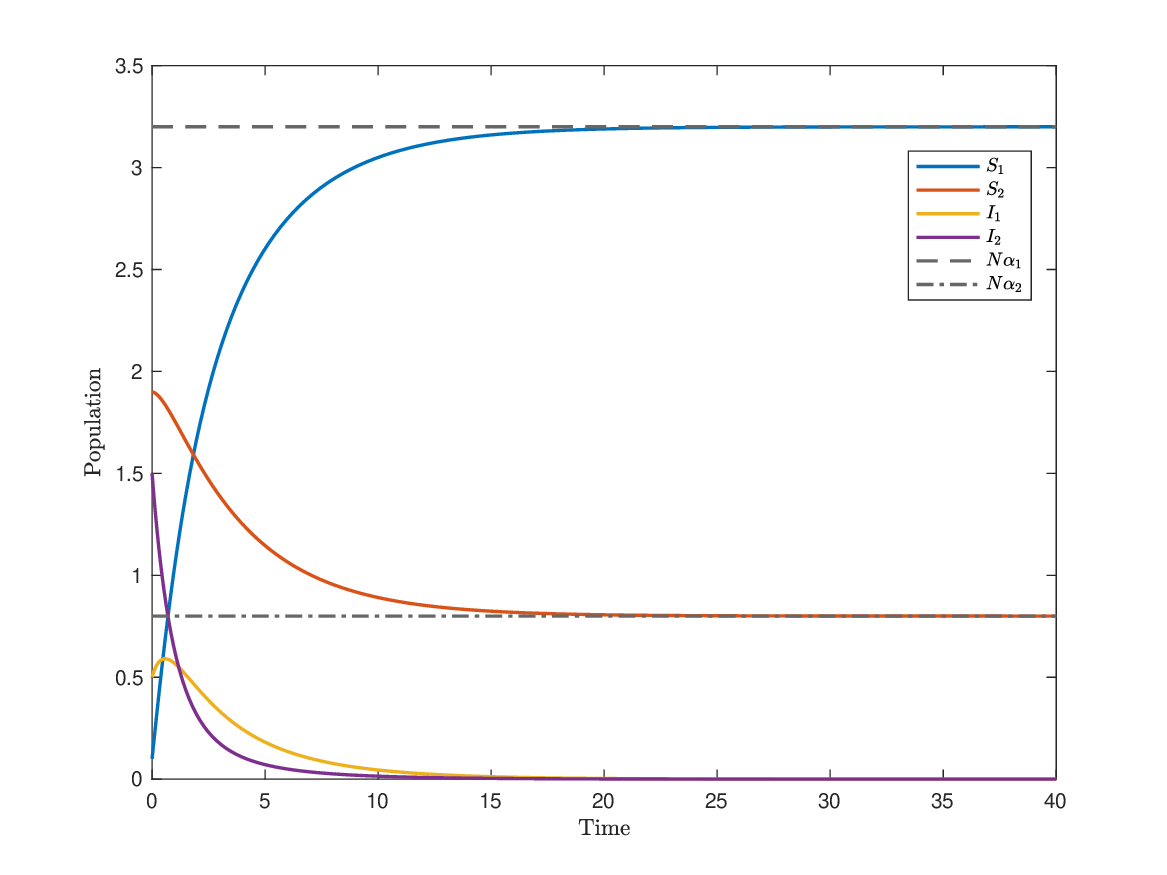}}}
\caption{Numerical simulations illustrating global dynamics of \eqref{model} when $\bm\mu=(0, 0)^T$ and ${\mathcal{R}}_0<1$ when the hypotheses of Theorem \ref{TH3}-{\rm (i)} and Remark \ref{RK2}-{\rm (i)} are satisfied.}.
\label{Ex5}
\end{center}
\end{figure}

\medskip

\noindent{\bf Experiment 6.} 
Let $\bm\gamma=(1.5, 0.5)^{T}$, $\bm\beta=2\bm\gamma=(3, 1)^{T}$, $\bm\zeta=5\bm\alpha=(4,1)^{T}$, so that the hypotheses of Theorem \ref{TH3} hold. In this case, we have $\tau=\frac{1}{2}<1$ and $m=5$ in Remark \ref{RK2}, and hence $\mathcal{R}_0=N/(\tau(m+N))=\frac{8}{9}<1$. Let $d_{S}=0.5$ and $d_{I}=2$. Then, we subsequently run our numerical simulations  for initial data  $(\bm S^{0}, \bm I^{0})=((1,1)^T,(1,1)^T)$ (see Figure \ref{Ex6}(a)), $(\bm S^{0}, \bm I^{0})=((1.5,0.1)^T,(0.5,1.9)^T)$ (see Figure \ref{Ex6}(b)), and $(\bm S^{0}, \bm I^{0})=((0.1,1.9)^T,(0.5,1.5)^T)$) (see Figure \ref{Ex6}(c)). For each initial condition, we observe that the disease is eventually  eradicated and the susceptible population stabilizes at 
$(3.2,0.8)^{T}=N\bm\alpha$ eventually, which is consistent with  Theorem \ref{TH3}-{\rm (i)} and Remark \ref{RK2}-{\rm (ii)}.


\begin{figure}[!ht]
\begin{center}
\subfigure[]{
\resizebox*{0.300\linewidth}{!}{\includegraphics{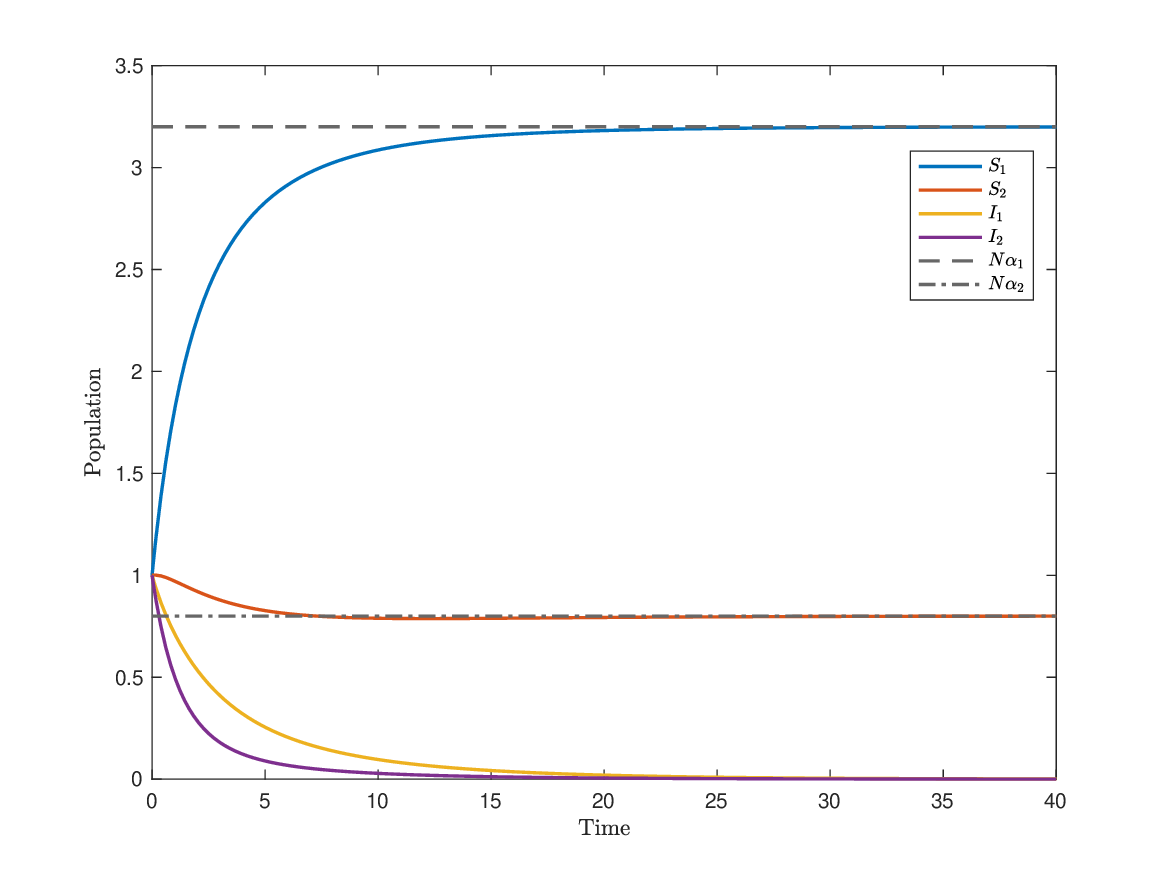}}}
\subfigure[]{\resizebox*{0.300\linewidth}{!}{\includegraphics{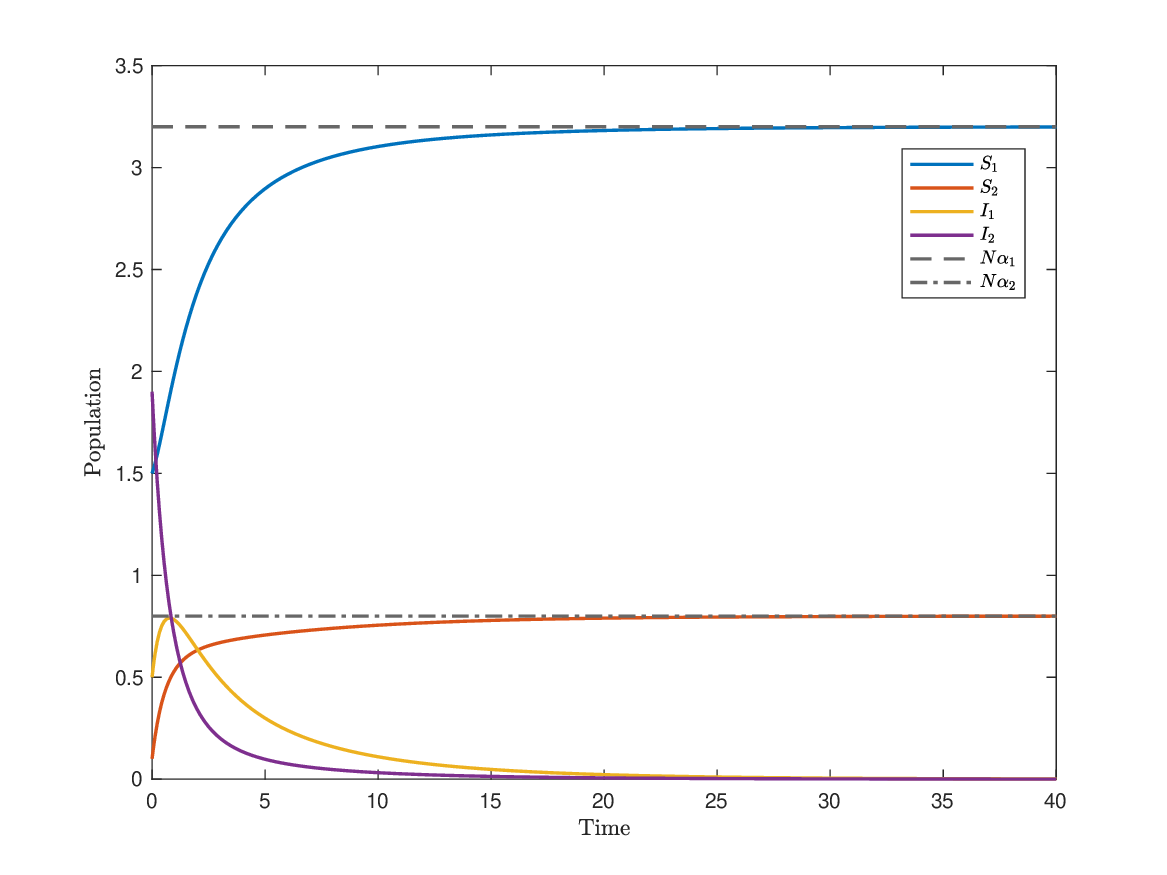}}}
\subfigure[]{\resizebox*{0.300\linewidth}{!}{\includegraphics{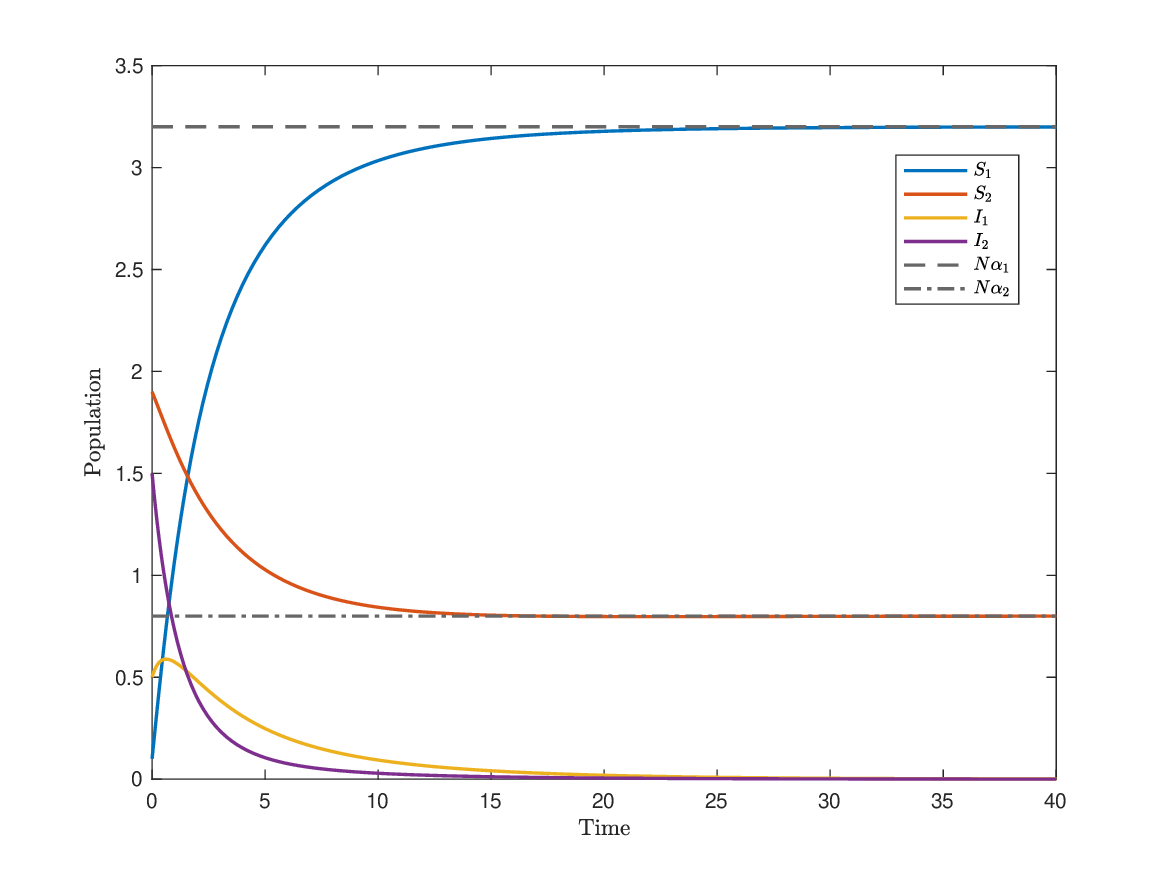}}}
\caption{Numerical simulations illustrating global dynamics of \eqref{model} when $\bm\mu=(0, 0)^T$ and ${\mathcal{R}}_0<1$ when the hypotheses of Theorem \ref{TH3}-{\rm (i)} and Remark \ref{RK2}-{\rm (ii)} are satisfied.}.
\label{Ex6}
\end{center}
\end{figure}

\medskip

\noindent{\bf Experiment 7.} 
Let $\bm\gamma=(1.5, 0.5)^{T}$, 
$\bm\beta=(3,1)^{T}$,
$\bm\zeta=(0.8,0.2)^{T}$. Hence ${\bm r}=\frac{1}{2}\bm 1\in{\rm span}(\bm 1)$ and $\bm\zeta=\bm\alpha\in{\rm span}(\bm\alpha)$, so that the hypotheses of Theorem \ref{TH3} hold. With these choices, we have that $\tau=\frac{1}{2}$ and $m=1$ in Remark \ref{RK2}, and hence $\mathcal{R}_0=N/(\tau(m+N))=\frac{8}{5}=1.6>1$. Let $d_S=0.5$ and $d_I=2$. Taking $(\bm S^{0}, \bm I^{0})=((1,1)^T,(1,1)^T)$. As time becomes larger and larger, we observe numerically that $(\bm S(t),\bm I(t))$ goes to $((2,0.5)^{T},(1.2,0.3)^{T})=(\tau(N+m)\bm \alpha,((1-\tau)N-\tau m)\bm \alpha)$, which is the unique EE solution of \eqref{model} (see Figure \ref{Ex7}(a)). Taking other initial data, we also observe that $(\bm S(t),\bm I(t))\to ((2,0.5)^{T},(1.2,0.3)^{T})$ (see Figure \ref{Ex7}(b) for $(\bm S^{0}, \bm I^{0})=((1.5,0.1)^T,(0.5,1.9)^T)$ and Figure \ref{Ex7}(c) for $(\bm S^{0}, \bm I^{0})=((0.1,1.9)^T,(0.5,1.5)^T)$), which implies that the unique EE solution is globally stable. This simulation agrees with the conclusions of Theorem \ref{TH3}-{\rm (ii)} and Remark \ref{RK2}-{\rm(iii)}.

\begin{figure}[!ht]
\begin{center}
\subfigure[]{
\resizebox*{0.300\linewidth}{!}{\includegraphics{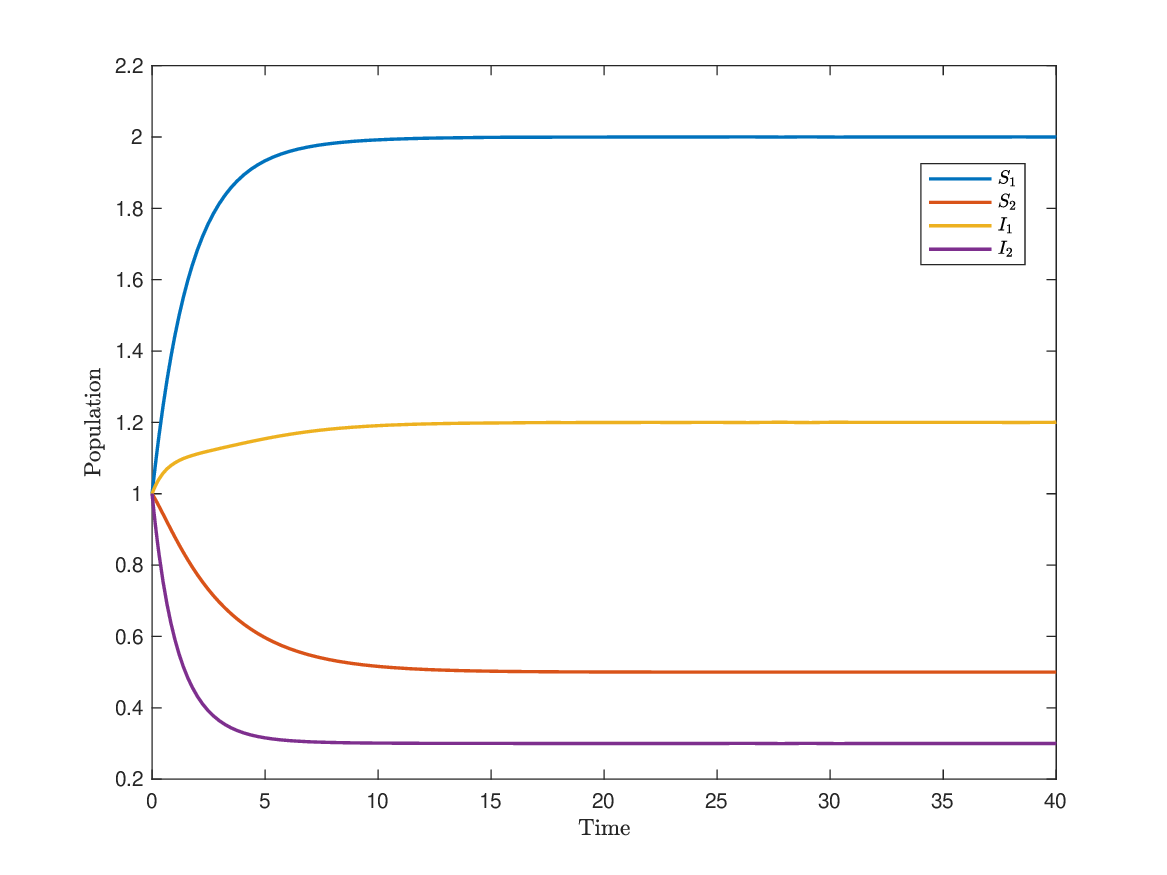}}}
\subfigure[]{\resizebox*{0.300\linewidth}{!}{\includegraphics{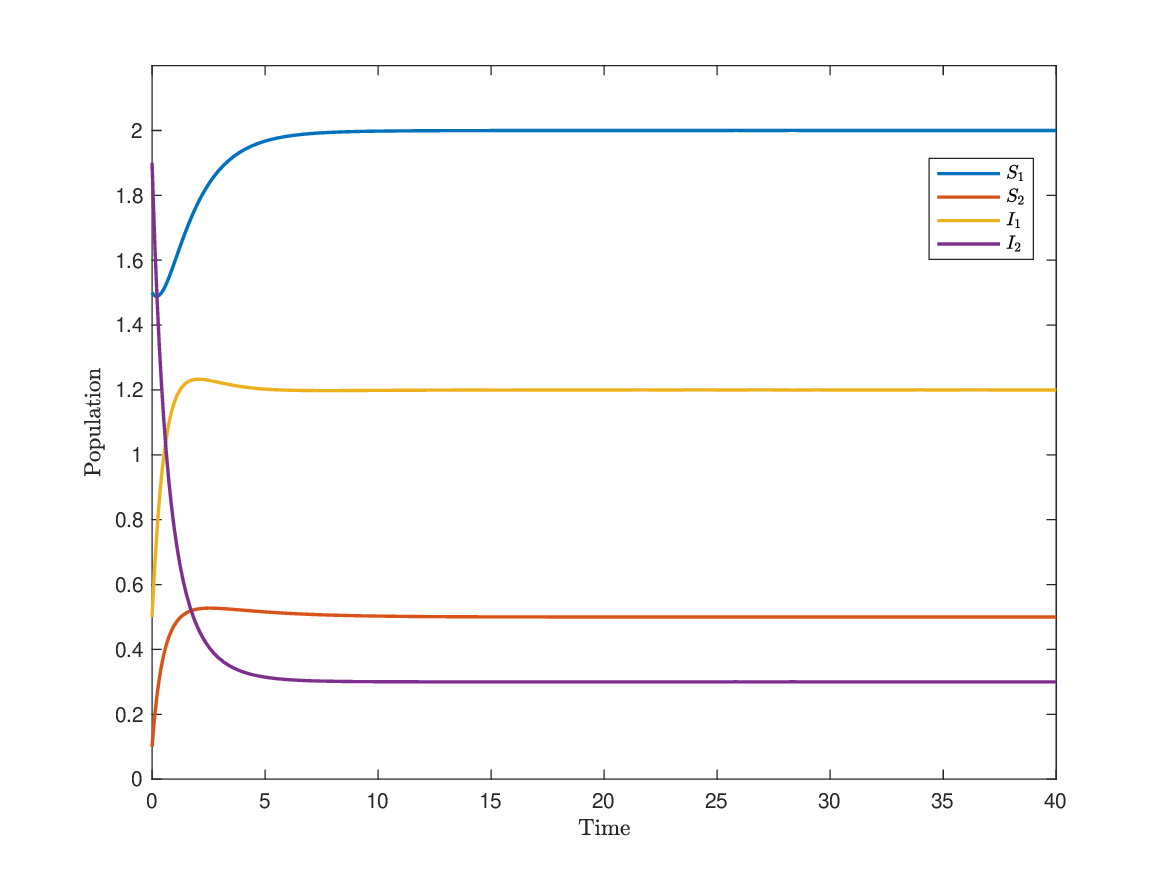}}}
\subfigure[]{\resizebox*{0.300\linewidth}{!}{\includegraphics{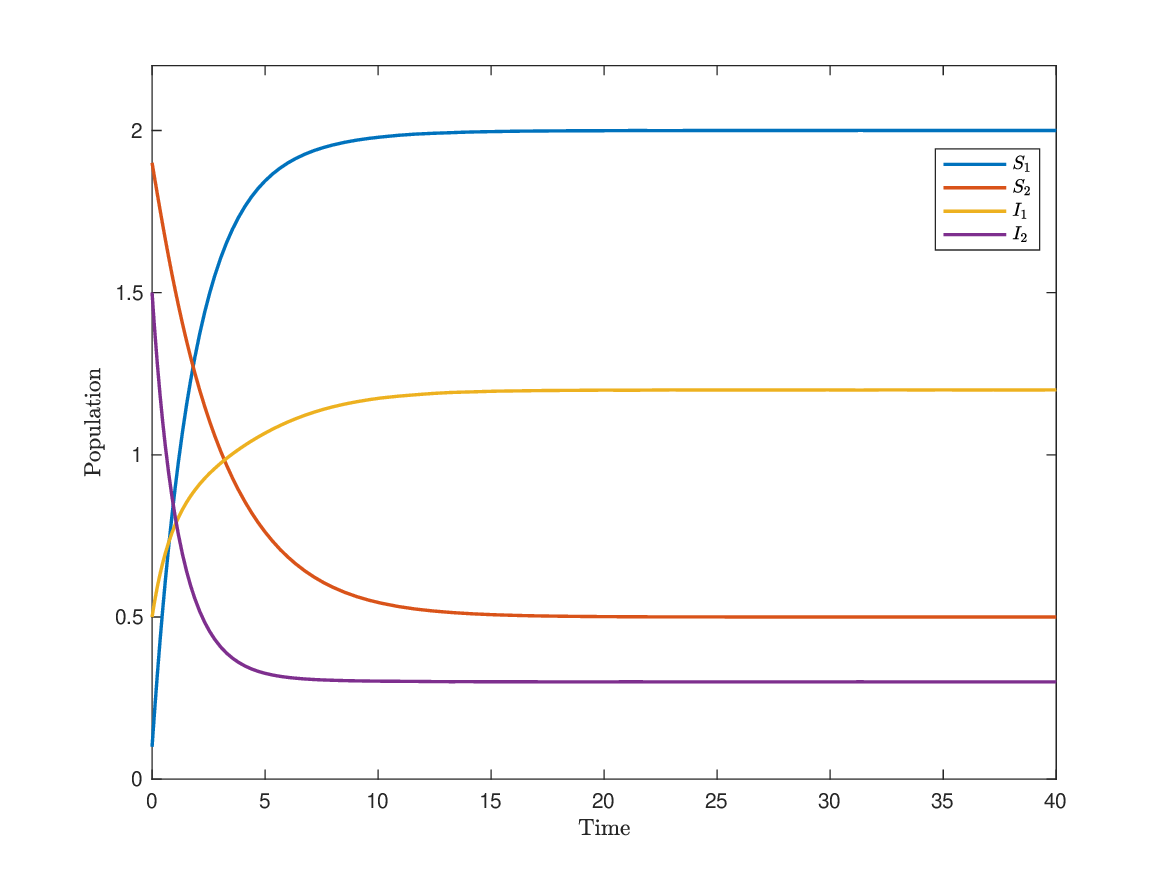}}}
\caption{Numerical simulations illustrating global dynamics of \eqref{model} when $\bm\mu=(0, 0)^T$ and ${\mathcal{R}}_0>1$ when the hypotheses of Theorem \ref{TH3}-{\rm (ii)} and Remark \ref{RK2}-{\rm(iii)} are satisfied.}
\label{Ex7}
\end{center}
\end{figure}

\medskip

\noindent{\bf Experiment 8.}  Let $\bm\gamma=(1.5, 0.5)^{T}$, 
$\bm\beta=(4.5,1)^{T}$, $\bm\zeta=(0.8,0.1)^{T}$ and $d_{S}=d_{I}=1$. With these choices, $\hat{\mathcal{R}}_0=2.9438>1$. We vary $N$ so that $\mathcal{R}_0$ varies. In Table \ref{Ex8}, we choose some values for $N$, we then obtain the corresponding values for $\mathcal{R}_0$. For each $\mathcal{R}_0$, we take the initial data $(\bm S^{0}, \bm I^{0})$ listed in Table \ref{Ex8}. We observe that when $\mathcal{R}_0\leq 1$, $(\bm S, \bm I)\to (N\bm\alpha,\bm0)$ as time becomes large, and when $\mathcal{R}_0>1$, $(\bm S, \bm I)$ goes to an EE solution of \eqref{model} as time becomes large. Moreover, when $\mathcal{R}_0>1$ is fixed and  we change the initial data, we observe the same EE solution, which indicates that the EE solution is unique. This observation is consistent with the conclusion of Theorems \ref{TH4} and  \ref{TH5}-(i). Observe that when $\mathcal{R}_0$ is close to $\hat{\mathcal{R}}_0$, both $S_1+S_2$ and $I_1+I_2$ become very large as time evolves (see the last column of Table \ref{Ex8}), which is consistent with the limiting profiles obtained in equation \eqref{Th6-eq1} of 
Theorem \ref{TH6}.
{Next, we keep $\bm\gamma$, $\bm\zeta$ and $d_{S}$ and then change $d_I$ and $\bm\beta$ to $d_{I}=2>d_{S}$ and $\bm\beta=(4.5,2)^{T}$. For this case, we have $N(\bm 1-2\bm r)\circ\bm\alpha=(\frac{16}{15},\frac{2}{5})^{T}>\bm r\circ\bm \zeta=(\frac{4}{15},\frac{1}{40})^{T}$ and $\mathcal{R}_0=2.5610>1$. Taking $(\bm S^{0}, \bm I^{0})=(1,1,1,1)^{T}$, we observe that as time becomes large, the solution goes to an EE solution $(1.329,0.245,1.906,0.520)^{T}$(see Figure \ref{Ex8-2}(a)). Taking other initial data, we observe the same EE solution (see Figure \ref{Ex8-2}(b) for $(\bm S^{0}, \bm I^{0})=(1.5,0.1,0.5,1.9)^{T}$ and Figure \ref{Ex8-2}(c) for $(\bm S^{0}, \bm I^{0})=(0.1,1.9,0.5,1.5)^{T}$), which indicates that the EE solution is unique. This simulation is consistent with Theorem \ref{TH5}(\rm ii).}
Finally, we simulate the existence of EE solutions when $\mathcal{R}_0<1$. let $L_{12}=L_{21}=0.5$, $\bm\zeta=(1,1)^{T}$, $\bm\beta=(2,4)^{T}$, $\bm\gamma=(1,3)^{T}$ so that the hypotheses of Proposition \ref{prop3} are satisfied. Let $d_{I}=4$ and $d_{S}=0.001$. Table \ref{Ex8-T} shows the EE solution as $\mathcal{R}_0$ varies from 0.9991 to 0.9955. When $\mathcal{R}_0<0.9955$, there is no EE solution. In the latter case, our  numerical solutions indicate that 
solutions converge to the DFE. 

\medskip

\begin{table}[!ht]
\centering
\begin{tabular}{|c| c| c |c |c| c |c|}
 \hline
 $N$ & 0.1 &0.45&0.5&0.55&1&$10^{6}$\\
 \hline
$\mathcal{R}_0$ &0.2788&0.9323& 1&1.0632&1.4886&2.943846\\
 \hline
$(\bm S_{0}, \bm I_{0})^T$ &$\begin{pmatrix} 0.025\\0.025\\0.025\\0.025\end{pmatrix}$&$\begin{pmatrix} 
0.15\\0.1\\0.1\\0.1
\end{pmatrix}$&$\begin{pmatrix} 
0.2\\0.1\\0.1\\0.1
\end{pmatrix}$&$\begin{pmatrix} 
0.25\\0.1\\0.1\\0.1
\end{pmatrix}$&$\begin{pmatrix} 
0.25\\0.25\\0.25\\0.25
\end{pmatrix}$&$\begin{pmatrix} 
$$2.5\times10^{5}$$\\$$2.5\times10^{5}$$\\$$2.5\times10^{5}$$\\$$2.5\times10^{5}$$
\end{pmatrix}$\\  
\hline
 DFE &$\begin{pmatrix} 0.08\\0.02\\0\\0\end{pmatrix}$&$\begin{pmatrix} 
0.36\\0.09\\ 0\\ 0
\end{pmatrix}$&$\begin{pmatrix} 
0.4\\0.1\\0\\0
\end{pmatrix}$&$\begin{pmatrix} 
0.44\\0.11\\0\\0
\end{pmatrix}$&$\begin{pmatrix} 
0.8\\ 0.2\\0\\0
\end{pmatrix}$&$\begin{pmatrix} 
$$8\times10^{5}$$\\$$2\times10^{5}$$\\0\\0
\end{pmatrix}$\\
\hline
 EE & None & None & None &$\begin{pmatrix} 
0.4138\\0.1036\\0.0262\\0.0064
\end{pmatrix}$&$\begin{pmatrix} 
0.5362\\0.1394\\0.2638\\0.0606
\end{pmatrix}$&$\begin{pmatrix} 
$$2.7\times10^{5}$$\\$$0.9\times10^{5}$$\\$$5.3\times10^{5}$$\\$$1.1\times10^{5}$$
\end{pmatrix}$\\
\hline
\end{tabular}
\vspace{\abovecaptionskip}
\caption{Numerical calculation of ${\mathcal{R}_0}$,  DFE and EE}
\label{Ex8}
\end{table}

\begin{table}[!ht]
\centering
\begin{tabular}{|c| c| c |c |c| c |}
 \hline
 $N$ & 3.82 & 3.81&3.80&3.79&3.78\\
 \hline
$\mathcal{R}_0$ &0.9991&0.9982& 0.9973&0.9964&0.9955\\
 \hline
 EE & $\begin{pmatrix} 
1.3440\\2.4685\\0.0039\\0.0036
\end{pmatrix}$ & $\begin{pmatrix} 
1.3727\\2.4308\\0.0034\\0.0031
\end{pmatrix}$ & $\begin{pmatrix} 
1.4077\\2.3867\\0.0029\\0.0027
\end{pmatrix}$ &$\begin{pmatrix} 
1.4546\\2.3309\\0.0024\\0.0021
\end{pmatrix}$&$\begin{pmatrix} 
1.5455\\2.2315\\0.0016\\0.0014
\end{pmatrix}$\\
\hline
\end{tabular}
\vspace{\abovecaptionskip}
\caption{Numerical calculation of EE when ${\mathcal{R}_0<1}$}
\label{Ex8-T}
\end{table}



\begin{figure}[!ht]
\begin{center}
\subfigure[]{
\resizebox*{0.300\linewidth}{!}{\includegraphics{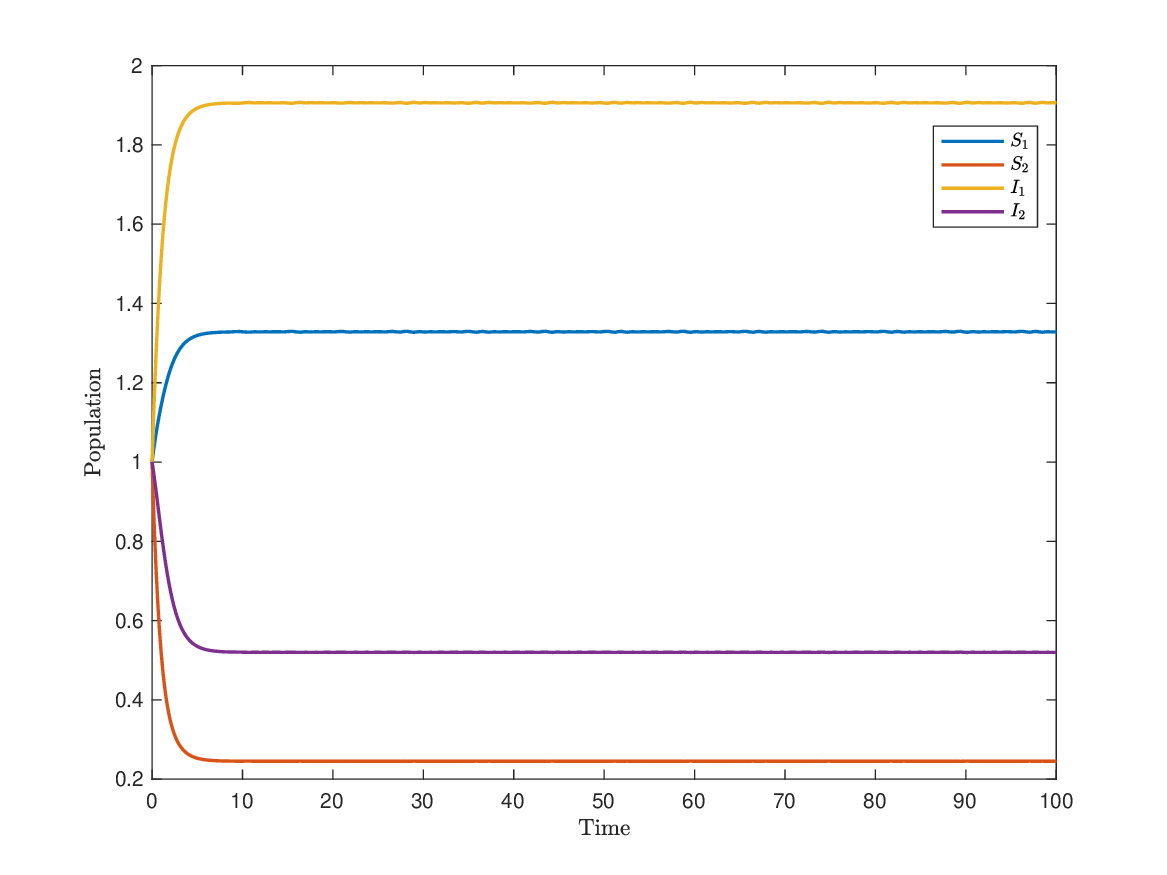}}}
\subfigure[]{\resizebox*{0.300\linewidth}{!}{\includegraphics{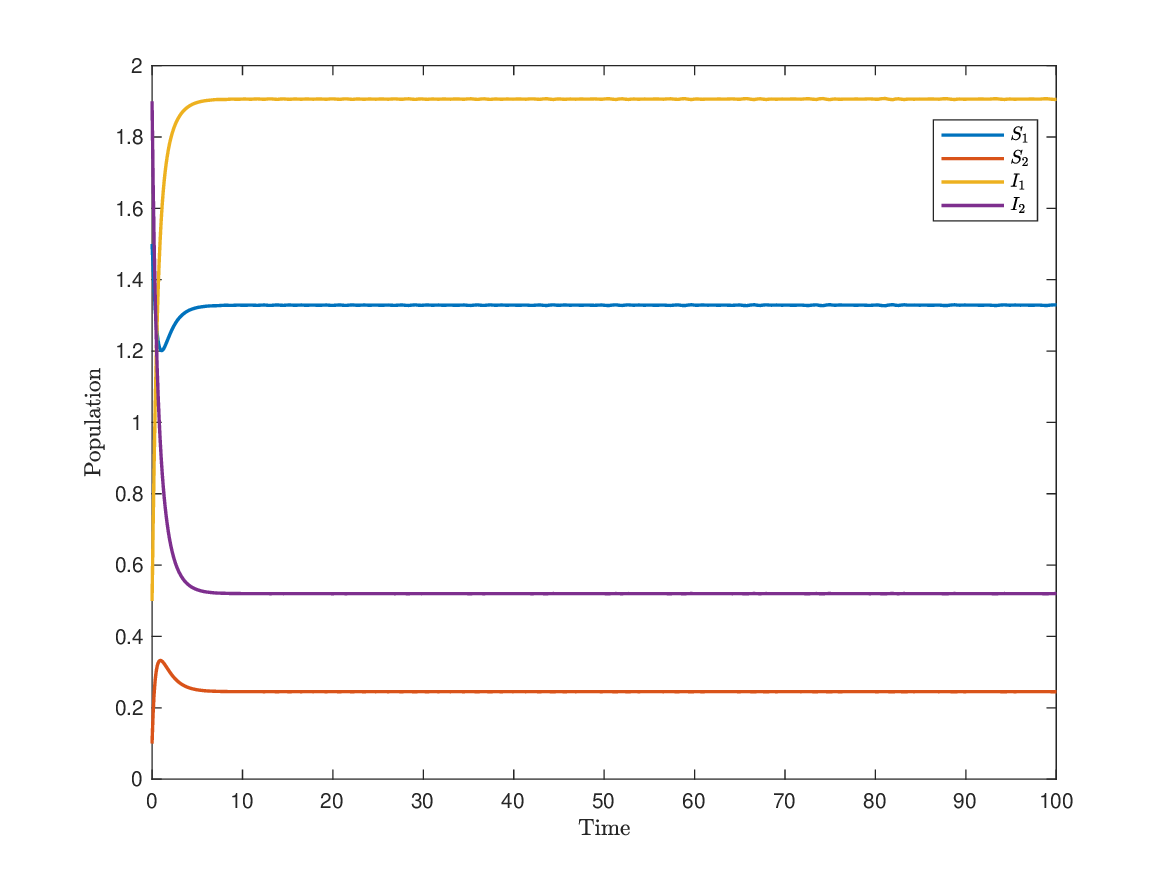}}}
\subfigure[]{\resizebox*{0.300\linewidth}{!}{\includegraphics{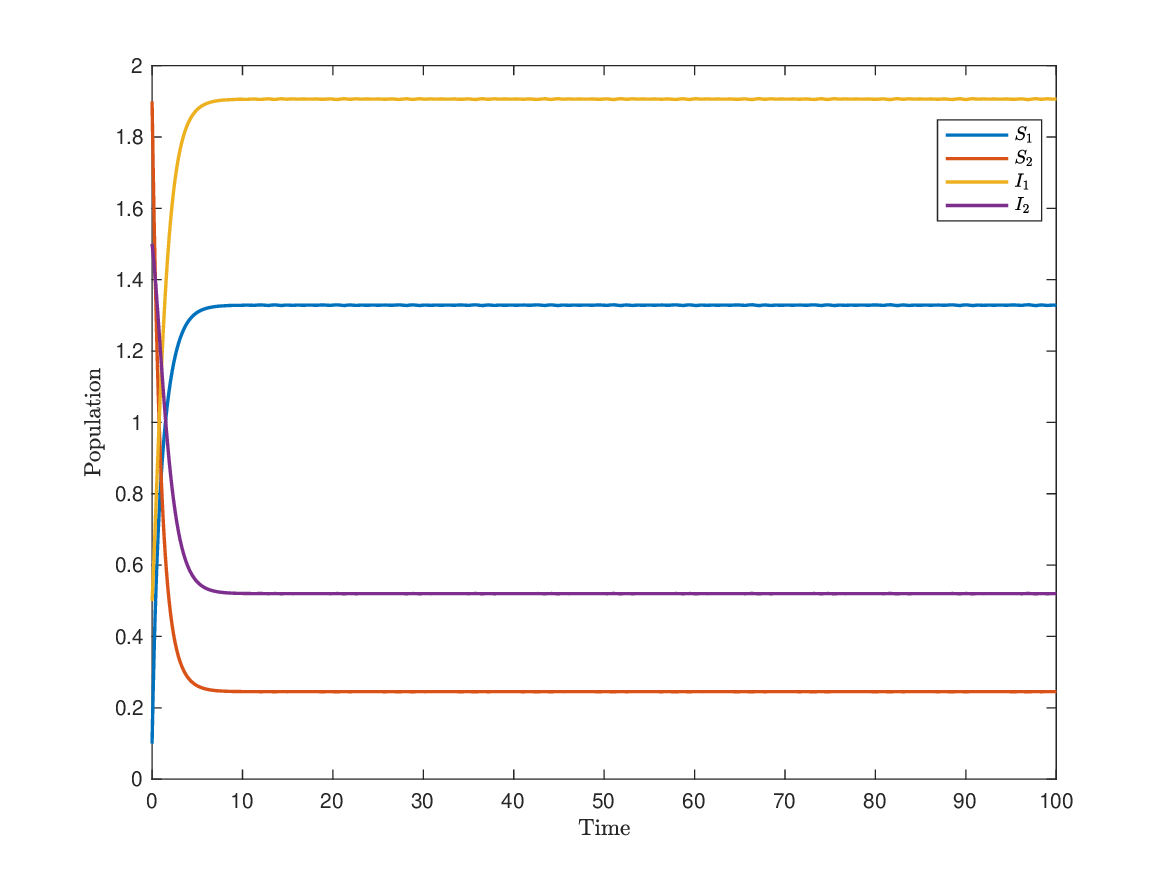}}}
\caption{Numerical simulations illustrating global dynamics of \eqref{model} when $\bm\mu=(0, 0)^T$ and ${\mathcal{R}}_0>1$ when the hypotheses of Theorem \ref{TH5}-{\rm (ii)} is satisfied.}
\label{Ex8-2}
\end{center}
\end{figure}

\medskip

\noindent{\bf Experiment 9.}
Let $\bm\gamma=(1.5, 0.5)^{T}$, 
$\bm\beta=(14,1)^{T}$,
$\bm\zeta=(0.8,0.1)^{T}$ and $d_{I}=1$.
Then $\bm r=(\frac{3}{28}, \frac{1}{2})^{T}$ and $\|\bm\zeta\circ {\bm r}/(\bm 1-{\bm r})\|_1=0.196$. So $\bm r_M=0.5<1$. Take $N=0.16<0.196$. We have $\mathcal{R}_0=1.2512>1$. 
For $d_{S}=10^{-1}$, 
we observe that there is an EE solution $(\bm S,\bm I)=(0.1006,0.04,0.0167,0.0027)^{T}$ (see Figure \ref{Ex9}(a)). As $d_{S}$ becomes smaller and smaller, we observe that the $\bm I$-component of EE goes to $(0,0)^{T}$ and the $\bm S$-component of EE goes to $(0.096,0.064)^{T}$ (see Figure \ref{Ex9}(c)). So we have $\|\bm I\|_1\to 0$ and $\|\bm S\|_1\to N$  as $d_S\to 0^+$, which is consistent with Theorem \ref{TH7}(i). In addition, 
we also simulate $(\bm S, \frac{1}{d_{S}}\bm I)$ and 
observe that $(\bm S, \frac{1}{d_{S}}\bm I)$ approaches $(0.0976,0.0624,0.6765,0.1312)^{T}\approx(l^*(\bm\alpha-d_I\bm P^*),l^*\bm P^*)$ where $l^{*}=0.9677$, $\bm P^*=(0.6991, 0.1356)^{T}$. This simulation is consistent with Theorem \ref{TH7}(i-1).

\begin{figure}[!ht]
\begin{center}
\subfigure[]{
\resizebox*{0.300\linewidth}{!}{\includegraphics{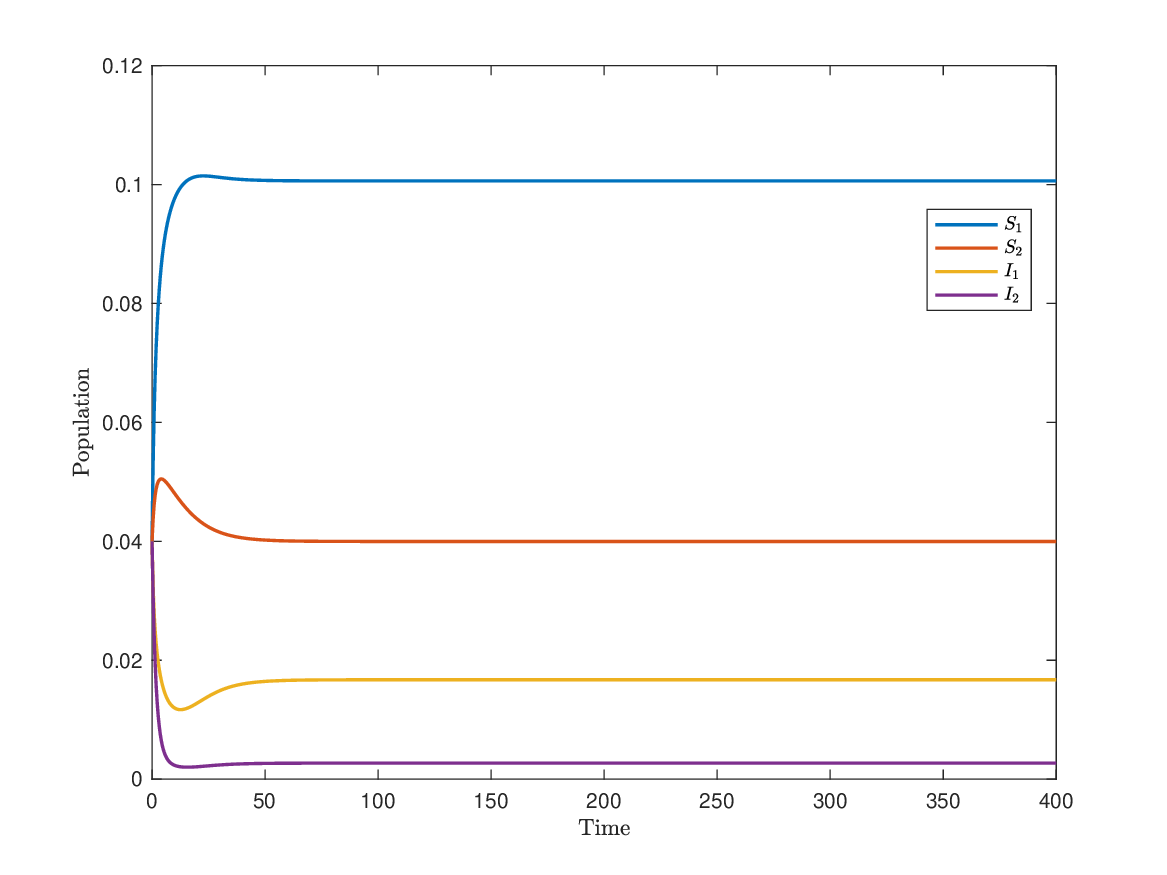}}}
\subfigure[]{\resizebox*{0.300\linewidth}{!}{\includegraphics{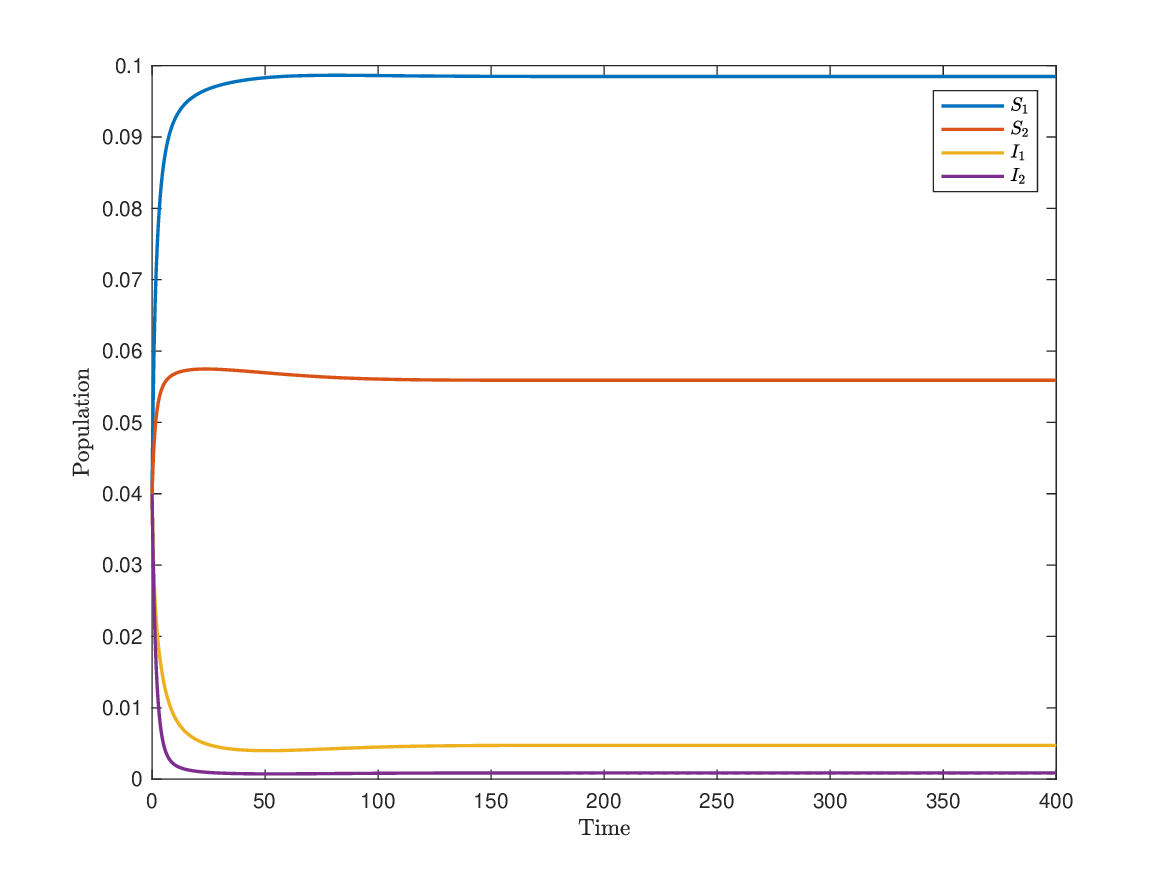}}}
\subfigure[]{\resizebox*{0.300\linewidth}{!}{\includegraphics{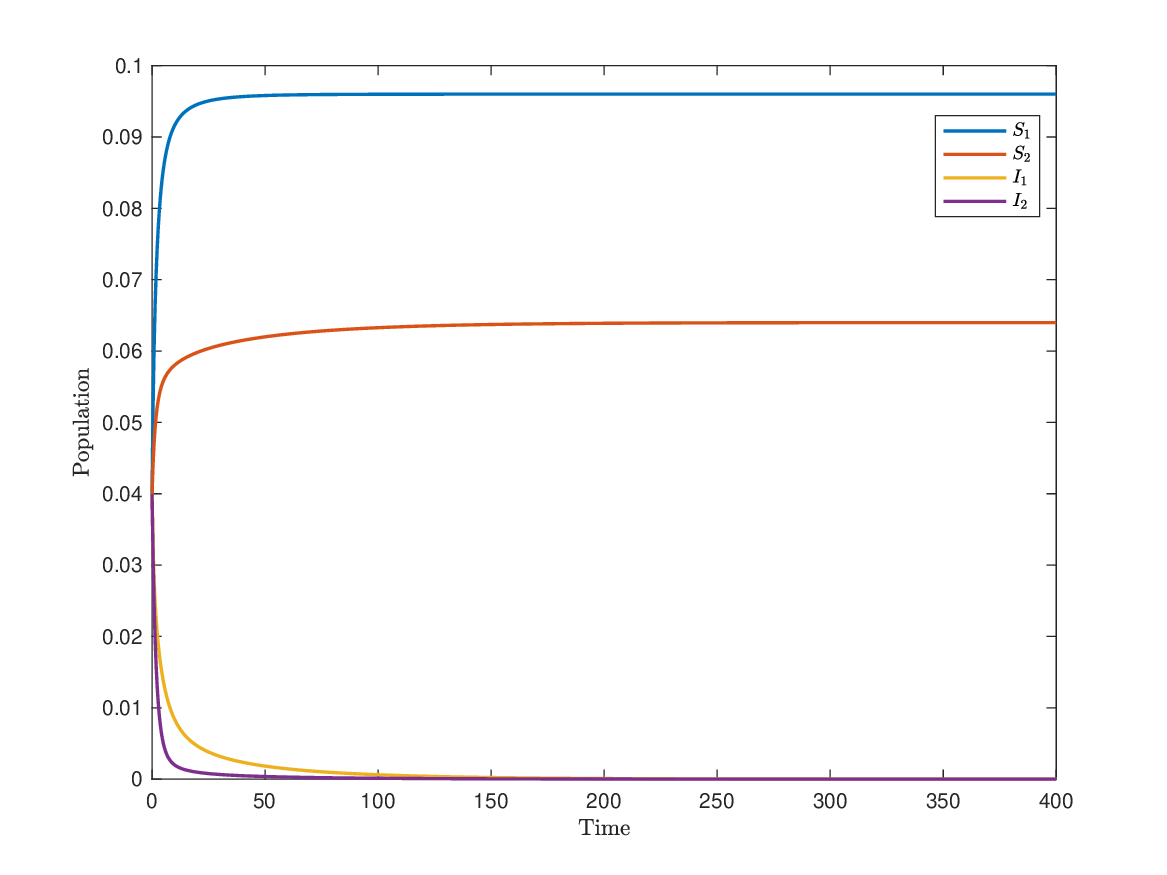}}}
\caption{Asymptotic profiles of EEs of \eqref{model} when $\bm r_M<1$ and $N<\|\bm\zeta\circ {\bm r}/(\bm 1-{\bm r})\|_1$: (a) $d_{S}=10^{-1}$, (b) $d_{S}=10^{-2}$, (c) $d_{S}=10^{-6}$.} 
\label{Ex9}
\end{center}
\end{figure}

\medskip

\noindent{\bf Experiment 10.}
Let $\bm\gamma$, $\bm\beta$, $\bm\zeta$, and $d_{I}$ be the same as in Experiment 9. Take $N=0.196$, then $\mathcal{R}_0=1.4862>1$. For this parameter setting, we have $\bm r_M<1$ and $N=\|\bm\zeta\circ {\bm r}/(\bm 1-{\bm r})\|_1$. For $d_{S}=10^{-1}$, 
Figure \ref{Ex10}(a) shows that there is an EE solution $(\bm S,\bm I)=(0.1025,0.0520,0.0353,0.0062)^{T}$.
As $d_{S}$ decreases, we observe that 
the EE solution $(\bm S,\bm I)\to ((0.096, 0.1)^{T}, (0,0)^{T})=(\bm\zeta\circ\bm r/(\bm 1-\bm r),\bm 0)$ (see Figure \ref{Ex10}(c)), which agrees with Theorem \ref{TH7}(i-2).

\begin{figure}[!ht]
\begin{center}
\subfigure[]{
\resizebox*{0.325\linewidth}{!}{\includegraphics{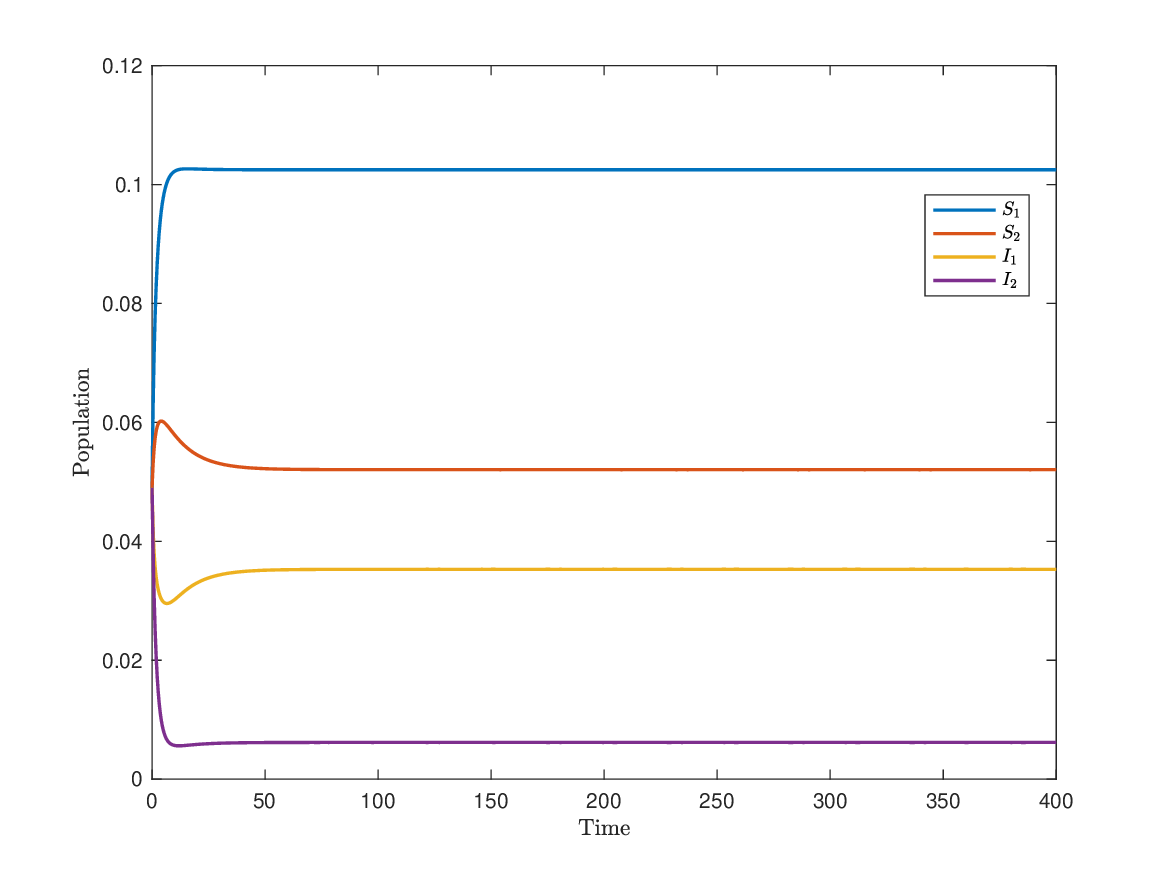}}}
\subfigure[]{\resizebox*{0.325\linewidth}{!}{\includegraphics{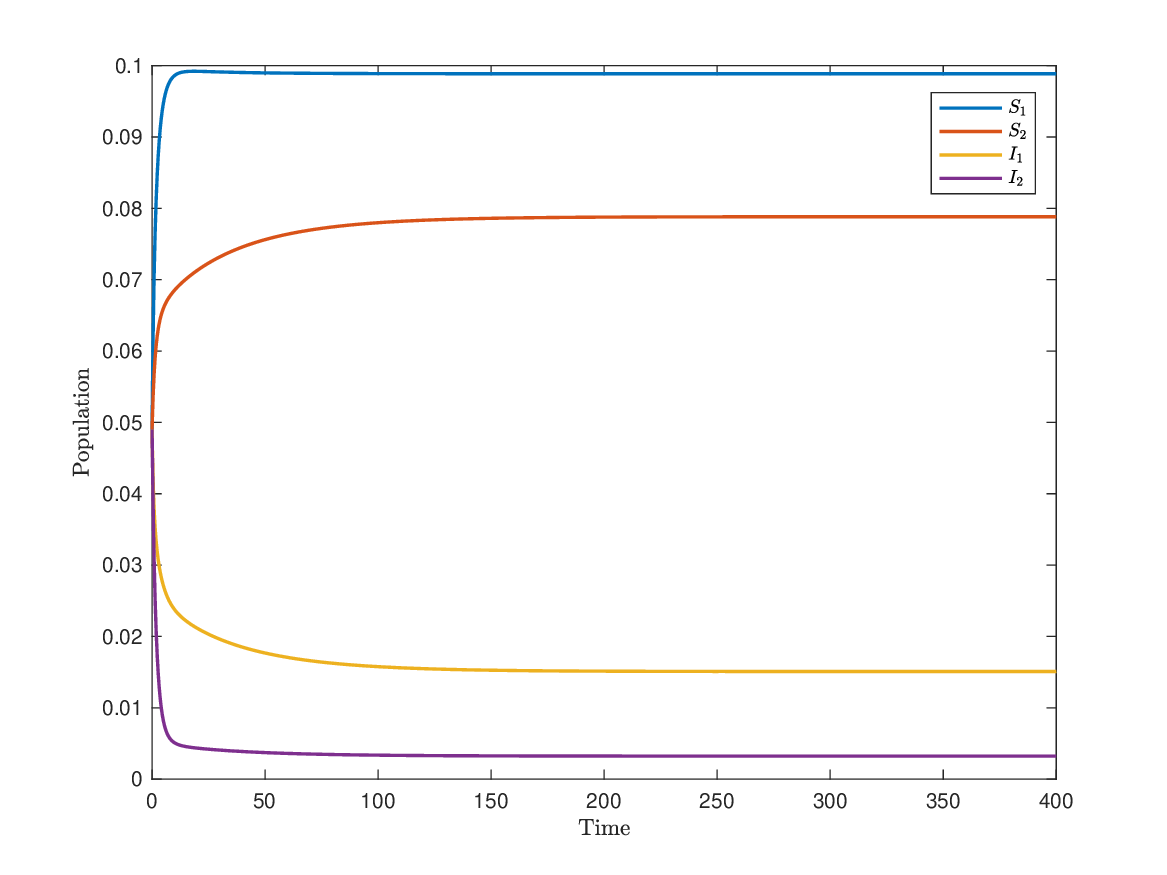}}}
\subfigure[]{\resizebox*{0.325\linewidth}{!}{\includegraphics{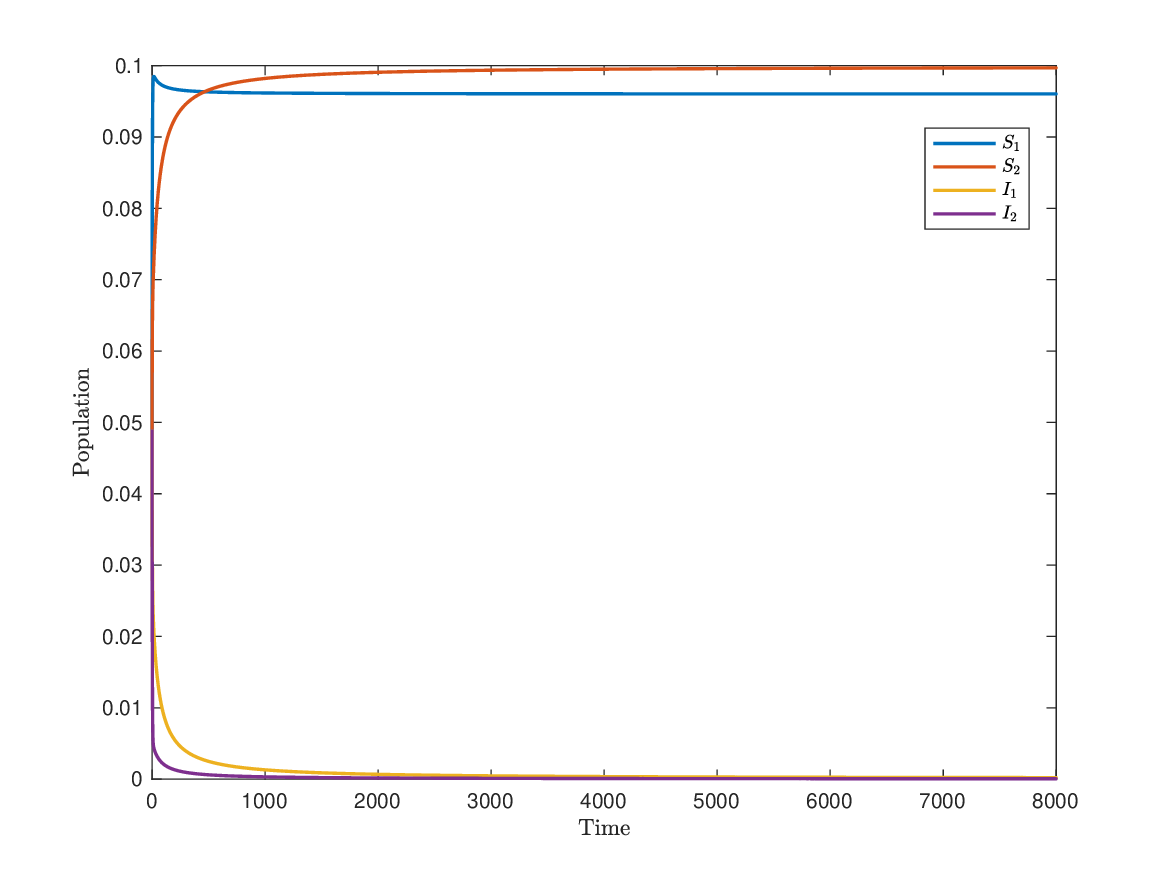}}}
\caption{Asymptotic profiles of EEs of \eqref{model} when $\bm r_M<1$ and $N=\|\bm\zeta\circ {\bm r}/(\bm 1-{\bm r})\|_1$: (a) $d_{S}=10^{-1}$, (b) $d_{S}=10^{-2}$, (c) $d_{S}=10^{-6}$,} 
\label{Ex10}
\end{center}
\end{figure}

\medskip

\noindent{\bf Experiment 11.}
Let $\bm\gamma$, $\bm\beta$, $\bm\zeta$, and $d_{I}$ be the same as in Experiment 9. Take $N=4$, then $\mathcal{R}_0=7.2327>1$. For this parameter setting, we have $\bm r_M<1$ and $N>\|\bm\zeta\circ {\bm r}/(\bm 1-{\bm r})\|_1$. Let $d_{S}=10^{-1}$, 
we observe that there is an EE solution $(0.3935,0.5813,2.4588,0.5644)^{T}$ (see Figure \ref{Ex11}(a)). We then decrease the value of $d_{S}$, we observe that the EE solution goes to $(\bm S^{*}, \bm I^{*})=((0.3778,0.6870)^T,(2.3481,0.5870)^{T})$ (see Figure \ref{Ex11}(c)), where $\bm S^{*}$ and $\bm I^{*}$ are given by \eqref{Eq2-TH7}. This simulation is consistent with Theorem \ref{TH7}(ii).

\begin{figure}[!ht]
\begin{center}
\subfigure[]{
\resizebox*{0.325\linewidth}{!}{\includegraphics{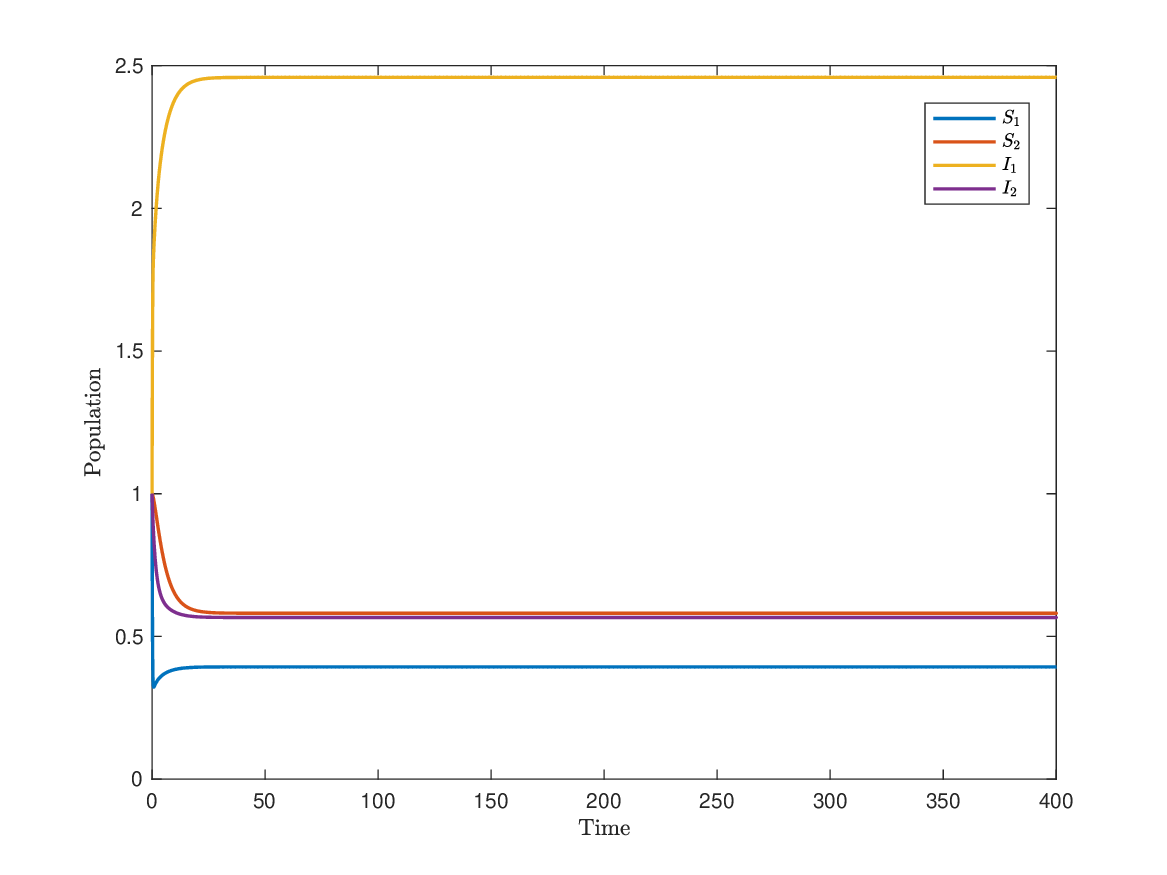}}}
\subfigure[]{\resizebox*{0.325\linewidth}{!}{\includegraphics{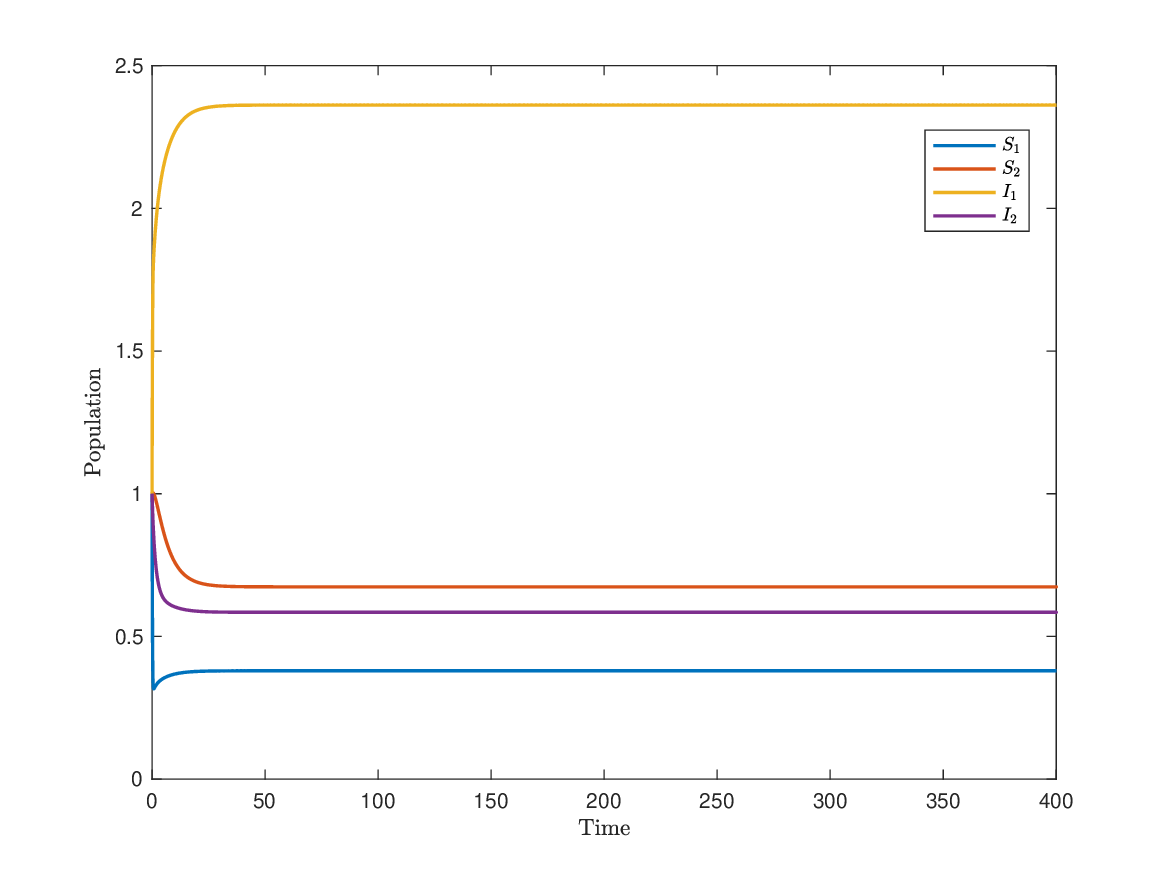}}}
\subfigure[]{\resizebox*{0.325\linewidth}{!}{\includegraphics{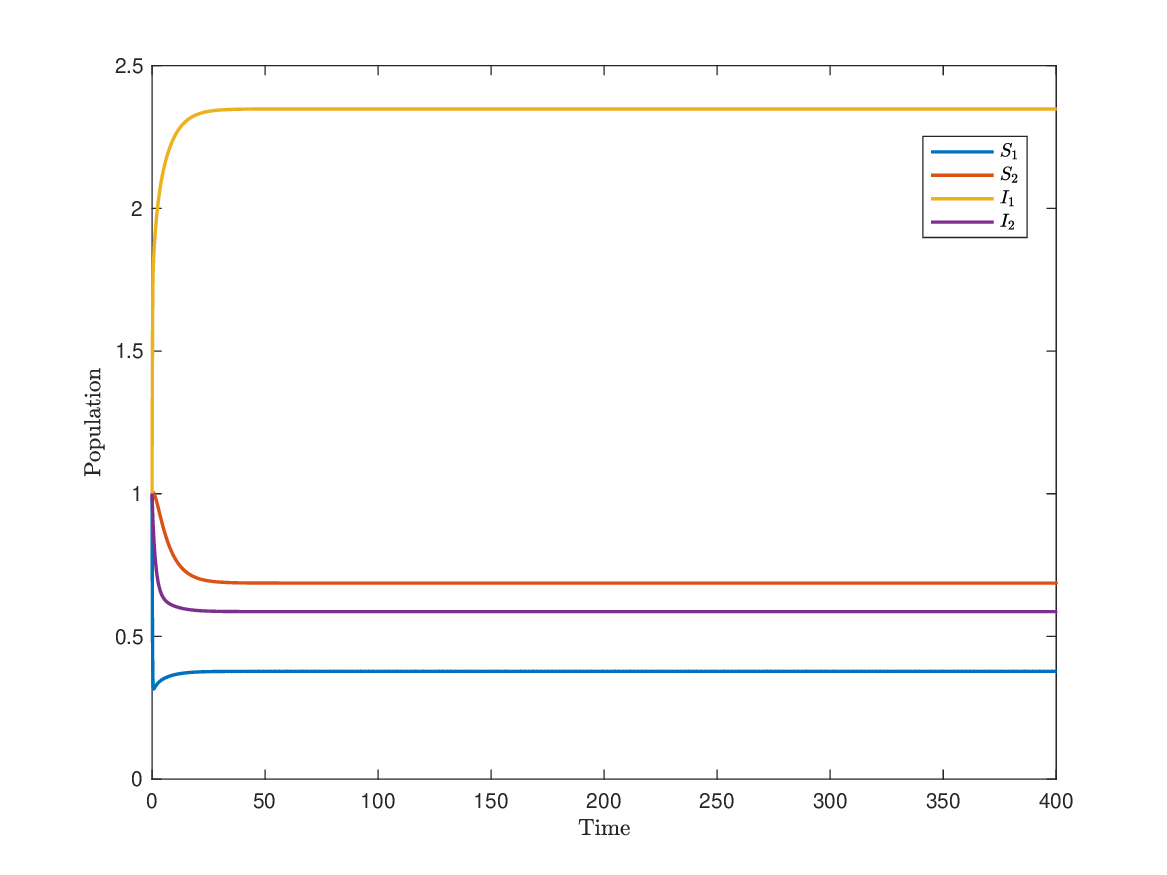}}}
\caption{Asymptotic profiles of EEs of \eqref{model} when $\bm r_M<1$ and $N>\|\bm\zeta\circ {\bm r}/(\bm 1-{\bm r})\|_1$: (a) $d_{S}=10^{-1}$, (b) $d_{S}=10^{-2}$, (c) $d_{S}=10^{-6}$} 
\label{Ex11}
\end{center}
\end{figure}

\medskip

\noindent{\bf Experiment 12.}
Let $\bm\gamma$, $\bm\beta$, and $d_{I}$ be the same as in Experiment 9. Take $\bm\zeta=(0.5,1)^{T}$, $N=4$, then $\mathcal{R}_0=1.0253>1$. For this parameter setting, we have $\bm r_M=3>1$.
Let $d_{S}=10^{-1}$, we observe that there is an EE solution $(3.2668,0.7190,0.0036,0.0106)^{T}$ (see Figure \ref{Ex12}(a)). As $d_{S}$ decreases, we observe that the EE solution goes to $(3.3504,0.6496,0,0)^{T}$ (see Figure \ref{Ex12}(c)). So we have $\|\bm I\|_1\to 0$ and $\|\bm S\|_1\to N$  as $d_S\to 0^+$, which is consistent with Theorem \ref{TH7}(i). In addition, 
we also simulate $(\bm S, \frac{1}{d_{S}}\bm I)$ and  
observe that $(\bm S, \frac{1}{d_{S}}\bm I)$ approaches $(3.3504,0.6496,0.0686,0.2052)^{T}\approx(l^*(\bm\alpha-d_I\bm P^*),l^*\bm P^*)$ where $l^{*}=4.2723$, $\bm P^*=(0.0160, 0.0478)^{T}$. This simulation is consistent with Theorem \ref{TH7}(i-1).

\begin{figure}[!ht]
\begin{center}
\subfigure[]{
\resizebox*{0.325\linewidth}{!}{\includegraphics{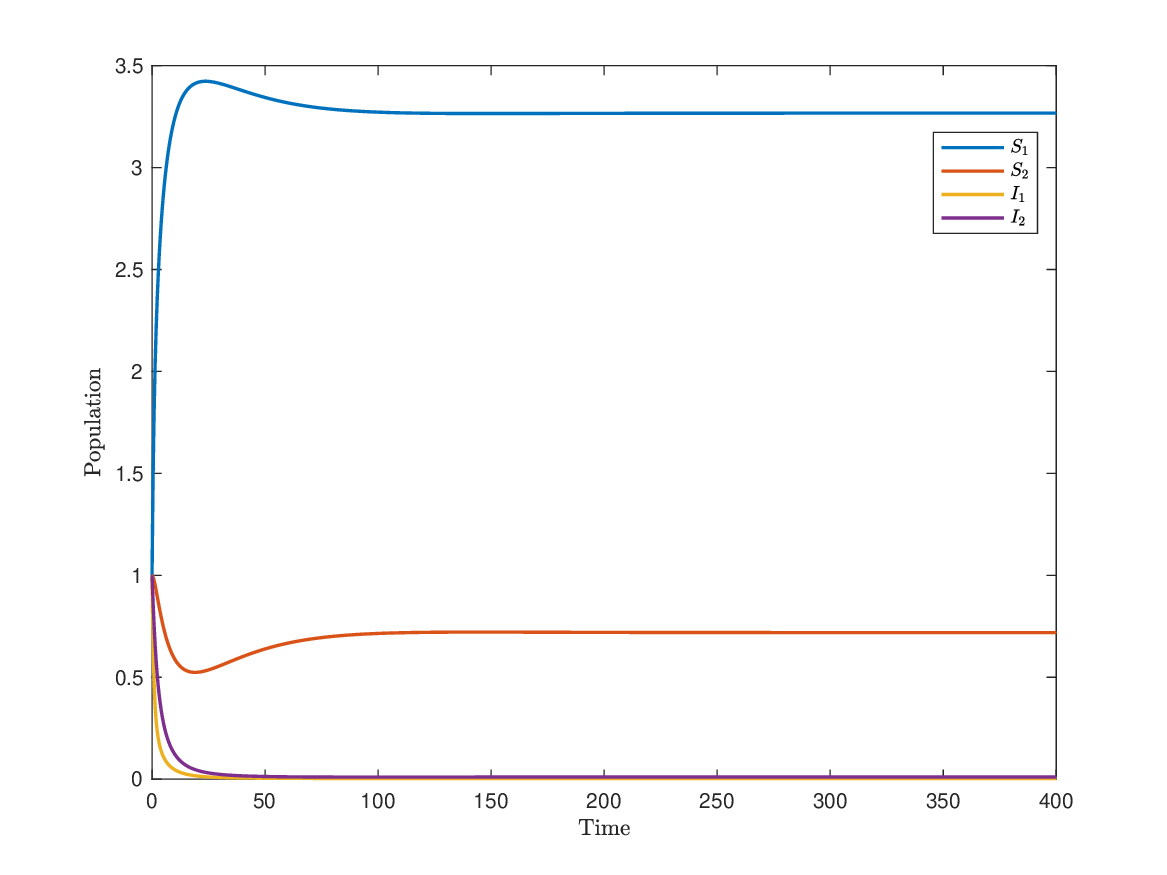}}}
\subfigure[]{\resizebox*{0.325\linewidth}{!}{\includegraphics{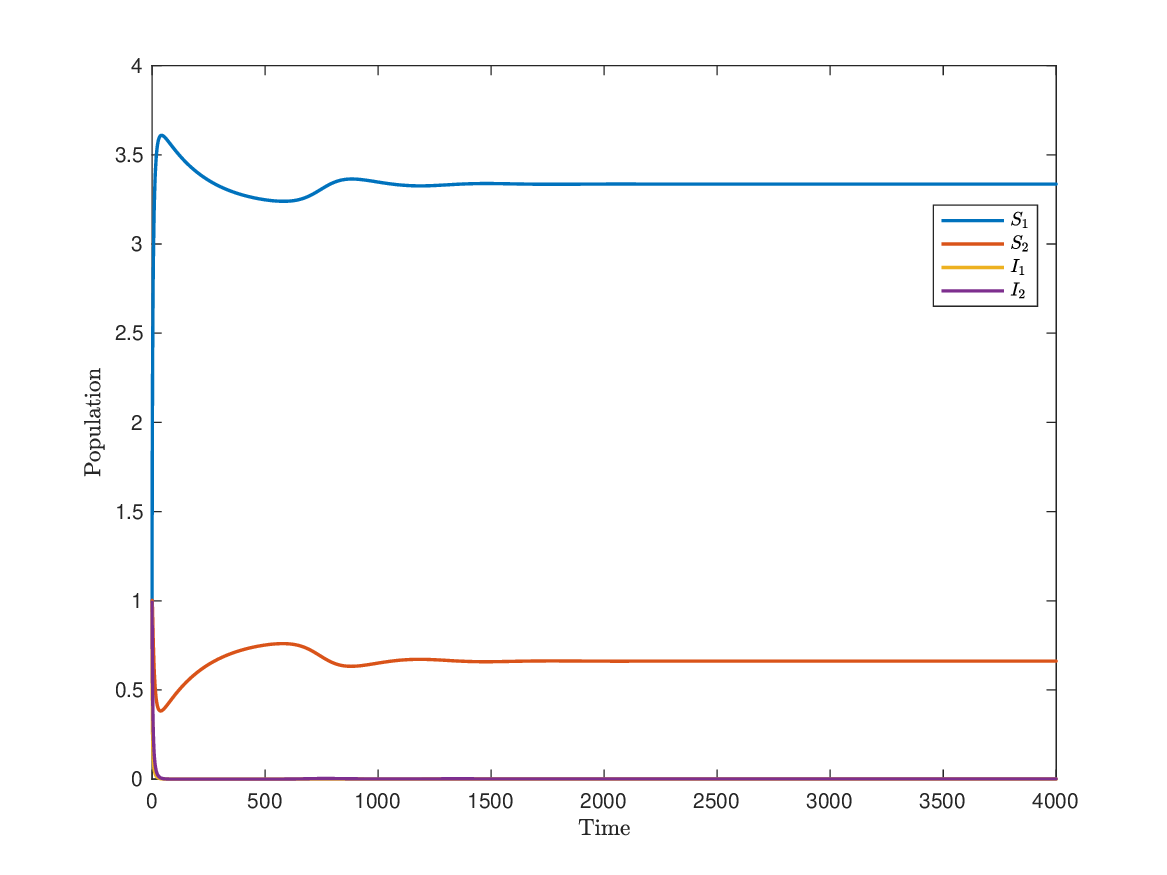}}}
\subfigure[]{\resizebox*{0.325\linewidth}{!}{\includegraphics{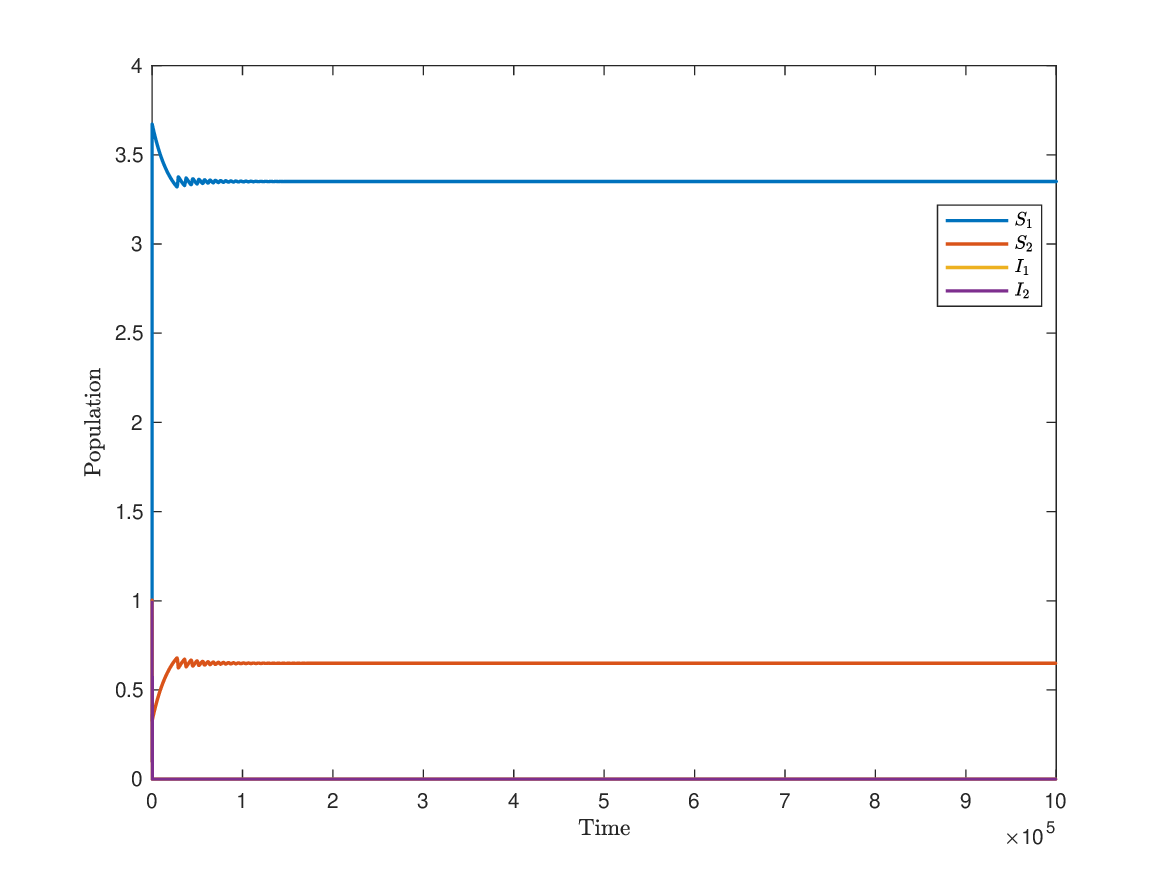}}}
\caption{Asymptotic profiles of EEs of \eqref{model} when $\bm r_M>1$: (a) $d_{S}=10^{-1}$, (b) $d_{S}=10^{-2}$, (c) $d_{S}=10^{-4}$} 
\label{Ex12}
\end{center}
\end{figure}

\medskip

\noindent{\bf Experiment 13.}
Let $\bm\gamma=(1.5,0.5)^{T}$, $\bm\beta=(0.5,1)^{T}$, $\bm\zeta=(0.8,0.1)^{T}$ and $d_{S}=1$. Take $N=4$, then 
$N\bm\alpha/(\bm r\circ(\bm\zeta+N\bm\alpha))=(\frac{4}{15}, \frac{16}{9})^{T}$. So
$\|N\bm\alpha/(\bm r\circ(\bm\zeta+N\bm\alpha))\|_{\infty}=\frac{16}{9}>1$. With these choices, the 
hypotheses of Theorem \ref{TH8} holds. We then choose a set of $d_{I}$. We observe that for every $0<d_{I}\leq0.1$, there is a unique EE solution $(\bm S,\bm I)$ (see Figure \ref{Ex13}). Moreover, as $d_{I}$ becomes smaller and smaller, $(S_{1}, I_{1})\to(2.7333,0)^{T}$ and $(S_{2}, I_{2})\to(0.6833,0.5833)^{T}$, which agrees with Theorem \ref{TH8} with $N^{*}\approx 3.4166$ in \eqref{N-star}.

\begin{figure}[!ht]
\begin{center}
\subfigure[]{
\resizebox*{0.325\linewidth}{!}{\includegraphics{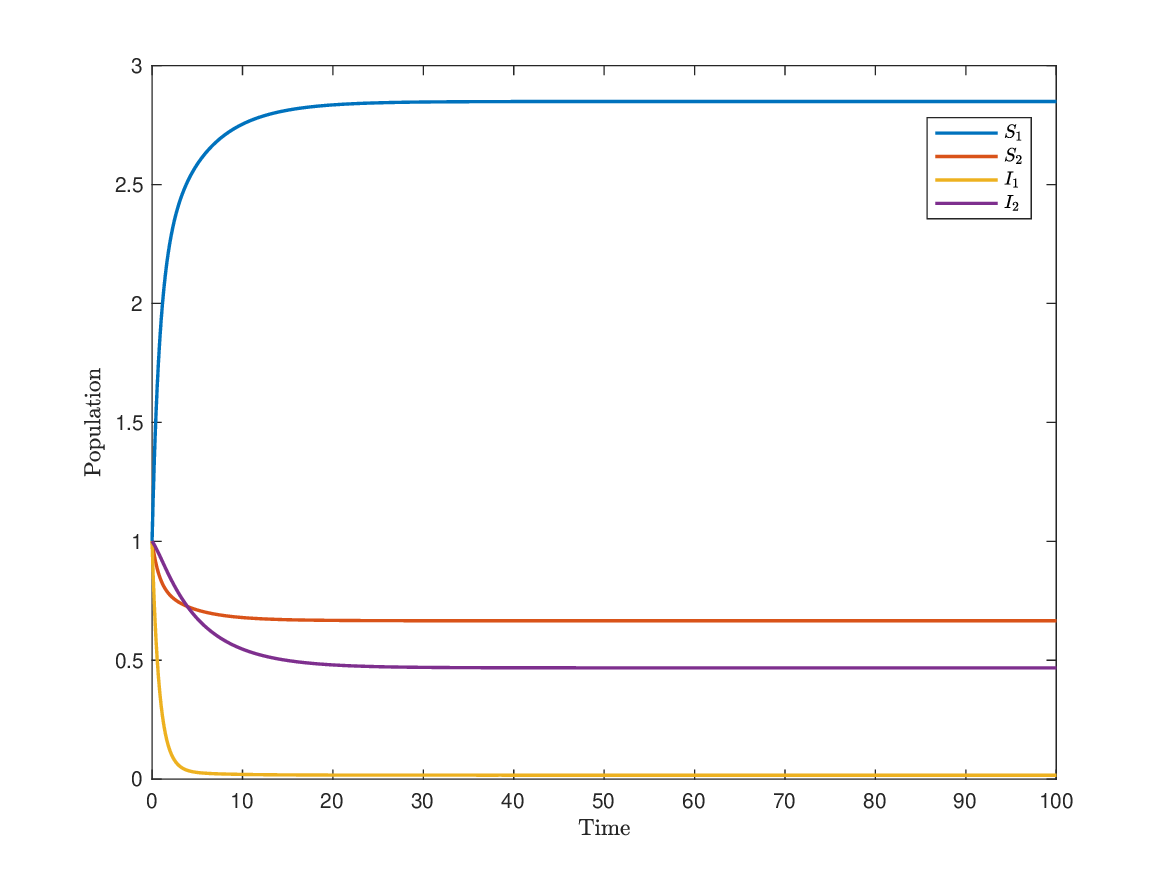}}}
\subfigure[]{\resizebox*{0.325\linewidth}{!}{\includegraphics{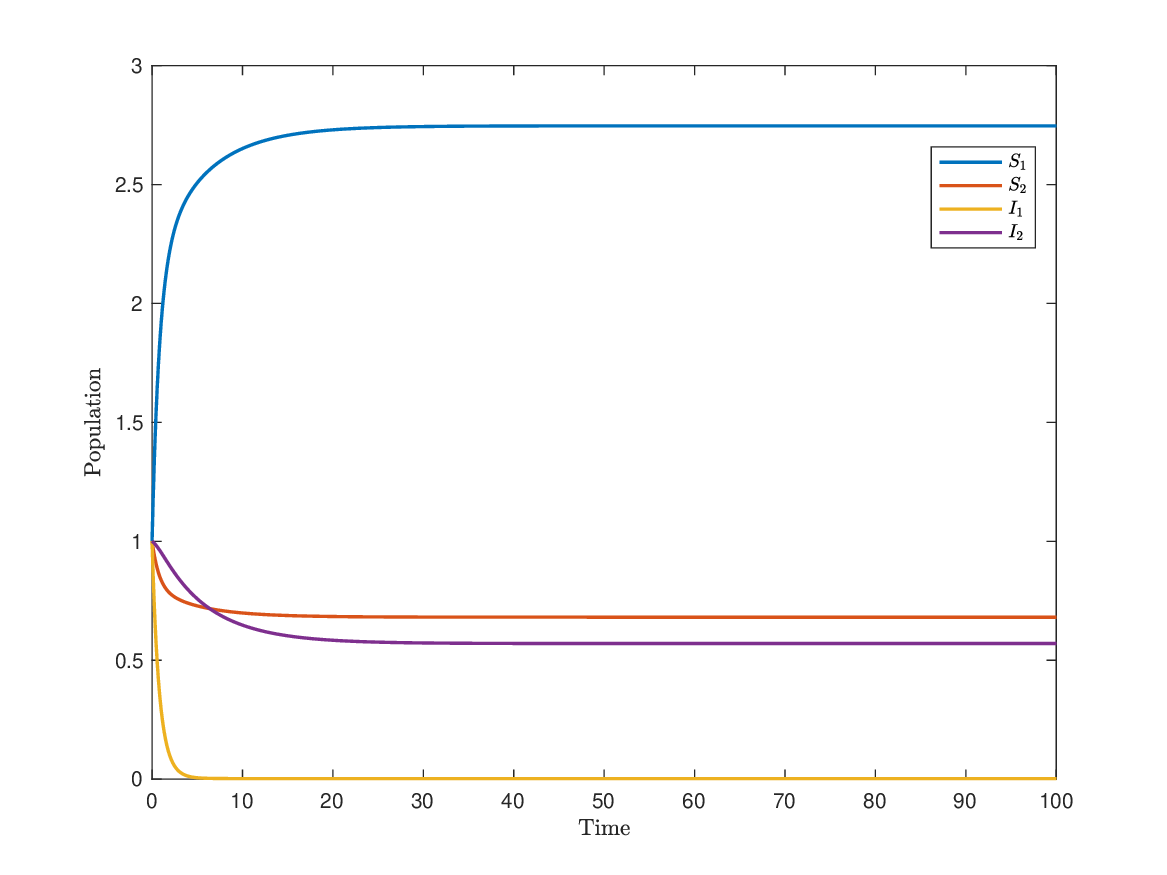}}}
\subfigure[]{\resizebox*{0.325\linewidth}{!}{\includegraphics{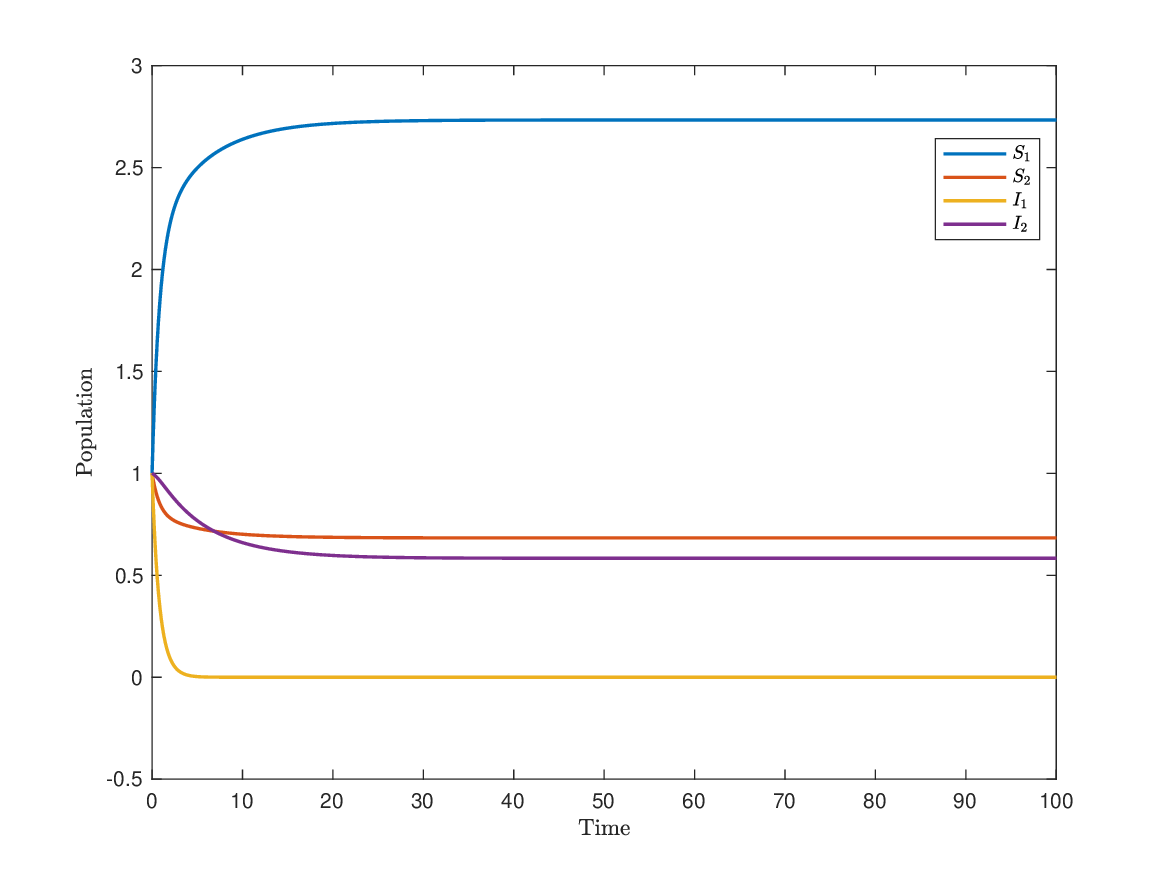}}}
\caption{Asymptotic profiles of EEs of \eqref{model} when the 
hypotheses of Theorem \ref{TH8} is satisfied: (a) $d_{S}=10^{-1}$, (b) $d_{S}=10^{-2}$, (c) $d_{S}=10^{-6}$} 
\label{Ex13}
\end{center}
\end{figure}

\subsection{Discussion}

This work examined the global dynamics of solutions to a multiple-patch epidemic model with saturated incidence mechanism \eqref{model}. In the first part, we focus on scenario when only the disease fatality rate is taken into consideration, that  $\bm\mu>\bm 0$, while the other demographics factors are negligible. In such a setting,  Theorem \eqref{TH1} predicts the eventual extinction of the disease. Moreover, our numerical simulations from experiments 1, 2, and 3 confirm our theoretical results. In the case of two patches epidemic network, these experiments discussed all the three possible scenarios.

 In the second part of our investigation, we assume that all factors affecting the total population size, including the disease-induced fatality rate, are negligible and set $\bm \mu = \bm 0$ in system \eqref{model}. Under this assumption, Theorems \ref{TH2}, \ref{TH10}, \ref{TH3}, and \ref{TH4} establish the global stability of the Disease-Free Equilibrium (DFE) under certain conditions. Specifically, if all patches are of low or moderate risk—meaning the disease transmission rate is less than or equal to the recovery rate in all patches—the disease will eventually be eradicated. Our analysis also reveals the existence of a new threshold quantity, $\tilde{\mathcal{R}}_0$ (defined by \eqref{tilde-R-0}), which is greater than or equal to the Basic Reproduction Number (BRN) $\mathcal{R}_0$ (given by formula \eqref{R-0}) for system \eqref{model}. The disease will be eradicated if $\tilde{\mathcal{R}}_0 \leq 1$. Figure \ref{Fig1_1} illustrates the curves of $\mathcal{R}_0$ and $\tilde{\mathcal{R}}_0$ with respect to the total population size. Notably, $\tilde{\mathcal{R}}_0$ and $\mathcal{R}_0$ are independent of the susceptible population dispersal rate. Our simulations in Experiment 4 demonstrate the dynamics of solutions as established by Theorem \ref{TH2}; Experiments 5, 6, and 7 illustrate the three possible scenarios under the hypotheses of Theorem \ref{TH3}, as explained in Remark \ref{RK2}; and the first part of Experiment 8 shows the global dynamics of solutions when the susceptible and infected populations have the same dispersal rate, as described in Theorem \ref{TH4}. 

 An interesting result established in \cite{SW2024a} is that the multiple-patch epidemic model \eqref{mass-action-incidence} may have multiple EEs for some range of the parameters when its BRN is less than one.  Similarly, unlike the corresponding model without a saturation effect \eqref{standard-incidence}, we find that the combination of saturation incidence, spatial heterogeneity among patches, and population movements can result in multiple endemic equilibria even when $\mathcal{R}_0<1$ and the susceptible population disperses very slowly while the infected population move faster (See Remark \ref{Rk3}). It is important to note that when the susceptible population disperses at least as quickly as the infected population, Theorem \ref{TH5}-{\rm (i)} shows that the existence of an EE solution depends entirely on whether $\mathcal{R}_0$ exceeds one. Furthermore, if $\mathcal{R}_0 > 1$ and $d_S \geq d_I$, the EE is always unique. If the requirement $d_S \geq d_I$ is relaxed, Theorem \ref{TH5}-{\rm (ii)} indicates that the EE is unique if $\mathcal{R}_0 > 1$ and the total population size is sufficiently large. Figure \ref{fig2} provides  illustrative pictures of the bifurcation diagram of  $\|\bm I\|_{\infty}$ at the EEs as $\mathcal{R}_0$ varies. Our simulations in the second part of Experiment 8 confirm these theoretical results. Consequently, when the susceptible population disperses at least as quickly as the infected population, our findings suggest that any disease control strategy aimed at reducing the BRN could effectively mitigate the impact of the disease.

 To better understand the structure of the set $\mathcal{C}_*$  of EEs for system \eqref{model} when $\bm \mu = \bm 0$, we establish in Theorem \ref{TH6} that, with fixed population dispersal rates, $\mathcal{C}_*$  forms a simple, connected, and unbounded curve that bifurcates from the set of DFEs at $\mathcal{R}_0 = 1$ as the  BRN varies. Furthermore, $\mathcal{R}_0 = 1$ is always a transcritical forward bifurcation point if either  the epidemic network consists of exactly two patches, or the susceptible population move faster than the infected population. It remains an open question whether this conclusion still holds if any of the scenarios {\rm (i)-(iii)} of Theorem \ref{TH6} are violated.

 During the recent COVID-19 pandemic, many countries adopted strategies to control the spread of the disease by limiting population movements. To assess the effectiveness of these strategies, researchers can examine the asymptotic behavior of epidemic equilibria (EEs) as the dispersal rates of populations approach zero. In this context, Theorem \ref{TH7} demonstrates that when $\bm\mu = \bm 0$, the impact of the disease can be significantly reduced by limiting the dispersal rate of susceptible populations if the epidemic network includes at least one patch of moderate or low risk, or if the total population size drops below a certain critical threshold. Specifically, if $\Omega$ consists solely of high-risk patches, i.e., $\tilde{\Omega} = \{i \in \Omega : \beta_i > \gamma_i\} = \Omega$, and $N \le \|\bm\zeta \circ \bm r / (\bm 1 - \bm r)\|_{1}$, then system \eqref{model} with $\bm\mu = \bm 0$ and $\bm \zeta \gg \bm 0$ predicts that the I-components of EEs will go extinct as $d_S$ becomes very small. This contrasts sharply with predictions from system \eqref{standard-incidence} under the same condition $\tilde{\Omega} = \Omega$. Additionally, according to \cite{SW2024a}, unlike the scenario described in Theorem \ref{TH7}, the total population size significantly affects the asymptotic behavior of the EEs in system \eqref{mass-action-incidence} as $d_S$ approaches zero when $\Omega = \tilde{\Omega}$. The conclusions of Theorem \ref{TH7} are illustrated through the simulations in Experiments 9, 10, 11, and 12. Theorem \ref{TH8} and Experiment 13 detail the asymptotic limits of EEs as the infected population dispersal rate approaches zero.

\section{Preliminary results and Proof of Proposition \ref{prop2}}\label{Sec3}
Recall that   $\bm \alpha$ is the eigenvector associated with $\sigma_*(\mathcal{L})=0$ as described in \eqref{alpha-eq}. 
We  recall the following Harnack's inequality type result for discrete dynamical  models from \cite{DBS2023}.

\begin{lem}\cite[Lemma 3.1]{DBS2023}\label{Harnck-lemma}  Suppose that  {\bf (A1)} holds. Let $d>0$ and $ \bm M\in C(\mathbb{R}_+, \mathbb{R}^n)$ such that 
\begin{equation*}
    \sup_{t\ge 0}\|\bm M(t)\|_{\infty}\le m_{\infty}<\infty.
\end{equation*}
Then there is a positive number $c_{d,m_{\infty}}$ such that any nonnegative solution $\bm X(t)$ of 
\begin{equation*}
    {\bm X'}=d\mathcal{L}\bm X +\bm M(t)\circ \bm X, \ t>0
\end{equation*}
satisfies
\begin{equation*}
    \|\bm X(t)\|_{\infty}\le c_{d,m_{\infty}}\bm X_{m}(t),\quad \forall\ t\ge 1.
\end{equation*}
    
\end{lem}

\noindent Let us also  recall the following three lemmas from \cite{SW2024a}.

\begin{lem}\cite[Lemma 1]{SW2024a}\label{lem2} Suppose that  {\bf (A1)} holds. 
    Let $d>0$  and $\bm F :\mathbb{R}_+\to \mathbb{R}^n$ be a continuous map satisfying $\|\bm F(t)\|_1\to 0$ as $t\to\infty$. If  $\bm X(t)$ is a bounded solution of the system
    \begin{equation*}
       \begin{cases}
       \bm X'(t)=d\mathcal{L}\bm X(t) +\bm F(t),\ t>0,\cr 
        \bm X(0)=\bm X^0\in\mathbb{R}^n,
        \end{cases}
    \end{equation*}
    then $\bm X(t)-({\sum_{j\in\Omega} X_j(t)})\bm\alpha\to\bm 0$ as $t\to\infty$. 
    In particular, if $\bm F(t)={\bf 0} $ for all $t\ge 0$,  then $\bm X(t)\to ({\sum_{j\in\Omega}X^0_j})\bm\alpha$ as $t\to\infty$.
\end{lem}  

We give a proof of Proposition \eqref{prop2}.

\begin{proof}[Proof of Proposition \ref{prop2}] Statement  {\rm (i)} is a well known result, see for example \cite[Theorem 2]{DW2002}. The proof of {\rm (ii)} can be found in \cite[Theorem 2.6 and theorem 2.7]{chen2020asymptotic} (see also \cite{gao2021impact}).  Note that $ F$ defined by \eqref{R-0} is strictly increasing in $N>0$, with $F\to {\rm diag}(\bm 0)$ as $N\to 0^+$ and $F\to {\rm diag}(\bm\beta)=\hat{F}$ as $N\to\infty$. Therefore, we have that $\mathcal{R}_0$ is strictly increasing in $N>0$, and \eqref{R-0-limit-in-N} holds. Moreover, since $F$ is continuous in $N>0$, then the existence  of $\mathcal{N}_0(d_I)$ follows by the intermediate value theorem. Hence {\rm (iii)} is proved.  We also note from the implicit function theorem that $\mathcal{N}_0(d_I)$ is continuous in $d_I$.  It remains to show that {\rm (iv)} holds.

\medskip

\noindent To this end, suppose that $\|\bm \beta/\bm \gamma\|_{\infty}>1$. First, note that this implies that $\tilde{\Omega}=\{j\in\Omega : \beta_j>\gamma_i\}$ is not empty. In addition, by the monotonicity of $\hat{\mathcal{R}}_0$ with respect to $d_I$ \cite[Theorem 2.6]{chen2020asymptotic}, there is $d_*\in(0,\infty]$  such that $\hat{\mathcal{R}}_0>1$ if and only if  $0<d_I<d_*$. Note that $d_*$ is independent of $\bm \zeta\gg\bm 0$ since $\hat{\mathcal{R}}_0$ is independent of it.  Hence, by {\rm (iii)}, $\mathcal{N}_0(d_I,\bm \zeta)$ is defined if and only if $0<d_I<d_*$. From here, we complete the proofs of {\rm (iv-1)-(iv-3)}.

\medskip

\noindent{\rm (iv-1)} Fix $\bm \zeta\gg \bm 0$ and suppose that there is $N^*>0$ such that $ (N^*\bm \alpha\circ\bm \beta)/(\bm\zeta+N^*\bm\alpha)=\bm\gamma$. In this case,  taking $N=N^*$ in \eqref{R-0}, we have from \eqref{R-0-limit} that  $\mathcal{R}_0=1$ for all $d_I>0$. Hence $d_*=\infty$ and $\mathcal{N}_0=N^*$ for all $ d_I>0$. 

\medskip

\noindent{\rm (iv-2)} Fix $\bm \zeta\gg \bm 0$ and suppose that for any $N>0$,  $ (N\bm \alpha\circ\bm \beta)/(\bm\zeta+N\bm\alpha)\ne \bm\gamma$. For the sake of clarity, for  every $d_I>0$ and $N>0$,  we set $\mathcal{R}_0=\mathcal{R}_0(d_I,N)$ to emphasize on the dependence of $\mathcal{R}_0$ in \eqref{R-0} with respect to $N>0$ and $d_I$. Fix $0<d_I<\tilde{d}_I<d_*$. We first show that $ (\mathcal{N}_0(d_I,\bm \zeta)\bm \alpha\circ\bm \beta)/(\bm\zeta+\mathcal{N}_0(d_I,\bm \zeta)\bm\alpha)\notin{\rm span}(\bm\gamma)$.  If this was false, there would exist $\tau>0$ such that $ (\mathcal{N}_0(d_I,\bm \zeta)\bm \alpha\circ\bm \beta)/(\bm\zeta+\mathcal{N}_0(d_I,\bm \zeta)\bm\alpha)=\tau\bm\gamma$. This along with \eqref{R-0-limit} implies that   $\mathcal{R}_0(d_I,\mathcal{N}_0(d_I,\bm \zeta))=\tau$. Thus, since $\mathcal{R}_0(d_I,\mathcal{N}_0(d_I,\bm \zeta))=1$, we must have that $\tau=1$, that is  $ (\mathcal{N}_0(d_I,\bm \zeta)\bm \alpha\circ\bm \beta)/(\bm\zeta+\mathcal{N}_0(d_I,\bm \zeta)\bm\alpha)=\bm\gamma$, which is contrary to our initial assumption. Therefore   $ (\mathcal{N}_0(d_I,\bm \zeta)\bm \alpha\circ\bm \beta)/(\bm\zeta+\mathcal{N}_0(d_I,\bm \zeta)\bm\alpha)\notin{\rm span}(\bm\gamma)$. As a result, we can invoke 
Proposition \eqref{prop2}-{\rm (ii)}
to conclude that $$\mathcal{R}_0(\tilde{d}_I,\mathcal{N}_0(d_I,\bm \zeta))<\mathcal{R}_0(d_I,\mathcal{N}_0(d_I,\bm \zeta))=1,$$
since $d_I<\tilde{d}_I$. Therefore, recalling (from {\rm (iii)}) that $\mathcal{R}_0(\tilde{d}_I,N)$ is strictly increasing in $N>0$, and $\mathcal{R}_0(\tilde{d}_I,\mathcal{N}_0(\tilde{d}_I,\bm \zeta))=1$, we must have that  $\mathcal{N}_0(\tilde{d}_I,\bm \zeta)>\mathcal{N}_0(d_I,\bm \zeta)$. This shows that $\mathcal{N}_0$ is strictly increasing in $0<d_I<d_*$. Therefore, $N_*:=\lim_{d_I\to 0^+}\mathcal{N}_0(d_I,\bm \zeta)$ exists in $[0,\infty)$. Recalling that $\mathcal{R}_0(d_I,\mathcal{N}_0(d_I,\bm \zeta))=1$ for all $0<d_I<d_*$, we can use a perturbation arguments and \cite[Theorem 2.7]{chen2020asymptotic} to obtain that 
$$
1=\lim_{d_I\to 0^+}\mathcal{R}_0(d_I,\mathcal{N}_0(d_I,\bm \zeta))=\max_{j\in\Omega}\frac{\beta_jN_*\alpha_j}{\gamma_j(\zeta_j+N_*\alpha_j)}.
$$
Solving for $N_*$ in the last equation, we get $N_*=\min_{j\in\tilde{\Omega}}\frac{\gamma_j\zeta_j}{(\beta_j-\gamma_j)\alpha_j}$ since $\frac{\beta_jN_*\alpha_j}{\gamma_j(\zeta_j+N_*\alpha_j)}<1$ whenever $j\in\Omega\setminus\tilde{\Omega}$.  

\medskip

\noindent{\rm (iv-3)} { Fix $0<d_I<d_*$. Since the diagonal matrix $F$ in \eqref{R-0} is decreasing in $\bm\zeta\gg\bm 0$, then $\mathcal{R}_0$ is strictly decreasing in $\bm \zeta\gg \bm0$. Then, thanks to the properties of $\mathcal{N}_0$ in {\rm (iii)}, we can proceed by a proper modification of the proof of the monotonicity of $\mathcal{N}_0$ in $d_I>0$ to establish   that $\mathcal{N}_0$ is strictly increasing in $\bm \zeta\gg \bm 0$.  Now, for every $\tau>0$ and $\bm\zeta\gg \bm 0$, since
$$
{\rm diag}(\tau\mathcal{N}_0(d_I,\bm\zeta)\bm\alpha\circ \bm\beta/(\tau\bm\zeta+\tau \mathcal{N}_0(d_I,\bm\zeta)\bm\alpha))={\rm diag}(\mathcal{N}_0(d_I,\bm\zeta)\bm\alpha\circ \bm\beta/(\bm\zeta+\mathcal{N}_0(d_I,\bm\zeta)\bm\alpha)),
$$
then 
$$
\rho({\rm diag}(\tau\mathcal{N}_0(d_I,\bm\zeta)\bm\alpha\circ \bm\beta/(\tau\bm\zeta+\tau \mathcal{N}_0(d_I,\bm\zeta)\bm\alpha))V^{-1})=\rho({\rm diag}(\mathcal{N}_0(d_I,\bm\zeta)\bm\alpha\circ \bm\beta/(\bm\zeta+\mathcal{N}_0(d_I,\bm\zeta)\bm\alpha))V^{-1})=1
$$
 which in turn  yields $\mathcal{N}_0(d_I,\tau\bm\zeta)=\tau\mathcal{N}_0(d_I,\bm\zeta)$. Therefore
 $$\bm\zeta_m\mathcal{N}_0(d_I,\bm 1)=\mathcal{N}_0(d_I,\bm\zeta_m\bm 1)\le \mathcal{N}_0(d_I,\bm\zeta)\le \mathcal{N}_0(d_I,\bm\zeta_M\bm 1)= \bm\zeta_M\mathcal{N}_0(d_I,\bm 1).
 $$
}

    
\end{proof}

\section{Proofs of the Main Results}\label{Sec4}

\subsection{Proof of Theorem \ref{TH1}}

\begin{proof}Suppose that {\bf (A1)-(A2)} holds.  Suppose also that $\bm\mu>{\bm 0}.$ Let $(\bm S(t),\bm I(t))$ be a solution of \eqref{model} with a positive initial data satisfying {\bf (A3)}. Recall from \eqref{Eq1:1} that 
\begin{equation}\label{Eq1}
    \frac{d}{dt}\sum_{j\in\Omega}(S_j+I_j)=-\sum_{j\in\Omega}\mu_jI_j\quad t\ge 0.
\end{equation}
Observing that 
$$ 
\sup_{t\ge 1}\|\bm\beta\circ\bm S(t)/(\bm\zeta+\bm S(t)+\bm I(t))-\bm\gamma-\bm\mu\|_{\infty}\leq \|\bm\beta\|_{\infty}+\|\bm\gamma\|_{\infty}+\|\bm\mu\|_{\infty},
$$
it follows from Lemma \ref{Harnck-lemma} that there is $c_{1}=c_1(d_I)$ such that 
\begin{equation}\label{Eq2}
  \|\bm I(t)\|_{\infty}\le c_1\min_{j\in\Omega}I_j(t)\quad \forall\ t\ge 1.  
\end{equation}
Thanks to \eqref{Eq1} and \eqref{Eq2}, 
\begin{equation*}
    \frac{d}{dt}\sum_{j\in\Omega}(S_j+I_j)\le -(\sum_{j\in\Omega}\mu_j)\min_{j\in\Omega}I_j(t)\le -\frac{\|\bm \mu\|_1}{c_1}\|\bm I(t)\|_{\infty}\quad \forall\ t\ge 1.
\end{equation*}
An integration of the last inequality gives
\begin{equation*}
    \sum_{j\in\Omega}(S_j(t)+I_j(t))+\frac{\|\bm \mu\|_1}{c_1}\int_{1}^{t}\|\bm I(s)\|_{\infty}ds\le \sum_{j\in\Omega}(S_j(1)+I_j(1))\leq \|\bm S^0+\bm I^0\|_1\quad \forall\ t\ge 1.
\end{equation*}
As a result, since $\|\bm X\|_{1}\le n\|\bm X\|_{\infty}$ for any $\bm X\in\mathbb{R}^n$, we have 
$$
\int_{1}^{\infty}\sum_{j\in\Omega}I_j(t)dt\le\frac{c_1 n\|\bm S^0+\bm I^0\|_1}{\|\bm \mu\|_1}.
$$
Observing that the mapping $[0,\infty)\ni t\mapsto \sum_{j\in\Omega}I_j(t)$ is Lipschitz continuous, we conclude from the last inequality that $\|\bm I(t)\|_1\to 0$ as $t\to\infty$. This in turn implies that 
$$
\|-\bm\beta\circ\bm I\circ \bm S/(\bm\zeta+\bm S+\bm I)+\bm\gamma\circ\bm I\|_{\infty}\le (\|\bm \beta\|_{\infty}+\|\bm \gamma\|_{\infty})\|\bm I(t)\|_{\infty}\to 0 \quad \text{as}\ t\to\infty.
$$
This along with Lemma \ref{lem2} gives 
\begin{equation}\label{Eq3}
    \lim_{t\to\infty}\|\bm S(t)-(\sum_{j\in\Omega}S_j(t))\bm\alpha\|_{\infty}=0.
\end{equation}
However, by \eqref{Eq1}, we have 
\begin{equation}\label{YU1}
\sum_{j\in\Omega}S_j(t)=\|\bm S^0+\bm I^0\|_1-\int_{0}^t\sum_{j\in\Omega}\mu_jI_j(s)ds -\sum_{j\in\Omega}I_j(t).
\end{equation}
Letting $t\to\infty$, and recalling that $\|\bm I(t)\|_1\to 0$ as $t\to\infty$, then 
\begin{equation}\label{Eq4}
\lim_{t\to\infty}\sum_{j\in\Omega}S_j(t)=\|\bm S^0+\bm I^0\|_1-\int_0^{\infty}\sum_{j\in\Omega}\mu_j I_j(t)dt.
\end{equation}
Therefore, by \eqref{Eq3}, $\bm S(t)\to \big(\|\bm S^0+\bm I^0\|_1-\int_0^{\infty}\sum_{j\in\Omega}\mu_j I_j(t)dt\big)\bm\alpha$ as $t\to\infty$. 

Next, for each $i\in\Omega$, it holds that
\begin{align*}
\frac{dS_i}{dt}=&d_S\sum_{j\in\Omega}L_{ij}S_j+(\gamma_i(\zeta_i+S_i+I_i)-\beta_iS_i)\frac{I_i}{\zeta_i+S_i+I_i}\cr 
\ge & d_S\sum_{j\in\Omega}L_{ij}S_j+(\bm\gamma_m\bm\zeta_m-\|\bm\beta\|_{\infty}\sum_{j\in\Omega}S_j)\frac{I_i}{\zeta_i+S_i+I_i}.
\end{align*}
Thus,
$$
\frac{d}{dt}\sum_{j\in\Omega}S_j\ge (\bm\gamma_m\bm\zeta_m-\|\bm\beta\|_{\infty}\sum_{j\in\Omega}S_j)\sum_{j\in\Omega}\frac{I_j}{\zeta_j+S_j+I_j}\quad t>0.
$$
We can now employ the comparison principle for ODEs to deduce that 
$$
\sum_{j\in\Omega}S_j(t)\ge \min\Big\{\frac{\bm\gamma_m\bm\zeta_m}{\|\bm\beta\|_{\infty}}, \sum_{j\in\Omega}S_j(t_0)\Big\}>0\quad \forall\ t\ge t_0>0.
$$ 
Letting $t\to\infty$ and recalling \eqref{Eq4}, we get 
$$ 
\|\bm S^0+\bm I^0\|_1-\int_0^{\infty}\sum_{j\in\Omega}\mu_j I_j(t)dt=\lim_{t\to\infty}\sum_{j\in\Omega}S_j(t)\ge \min\Big\{\frac{\bm\gamma_m\bm\zeta_m}{\|\bm\beta\|_{\infty}}, \sum_{j\in\Omega}S_j(t_0)\Big\}>0\quad \forall\  t_0>0,
$$
which completes the proof of the theorem.

\end{proof}

{
\noindent For reference, we state the following result on the large-time behavior of solutions of \eqref{model} when $|\Omega|=1$. Suppose that $|\Omega|=1$, hence the dispersal terms cancel out in \eqref{model}. Fix $S^0\ge 0$, $I^0>0$ and set
\begin{equation}\label{YU3}
    S^{*}:=\begin{cases}
    0 & \text{if}\ \beta\ge \mu+\gamma,\cr
        e^{-\frac{\mu I^0}{(S^0+I^0)\gamma}} & \text{if} \ \beta=\mu\cr
        \Big(1-\frac{(\mu-\beta)I^0}{(\mu+\gamma-\beta)(S^0+I^0)}\Big)^{\frac{\mu}{\mu-\beta}} & \text{if} \ \beta\ne\mu \ \text{and}\quad \beta<\mu+\gamma.
    \end{cases}
\end{equation}

\medskip
\begin{tm}\label{TH9} Suppose that $ |\Omega|=1$, $\mu>0$ and $\zeta\ge 0$ in system \eqref{model}.  Let $(S(t),I(t))$ be the solution of \eqref{model} with initial data $(S^0,I^0)$ satisfying {\bf (A3)}. Set $N^0=S^0+I^0>0$ and let $S^*$ be defined by \eqref{YU3}. 
\begin{itemize}
    \item[\rm (i)] If $\zeta>0$, then $(S(t),I(t))\to \Big(S^{(\zeta)},0\Big)$ as $t\to\infty$ for some positive number $S^{(\zeta)}$. Moreover, $S^{(\zeta)}=\frac{\zeta^{\frac{\mu}{\beta}}(N^0+\zeta)}{(I^0+\zeta)^{\frac{\mu}{\beta}}}-\zeta$ if $\beta=\mu+\gamma$, and $S^{(\zeta)}\to 0$ as $\zeta\to0^+$ if $\beta\ge \gamma+\mu$.
    
    \item[\rm(ii)] If $\zeta=0$, then $(S(t),I(t))\to (N^0S^*,0)$ as $t\to\infty$.
\end{itemize}
\end{tm}
\begin{proof} Set $N(t)=\zeta+S(t)+I(t)$ for all $t\ge 0$. Then 
$$
\frac{1}{N}\frac{dN}{dt}=-\mu \frac{I}{N}=
\frac{\mu}{\beta}\Big(\frac{1}{I+\zeta}\frac{dI}{dt}-\frac{(\beta-\mu-\gamma)I}{I+\zeta}\Big)\quad t>0.
$$
    Therefore
    \begin{equation*}
        \ln\Big(\frac{N}{N(0)}\Big)=\frac{\mu}{\beta}\Big(\ln\Big(\frac{I+\zeta}{I^0+\zeta}\Big)-(\beta-\gamma-\mu)\int_0^t\frac{I(s)}{I(s)+\zeta}ds\Big) \quad \forall\ t>0.
    \end{equation*}
    Equivalently,
    \begin{equation}\label{YU2}
        N(t)=(N^0+\zeta)\left[\frac{(I(t)+\zeta)}{(I^0+\zeta)}e^{-(\beta-\gamma-\mu)\int_{0}^t\frac{I(s)}{I(s)+\zeta}ds}\right]^{\frac{\mu}{\beta}}\quad \forall\, t>0.
    \end{equation}
    {\rm (i)} Suppose that $\zeta>0$.  Thanks to \eqref{YU2} and the fact that $I(t)\to 0$ as $t\to\infty$ (by Theorem \ref{TH1}),
    $$
    S(t)=N(t)-I(t)-\zeta\to S^{(\zeta)}:=(N^0+\zeta)\left[\frac{\zeta}{I^0+\zeta}\right]^{\frac{\mu}{\beta}}e^{-\frac{\mu(\beta-\mu-\gamma)}{\beta}\int_0^{\infty}\frac{I(s)}{I(s)+\zeta}ds}-\zeta\quad \text{as}\quad t\to\infty.    $$
    
    \noindent{} If $ \beta=\mu+\gamma$, then $ S^{(\zeta)}=(N^0+\zeta)\left[\frac{\zeta}{I^0+\zeta}\right]^{\frac{\mu}{\beta}}-\zeta $. It is clear that $S^{(\zeta)}\to0$ as $\zeta\to0^+$.
    
    \noindent{} If $\beta>\mu+\gamma$, then $ S^{(\zeta)}\le(N^0+\zeta)\left[\frac{\zeta}{I^0+\zeta}\right]^{\frac{\mu}{\beta}}-\zeta\to 0 $ as $\zeta\to 0^+$.

    \medskip

    \noindent{\rm (ii)} Suppose that $\zeta=0$. We note as in the proof of Theorem \eqref{TH1} that $I(t)\to 0$ as $t\to\infty$ since $\mu>0$.  Next, by \eqref{YU2}, $N(t)= N^0Z^{\frac{\mu}{\beta}}(t)$ for all $t>0$, where $Z(t):=\frac{I(t)}{I^0}e^{(\mu+\gamma-\beta)t}$ for all $t>0$. Now, observe that 
    \begin{align*}
    Z^{\frac{\mu}{\beta}-2}\frac{dZ}{dt}=-\frac{I^0\beta}{N^0}e^{-(\mu+\gamma-\beta)t}\quad t>0,\quad Z(0)=1. 
    \end{align*}
    Using elementary analyses for ODEs, we can solve  $Z(t)$ to obtain
    \begin{equation}\label{Z-formula}
        Z(t)=\begin{cases}
            e^{-\frac{\beta I^0}{\gamma N^0}(1-e^{-\gamma t})} & \ \forall\ t>0\  \text{if} \ \beta=\mu,\cr
            \left[1+\frac{\gamma I^0}{N^0}t\right]^{-\frac{\beta}{\gamma}} &\ \forall\ t>0 \ \text{if}\  \beta=\mu+\gamma,\cr
            \left[1-\frac{(\mu-\beta)I^0}{(\mu+\gamma-\beta)N^0}\Big(1-e^{-(\mu+\gamma-\beta)t}\Big)\right]^{\frac{\beta}{\mu-\beta}} &\ \forall\ t>0\  \text{if}\ \mu\ne \beta \ \text{and}\ \beta\ne\mu+\gamma.
        \end{cases} 
    \end{equation}
    Taking limit as $t\to\infty$, we obtain that
    $$
    Z(t)\to \begin{cases}
        e^{-\frac{\beta I^0}{\gamma N^0}}  &\ \text{if} \ \beta=\mu,\cr 
        \left[1-\frac{(\mu-\beta)I^0}{(\mu+\gamma-\beta)N^0}\right]^{\frac{\beta}{\mu-\beta}} &\   \text{if}\ \mu\ne \beta \ \text{and}\ \beta<\mu+\gamma,\cr
        0 &\ \text{if} \ \beta{\ge} \mu+\gamma
    \end{cases}=[S^{*}]^{\frac{\beta}{\mu}}
    \quad \text{as}\ t\to\infty.
    $$
    Therefore, $S(t)=N(t)-I(t)=N^0Z^{\frac{\mu}{\beta}}(t)-I(t)\to N^0S^*$ as $t\to\infty$.
\end{proof}
\begin{rk} Assume the hypotheses of Theorem \ref{TH9}-{\rm (ii)} hold. 
The work \cite[Theorem A.1]{GLPZ2023} also studied the large time behavior of solutions and established the extinction of the population if $\beta{\ge} \mu+\gamma$, and the persistence of only the susceptible population if $\beta{<}\gamma+\mu$. However, if $\beta{<}\mu+\gamma$, the explicit formula for limit of $(S(t),I(t))$ as $t\to\infty$ was not provided in \cite{GLPZ2023}. 

    
\end{rk}

 }

\subsection{Proof of Theorems \ref{TH2}, \ref{TH10}, \ref{TH3}, and \ref{TH4}}

\begin{proof}[Proof of Theorem \ref{TH2}] Suppose that $\bm\mu=\bm 0$ and $\tilde{\mathcal{R}}_0\le 1$. Let $(\bm S(t),\bm I(t))$ be a solution of \eqref{model} with a positive initial data satisfying {\bf (A3)} which belongs to $\mathcal{E}$.  Since $\mathcal{E}$ is invariant for \eqref{model}, then for every $i\in\Omega$, we have 
$$
\frac{\beta_iS_i}{\zeta_i+S_i+I_i}<\frac{\beta_i\sum_{j\in\Omega}S_j}{\zeta_i+\sum_{j\in\Omega}S_j+I_i}<\frac{N\beta_i}{\zeta_i+N+I_i}\quad \forall\ t>0.
$$
Therefore,
\begin{equation}\label{Eqa2-1}
    \bm I'(t)\le d_I\mathcal{L}\bm I +(N\bm\beta/(\bm\zeta+N{\bm 1}+\bm I)-\bm\gamma)\circ \bm I \quad t>0.
\end{equation}
Let $\tilde{\bm \alpha}$ denote the positive eigenfunction associated with $\tilde{\sigma}_*:=\sigma_*(d_I\mathcal{L}+{\rm diag}(N\bm\beta/(\bm\zeta+N{\rm 1})-\bm\gamma))$ satisfying $(\tilde{\bm \alpha})_m=N$. Since $\sum_{j\in\Omega}(S_j(t)+I_j(t))=N$ for all $t\ge 0$, then $\bm I(0)\le \tilde{\bm \alpha}$. Next,  define 
$$
\eta(t)=(\bm I(t)/\tilde{\bm \alpha})_M\quad \forall\ t\ge 0.
$$
Since, $\bm I(0)\le \tilde{\bm \alpha}$, then $\eta(0)\le 1$. Now claim that 
\begin{equation}\label{Eqa2-2}
 \eta(t+\tau)\le \eta(\tau)\quad t,\tau \ge 0.
\end{equation}
Indeed, fix $\tau>0$ and  note that $\tilde{\bm I}(t):=\eta(\tau)e^{t\tilde{\sigma}_*}\tilde{\bm\alpha} $, $t\ge 0$, satisfies 
\begin{equation*}
    \tilde{\bm I}(t)=d_I\mathcal{L}\tilde{\bm I}+(N\bm\beta/(\bm\zeta+N\bm 1)-\bm \gamma)\circ\tilde{\bm I} \quad t\ge 0.
\end{equation*}
Note also from \eqref{Eqa2-1} that the mapping  $\underline{\bm I}(t):=\bm I(t+\tau)$, $t\ge 0,$ satisfies 
\begin{align*}
    \underline{\bm I}'(t)\le & d_I\mathcal{L}\underline{\bm I}+(N\bm\beta/(\bm\zeta+N\bm 1+\underline{\bm I})-\bm\gamma)\circ\underline{\bm I} 
    \le d_I\mathcal{L}\underline{\bm I}+(N\bm\beta/(\bm\zeta+N\bm 1)-\bm\gamma)\circ\underline{\bm I}\quad t>0.
\end{align*}
Therefore, since $\underline{\bm I}(0)=\bm I(\tau)\le \eta(\tau)\tilde{\bm \alpha}=\tilde{\bm I}(0)$, and $\mathcal{L}$ is quasipositive and irreducible, we can employ the comparison principle for cooperative systems to conclude that $\underline{\bm I}(t)\le \tilde{\bm I}(t)$ for all $t> 0$. Equivalently, $\bm I(t+\tau)\le \eta(\tau)e^{\tilde{\sigma}_*t}\tilde{\bm\alpha}$ for all $t\ge 0$. Therefore, observing that $\tilde{\sigma}_*\le 0$ since $\tilde{\mathcal{R}}_0\le 1$, then $\bm I(t+\tau)\le \eta(\tau)\tilde{\bm\alpha}$ for all $t\ge 0$, that is $\eta(t+
\tau)\le \eta(\tau)$ for all $t\ge 0$. This shows that \eqref{Eqa2-2} holds since $\tau\ge 0$ is arbitrary fixed.

\medskip

\noindent Thanks to \eqref{Eqa2-2}, we have that 
$$
\eta^*:=\inf_{t>0}\eta(t)=\lim_{t\to\infty}\eta(t).
$$
Next, we claim that 
\begin{equation}
\label{Eqa2-3}
    \eta^*=0.
\end{equation}
To establish \eqref{Eqa2-3}, we first note from \eqref{Eq2} that 
$$
\eta^*\le\eta(t)\le \|\bm I\|_{\infty}(\bm 1/\tilde{\bm \alpha})_M\le c_1(\bm 1/\tilde{\bm \alpha})_M\bm I_m:=\tilde{c}_1\bm I_m\quad \forall\ t\ge 1,
$$
where $\tilde{c}_1=c_1(\bm 1/\tilde{\bm \alpha})_M$ and $c_1$ is as in \eqref{Eq2}. This along with \eqref{Eqa2-1} implies that 
\begin{align}\label{Eqd-21}
    \bm I'(t)\le d_I\mathcal{L}\bm I +\Big(N\bm\beta/(\bm\zeta+N{\bm 1}+\frac{\eta^*}{\tilde{c}_1}\bm 1)-\bm\gamma\Big)\circ \bm I \quad t>1.
\end{align}
If \eqref{Eqa2-3} was false, that is $\eta^*>0$, then we would have 
$$\sigma_*(d_I\mathcal{L}+{\rm diag}(N\bm\beta/(\bm\zeta+N{\bm 1}+\frac{\eta^*}{\tilde{c}_1}\bm 1)-\bm\gamma))<\sigma_*(d_I\mathcal{L}+{\rm diag}(N\bm\beta/(\bm\zeta+N{\bm 1})-\bm\gamma))\le0.$$
As a result, it follows from \eqref{Eqd-21} that $\|\bm I(t)\|_{\infty}\to 0$ as $t\to\infty$, which in turn gives
$$
\eta^*\le \eta(t)\le \tilde{c}_1\bm I_{m}\to 0 \quad \text{as} \ t\to\infty.
$$
This contradicts our assumption that $\eta^*>0$. Therefore \eqref{Eqa2-3} must hold.

\medskip

\noindent Finally, since \eqref{Eqa2-3} holds, and recalling from the definition of $\eta(t)$ that 
$$
\bm I(t)\le \eta(t)\tilde{\bm \alpha} \quad \forall\ t\ge 0,
$$
 then 
$\bm I(t)\to \bm 0 $ as $t\to\infty$. This along with Lemma \ref{lem2} and the fact that $\sum_{j\in\Omega}S_j(t)=N-\sum_{j\in\Omega}I_j(t)\to N$ as $t\to\infty$ implies that $\bm S(t)\to N\bm\alpha$ as $t\to\infty$. 
\end{proof}

{\begin{proof}[Proof of Theorem \ref{TH10}]
Suppose that $0<N\le\mathcal{N}^*_{\rm up}$. Let $(\bm S(t),\bm I(t))$ be the  solution of \eqref{model} with initial in $\mathcal{E}$ satisfying {\bf (A3)}. We claim that 
\begin{equation}\label{OP1}
    \lim_{t\to\infty}\|\bm I(t)\|_{\infty}=0.
\end{equation}
To this end, we distinguish two cases.

\noindent {\bf Case 1.} In this case, suppose that $\int_0^{\infty}{(\sum_{j\in\Omega}I_j(t))^2}dt<\infty$. In this case, it is easy  to see that \eqref{OP1} holds.

\noindent{\bf Case 2.} Here, we suppose that $\int_{0}^{\infty}{(\sum_{j\in\Omega}I_j(t))^2}dt=\infty$. First, fix $0<\tau_0<1$ such that $(1-\tau_0)N\bm \alpha\le \bm S(1).$
Next, set 
$$
K_*:=\Big({\bm\beta}/((\bm\zeta+2\mathcal{N}^*_{\rm up}\bm 1)\circ(\bm\zeta+\mathcal{N}^*_{\rm up}\bm 1))\Big)_m>0.
$$
Finally, set  $\nu=\frac{K_*(1-\tau_0)}{\tau_0{c_1^2|\Omega|^2}}$, 
where $c_1$ is given by \eqref{Eq2}, and define  
$$ \underline{\bm S}(t)=(1-\tau_0 e^{-\nu\int_1^t{(\sum_{j\in\Omega}I_j)^2}ds})N\bm\alpha\quad t>1.
$$
Then 
\begin{equation}\label{OP2}
\underline{\bm S}(1)= (1-\tau_0)N\bm\alpha\le \bm S(1) \quad \text{and}\quad \underline{\bm S}(t)\ge \underline{\bm S}(1)\gg \bm 0 \quad \forall\ t>1. 
\end{equation}
By computations, we have 
\begin{align*}
    &\frac{d\underline{\bm S}}{dt}-d_S\mathcal{L}\underline{\bm S}-\bm\beta\circ(\bm r-\underline{\bm S}/(\bm\zeta+\underline{\bm S}+\bm I))\circ\bm I\cr
=&\tau_0\nu Ne^{-\nu\int_1^t{(\sum_{j\in\Omega}I_j(s))^2}ds}{(\sum_{j\in\Omega}I_j(t))^2}\bm\alpha-\bm\beta\circ\Big(\bm r-\underline{\bm S}/(\bm\zeta +\underline{\bm S})+(\underline{\bm S}/(\bm\zeta +\underline{\bm S})-\underline{\bm S}/(\bm\zeta +\underline{\bm S}+\bm I)) \Big)\circ\bm I
\end{align*}
Observe that since $\underline{\bm S}\le \mathcal{N}^*_{\rm up}\bm\alpha$, then 
\begin{equation*}
    \bm r-\underline{\bm S}/(\bm\zeta +\underline{\bm S})\ge \bm r-\mathcal{N}^*_{\rm up}\bm\alpha/(\bm\zeta+\mathcal{N}^*_{\rm up}\bm\alpha)=(\bm r\circ \bm\zeta-\mathcal{N}^*_{\rm up}(1-\bm r)\circ\bm\alpha)/(\bm\zeta +\mathcal{N}^*_{\rm up}\bm\alpha)\ge \bm 0\quad \forall\ t>1.
\end{equation*}
Then, since $\bm I\le N\bm 1\le \mathcal{N}^*_{\rm up}\bm 1$, $\underline{\bm S}(1)\le \underline{\bm S}\le \mathcal{N}^*_{\rm up}\bm\alpha \le\mathcal{N}^*_{\rm up}\bm 1$,  and  \eqref{Eq2} holds, we  have
\begin{align}
    &\frac{d\underline{\bm S}}{dt}-d_S\mathcal{L}\underline{\bm S}-\bm\beta\circ(\bm r-\underline{\bm S}/(\bm\zeta+\underline{\bm S}+\bm I))\circ\bm I\cr
\le &\tau_0\nu N e^{-\nu\int_1^t{(\sum_{j\in\Omega}I_j(s))^2}ds}{(\sum_{j\in\Omega}I_j(t))^2}\bm\alpha-\bm\beta\circ\Big({\bm 1}/(\bm\zeta +\underline{\bm S})-{\bm 1}/(\bm\zeta +\underline{\bm S}+\bm I) \Big)\circ\underline{\bm S}\circ\bm I\cr 
=&\tau_0\nu Ne^{-\nu\int_1^t{(\sum_{j\in\Omega}I_j(s))^2}ds}{(\sum_{j\in\Omega}I_j(t))^2}\bm\alpha-\big(\bm\beta/((\bm\zeta+\underline{\bm S}+\bm I)\circ(\bm\zeta+\underline{\bm S}))\big)\circ\underline{\bm S}\circ\bm I{\circ\bm I}\cr 
\le &\tau_0\nu{c_1^2|\Omega|^2}Ne^{-\nu\int_1^t{(\sum_{j\in\Omega}I_j(s))^2}ds}{(\min_{j\in\Omega}I_j)^2}\bm\alpha-K_*\underline{\bm S}\circ\bm I\circ{\bm I}\cr 
\le& \tau_0\nu{c_1^2|\Omega|^2} Ne^{-\nu\int_1^t{(\sum_{j\in\Omega}I_j(s))^2}ds}\bm I\circ{\bm I}\circ\bm\alpha-K_*\underline{\bm S}\circ\bm I\circ{\bm I}\cr 
=&\Big(\nu \tau_0{c_1^2|\Omega|^2} Ne^{-\nu\int_1^t{(\sum_{j\in\Omega}I_j(s))^2}ds }\bm\alpha-K_*\underline{\bm S}\Big)\circ \bm I\circ{\bm I}\cr 
\le & \Big(\nu \tau_0{c_1^2|\Omega|^2}N\bm\alpha-K_*\underline{\bm S}(1)\Big)\circ \bm I\circ{\bm I} 
= (\nu\tau_0{c_1^2|\Omega|^2}-K_*(1-\tau_0))N\bm\alpha\circ \bm I\circ{\bm I}=\bm 0.
\end{align}
Therefore, by \eqref{OP2} and the comparison principle, we have that $\underline{\bm S}(t)\le \bm S(t)$ for all $t\ge 1$. Hence 
$$
(1-\tau_0e^{-\nu\int_1^t{(\sum_{j\in\Omega}I_j(s))^2}ds})N=\sum_{j\in\Omega}\underline{S}_j(t)\le \sum_{j\in\Omega}S_j(t)\le N\quad \forall\, t>1.
$$
Letting $t\to\infty$ in the last inequality and recalling that $\int_{0}^{\infty}{(\sum_{j\in\Omega}I_j(t))^2}dt=\infty$, we obtain that $\|\bm S(t)\|_1\to N$ as $t\to\infty$, which implies that \eqref{OP1} holds.

From the above two cases, we have that \eqref{OP1} holds. Finally, thanks to \eqref{OP1}, we can proceed as in the proof of Theorem \ref{TH2} to conclude that $\bm S(t)\to N\bm\alpha$ as $t\to\infty$. 
\end{proof}
}

\begin{proof}[Proof of Theorem \ref{TH3}]Suppose that $\bm r\in{\rm span}(\bm 1)$ and $\bm \zeta\in{\rm span}(\bm \alpha)$. Hence, there exist $\tau>0$ and $m>0$ such that
$$
\bm r=\tau\bm 1\quad\text{and}\quad
\bm\zeta=m\bm\alpha.
$$
Hence, by \eqref{R-0-limit}, it holds that
\begin{equation}\label{Eq18}
\mathcal{R}_0=\frac{N}{\tau(m+N)}.
\end{equation}
Now, for each $i\in\Omega$, we have
\begin{align}\label{Eq5}
\frac{dS_i}{dt}=
d_S\sum_{j\in\Omega}L_{ij}S_j+\beta_i(\tau m\alpha_i+\tau I_i-(1-\tau)S_i)\frac{I_i}{\zeta_i+S_i+I_i},
\end{align}
and 
\begin{align}\label{Eq6}
\frac{dI_i}{dt}
= d_I\sum_{j\in\Omega}L_{ij}I_j+\beta_i((1-\tau)S_i-\tau m\alpha_i-\tau I_i)\frac{I_i}{\zeta_i+S_i+I_i}.
\end{align}

\noindent{\bf Case 1.} First, we suppose that $\tau\ge 1 $ and show that $ \bm I(t)\to\bm0$ and $\bm S(t)\to N\bm\alpha$ as $t\to\infty$.  Then, by \eqref{Eq6}, for each $i\in\Omega$,
\begin{align*}
    \frac{d I_i}{dt}
    \le & d_I\sum_{j\in\Omega}L_{ij}I_j-\beta_i((\tau-1)S_i+\tau m\alpha_i+\tau I_i)\frac{I_i}{\sum_{k\in\Omega}\zeta_k+\sum_{k\in\Omega}(S_k+I_k)}\cr
    \le & d_I\sum_{j\in\Omega}L_{ij}I_j-\frac{\tau m\alpha_i\beta_i}{\|\bm\zeta\|_1+N}I_i
    \le  d_I\sum_{j\in\Omega}L_{ij}I_j-\frac{\tau m\bm\alpha_m\bm\beta_m}{\|\bm\zeta\|_1+N}I_i.
\end{align*}
Hence 
$$
\bm I'\le d_I\mathcal{L}\bm I-\frac{\tau m\bm\alpha_m\bm\beta_m}{\|\bm\zeta\|_1+N}\bm I\quad\quad  t>0,
$$
which thanks to the comparison principle for cooperative systems yields 
$$
\bm I(t)\le e^{-\frac{\tau m\bm\alpha_m\bm\beta_m}{\|\bm\zeta\|_1+N}t}e^{td_I\mathcal{L}}\bm I(0)\quad t\ge 0.
$$
Therefore, 
$$
\|\bm I(t)\|_{\infty}\le e^{-{\frac{\tau m\bm\alpha_m\bm\beta_m}{\|\bm\zeta\|_1+N}t}}\|e^{td_I\mathcal{L}}\bm I(0)\|_{\infty}\le e^{-\frac{\tau m\bm\alpha_m\bm\beta_m}{\|\bm\zeta\|_1+N}t}\|\bm I(0)/\bm\alpha\|_{\infty}\|\bm\alpha\|_{\infty}\to 0 \quad \text{as}\ t\to\infty.
$$
This along with Lemma \ref{lem2} and the fact that $\sum_{j\in\Omega}S_j=N-\sum_{j\in\Omega}I_i\to N$ as $t\to\infty$ imply that $\bm S(t)\to N\bm\alpha$ as $t\to\infty$.

\medskip

\noindent{\bf Case 2.} Next, suppose that $ 0<\tau< 1$. 
We introduce the change of variables 
$$
\bm Z(t)=(1-\tau)\bm S(t), \quad \bm V=\tau \bm \zeta +\tau \bm I(t) \quad \text{and}\quad \bm F(t)=\bm\beta\circ\bm I(t)/(\bm \zeta +\bm S(t)+\bm I(t))\quad \forall\ t\ge 0.
$$
Hence, thanks to \eqref{Eq5} and \eqref{Eq6} and the fact that $\mathcal{L}\bm\zeta =0$, we have that 
\begin{equation}\label{Eq7}
    \begin{cases}
    \bm Z'(t)=d_S\mathcal{L}\bm Z +(1-\tau)(\bm V(t)-\bm Z(t))\circ \bm F(t)\quad t>0,\cr 
    \bm V'(t)=d_I\mathcal{L}\bm V+\tau(\bm Z(t)-\bm V(t))\circ\bm F(t)\quad t>0.
    \end{cases}
\end{equation}

\noindent Treating $\bm F(t)\gg 0$, for all $t>0$, as given in \eqref{Eq7}, then system \eqref{Eq7} is a cooperative system. Next, define 
\begin{equation}\label{Eq10-1}
c_*(t)=\min\Big\{(\bm Z(t)/\bm \alpha)_m, (\bm V(t)/\bm \alpha)_m\Big\}\quad \text{and}\quad c^*(t)=\max\Big\{(\bm Z(t)/\bm \alpha)_M, (\bm V(t)/\bm \alpha)_M\Big\}\quad \forall\ t\ge 0.
\end{equation}
Fix $t_0> 0$. Observe that $({\bm Z}_{*,t_0}(t),\bm V_{*,t_0}(t)):=(c_*(t_0)\bm\alpha,c_*(t_0)\bm\alpha)$,  $t\ge t_0$,   solves \eqref{Eq7} on $(t_0,\infty)$ and $({\bm Z}_{*,t_0}(t),\bm V_{*,t_0}(t))\le (\bm Z(t_0),\bm V(t_0))$. Hence, by the comparison principle for cooperative systems, we have that 
$$
({\bm Z}_{*,t_0}(t),\bm V_{*,t_0}(t))\le (\bm Z(t),\bm V(t)) \quad \forall\ t\ge t_0.
$$
Thus 
\begin{equation}\label{Eq8}
    c_*(t_0)\le c_*(t)\quad  \forall\ t\ge t_0>0. 
\end{equation}
Similarly, observe that $({\bm Z}^*_{t_0}(t),\bm V^*_{t_0}(t)):=(c^*(t_0)\bm\alpha,c^*(t_0)\bm\alpha)$,  $t\ge t_0$ ,  solves \eqref{Eq7} on $(t_0,\infty)$ and also $({\bm Z}^*_{t_0}(t),\bm V^*_{t_0}(t))\ge (\bm Z(t_0),\bm V(t_0))$. Hence, by the comparison principle for cooperative systems, we have that 
$$
({\bm Z}^*_{t_0}(t),\bm V^*_{t_0}(t))\ge (\bm Z(t),\bm V(t)) \quad \forall\ t> t_0.
$$
Thus 
\begin{equation}\label{Eq9}
    c^*(t_0)\ge c^*(t)\quad  \forall\ t\ge t_0>0. 
\end{equation} 
Since $t_0> 0$ was arbitrary fixed, it follows from \eqref{Eq8}-\eqref{Eq9} that $c_*(t)$ and $c^*(t)$ are monotone nondecreasing and nonincreasing in $t>0$, respectively. Hence, 
\begin{equation}\label{Eq10}
    \tilde{c}_{*}:=\sup_{t\ge 0}c_*(t)=\lim_{t\to\infty} c_*(t)\quad \text{and}\quad \tilde{c}^*=\inf_{t>0}c^*(t)=\lim_{t\to\infty}c^*(t).
\end{equation}
From this point, we distinguish two subcases. 

\medskip

\noindent{\bf Sub-case 1.} Here, we suppose that $ (1-\tau)N\le \tau m$ and show that $\bm I(t)\to \bm 0$ as $t\to\infty$.  Suppose to the contrary that the latter conclusion is false. Hence, thanks to \eqref{Eq2}, there is a sequence $\{t_n\}_{n\ge 1}$ converging to infinity such that 
\begin{equation}\label{Eq11}
\inf_{n\ge 1}\bm I_m(t_n)>0.
\end{equation}
Since $\mathcal{E}$ is invariant for system \eqref{model}, then $\sup_{t>0}\max_{i\in\Omega}(|\frac{d S_i(t)}{dt}|+|\frac{d I_i(t)}{dt}|)<\infty$. Therefore, by the Arzela-Ascoli theorem, possibly after passing to a subsequence, we may suppose that there is $(\bm S^{\infty}(t),\bm I^{\infty}(t))$, of class $C^1$, such that 
$$
(\bm S(t+t_n),\bm I(t+t_n))\to (\bm S^{\infty}(t),\bm I^{\infty}(t)) \quad \text{as} \ n\to \infty,
$$
locally uniformly in $t\in \mathbb{R}$. Moreover,  $(\bm S^{\infty}(t),\bm I^{\infty}(t))$ is an entire solution of \eqref{model}. As a result, $(\bm Z(t+t_n),\bm V(t+t_n))\to ((1-\tau)\bm S^{\infty}(t),\tau\bm\zeta+\tau\bm I^{\infty}(t) ):=(\bm Z^{\infty}(t),\bm V^{\infty}(t))$ and $\bm F(t+t_n)\to \bm\beta\circ \bm I^{\infty}(t)/(\bm \zeta +\bm S^{\infty}(t)+\bm I^{\infty}(t)):=\bm F^{\infty}(t)$ as $n\to\infty$, locally uniformly on $\mathbb{R}$. Furthermore,

\begin{equation}\label{Eq12}
    \begin{cases}
       \frac{d\bm Z^{\infty}}{dt}=d_S\mathcal{L}\bm Z^{\infty}+(1-\tau)(\bm V^{\infty}-\bm Z^{\infty})\circ\bm F^{\infty}(t) & t\in\mathbb{R},\cr 
       \frac{d\bm V^{\infty}}{dt}=d_I\mathcal{L}\bm V^{\infty}+\tau(\bm Z^{\infty}-\bm V^{\infty})\circ\bm F^{\infty}(t) & t\in\mathbb{R}.
    \end{cases}
\end{equation}
By \eqref{Eq11} and the fact that $\bm I(t_n)\to \bm I^{\infty}(0)$ as $n\to\infty$, then $\bm I^{\infty}(0)\gg\bm0$. This along with the fact that $(\bm S^{\infty}(t),\bm I^{\infty}(t))$ is an entire solution of \eqref{model} and strict positivity of the $\{e^{t\mathcal{L}}\}_{t>0}$ that ${\bm I}^{\infty}(t)\gg\bm 0$ for all $t\in\mathbb{R}$, and hence $\bm F^{\infty}(t)\gg\bm 0$ for all $t\in\mathbb{R}$. Next, observe that from \eqref{Eq10-1} and \eqref{Eq10}
\begin{equation}\label{Eq14-1}
\tilde{c}_*=\lim_{n\to\infty}c_*(t+n)=\min\Big\{(\bm Z^{\infty}(t)/\bm\alpha)_m, (\bm V^{\infty}(t)/\bm\alpha)_m\Big\}\quad \forall\ t\in\mathbb{R}
\end{equation}
and 
\begin{equation}\label{Eq14-2}
\tilde{c}^*=\lim_{n\to\infty}c^*(t+n)={\max}\Big\{(\bm Z^{\infty}(t)/\bm\alpha)_{M}, (\bm V^{\infty}(t)/\bm\alpha)_{M}\Big\}\quad \forall\ t\in\mathbb{R}.
\end{equation}
Hence,
\begin{equation*}
    (\tilde{c}_*\bm \alpha,\tilde{c}_*\bm \alpha)\le (\bm Z^{\infty}(t),\bm V^{\infty}(t))\le (\tilde{c}^*\bm \alpha,\tilde{c}^*\bm\alpha)\quad \forall\ t\in\mathbb{R}.
\end{equation*}
Now, we claim that 
\begin{equation}\label{Eq15}
    (\tilde{c}_*\bm \alpha,\tilde{c}_*\bm \alpha)= (\bm Z^{\infty}(t),\bm V^{\infty}(t))\quad \forall\ t\in\mathbb{R}.
\end{equation}
Suppose to the contrary that \eqref{Eq15} is false. Thus there is $t_0\in\mathbb{R}$ such that $(\tilde{c}_*\bm \alpha,\tilde{c}_*\bm \alpha)< (\bm Z^{\infty}(t_0),\bm V^{\infty}(t_0))$.  Set $\tilde{\bm Z}=e^{Mt}(\bm Z^{\infty}(t)-\tilde{c}_*\bm\alpha)$ and $\tilde{\bm V}=e^{Mt}(\bm V^{\infty}(t)-\tilde{c}_*\bm\alpha) $ for $t>t_0$, where $M> \sup_{t\in\mathbb{R}}\|\bm F{^{\infty}}(t)\|_{\infty}$ is fixed.  Observing that 
$$
\begin{cases}
    \frac{d\tilde{\bm Z}}{dt}\ge d_S\mathcal{L}\tilde{\bm Z} +{(1-\tau)}\bm F{^{\infty}}(t)\circ \tilde{\bm V} & t\in\mathbb{R}, \cr
    \frac{d\tilde{\bm V}}{dt}\ge d_I\mathcal{L}\tilde{\bm V} +{\tau}\bm F{^{\infty}}(t)\circ \tilde{\bm Z} & t\in\mathbb{R},
\end{cases}
$$
and recalling that $\mathcal{L}$ generates a strongly positive matrix semigroup,  then 
$$
\tilde{\bm Z}(t)\ge e^{(t-t_0)d_S\mathcal{L}}\tilde{\bm Z}(t_0)+(1-\tau)\int_{t_0}^te^{(t-s)d_S\mathcal{L}}\bm F^{\infty}(s)\circ \tilde{\bm V}(s)ds\quad \forall\ t>t_0$$   
{and} 
$$\tilde{\bm V}(t)\ge e^{(t-t_0)d_S\mathcal{L}}\tilde{\bm V}(t_0)+\tau\int_{t_0}^te^{(t-s)d_S\mathcal{L}}\bm F^{\infty}(s)\circ \tilde{\bm Z}(s)ds \quad \forall\ t>t_0.
$$
Therefore, since  $\bm F^{\infty}(t)\gg \bm 0$, $\tilde{\bm Z}(t)\ge \bm 0$, and $\tilde{\bm V}(t)\ge \bm 0$  for all $t\in\mathbb{R}$, $\mathcal{L}$ generates a strongly positive matrix semigroup, and $(\bm 0,\bm 0)<(\tilde{\bm Z}(t_0),\tilde{\bm V}(t_0))$, we conclude that  $(\bm 0,\bm 0)\ll (\tilde{\bm Z}(t),\tilde{\bm V}(t))$ for all $t>t_0$.
This in turn implies that
$$
\tilde{c}_*< \min\Big\{(\bm Z^{\infty}(t)/\bm\alpha)_m, (\bm V^{\infty}(t)/\bm\alpha)_m\Big\}\quad \forall\ t>t_0,
$$
which contradicts with \eqref{Eq14-1}. Therefore, \eqref{Eq15} holds. Now, by \eqref{Eq15} and \eqref{Eq14-2}, we have that 
\begin{equation}\label{Eq16}
    (\tilde{c}^*\bm \alpha,\tilde{c}^*\bm \alpha)= (\bm Z^{\infty}(t),\bm V^{\infty}(t))\quad \forall\ t\in\mathbb{R}.
\end{equation}
Therefore, 
$$((1-\tau)\bm S^{\infty}(t), \tau \bm\zeta +\tau\bm I^{\infty}(t))=(\bm Z^{\infty}(t),\bm V^{\infty}(t))=  (\tilde{c}^*\bm \alpha,\tilde{c}^*\bm \alpha)= (\tilde{c}_*\bm \alpha,\tilde{c}_*\bm \alpha)\quad \forall\ t\in\mathbb{R},$$
or equivalently,
\begin{equation}\label{Eq17}
(\bm S^{\infty}(t),\bm I^{\infty}(t))=\Big(\frac{\tilde{c}^*}{1-\tau}\bm \alpha, \big(\frac{\tilde{c}^*}{\tau}-m\big)\bm \alpha\Big)\quad \forall\ t\in\mathbb{R},
\end{equation}
 where we  used the fact that $\bm\zeta=m\bm\alpha$.   Recalling that $\bm I^{\infty}(t)\gg\bm 0$ for all $t\in\mathbb{R}$, it follows from \eqref{Eq17} that 
 \begin{equation*}
     \frac{\tilde{c}^*}{\tau}-m>0.
 \end{equation*}
Thanks to \eqref{Eq17} again, and  recalling that $(\bm S^{\infty}(t),\bm I^{\infty}(t))\in\mathcal{E}$ for all $t\in\mathbb{R}$, and $\sum_{j\in\Omega}\alpha_j=1$, then 
$$
N=\frac{\tilde{c}^*}{1-\tau}+\frac{\tilde{c}^*}{\tau}-m=\frac{\tilde{c}^*}{(1-\tau)\tau}-m,
$$
from which it follows that
\begin{equation}\label{Eq20}
\tau m-(1-\tau)N=\tau m-\frac{\tilde{c}^*}{\tau}+(1-\tau)m=-\Big(\frac{\tilde{c}^*}{\tau}-m\Big).
\end{equation}
We then conclude from \eqref{Eq17} that $\tau m<(1-\tau)N$, which is contrary to our initial assumption. Therefore, it holds that $\bm I(t)\to \bm 0$ as $t\to\infty$.  This along with Lemma \ref{lem2} implies that $\bm S(t)\to N\bm\alpha$ as $t\to\infty$.

\medskip

\noindent{\bf Sub-case 2.} Next, we suppose that $(1-\tau)N>\tau m$ and show that $\bm S\to (N-h)\bm \alpha$ and $\bm I\to h\bm\alpha$ as $t\to\infty$ where $h=(1-\tau)N-\tau m$.  To this end, we first show that 
\begin{equation}\label{Eq19}
    \liminf_{t\to\infty}{\bm I_{m}(t)}>0.
\end{equation}
Since $(1-\tau)N>\tau m$, it follows from \eqref{Eq18} that $\mathcal{R}_0>1$. Thus, \eqref{Eq19} follows from \eqref{TH1-2-eq1}.  Next, let $\{t_n\}_{n\ge 1}$ be any sequence converging to infinity. By the similar arguments as in sub-case 1, possibly after passing to a subsequence, there is an entire solution $(\bm S^{\infty}(t),\bm I^{\infty}(t))$ of \eqref{model} such that $(\bm S(t+t_n),\bm I(t+t_n))\to (\bm S^{\infty}(t),\bm I^{\infty}(t) )$ as $n\to\infty$, for $t$ locally uniformly in $\mathbb{R}$. Furthermore, $(\bm Z(t+t_n),\bm V(t+t_n))\to ((1-\tau)\bm S^{\infty}(t),\tau\bm\zeta+\tau\bm I^{\infty}(t) ):=(\bm Z^{\infty}(t),\bm V^{\infty}(t))$ and $\bm F(t+t_n)\to \bm\beta\circ \bm I^{\infty}(t)/(\bm \zeta +\bm S^{\infty}(t)+\bm I^{\infty}(t)):=\bm F^{\infty}(t)$ as $n\to\infty$, locally uniformly on $\mathbb{R}$, and $(\bm Z^{\infty}(t),\bm V^{\infty}(t))$ is an entire solution of \eqref{Eq12}. Moreover, since \eqref{Eq19} holds, then $\inf_{t\in\mathbb{R}}\min_{i\in\Omega}{F^{\infty}_i(t)}>0$. We can conclude that \eqref{Eq15} and \eqref{Eq16} hold, from which we derive that \eqref{Eq17} also holds. Note from \eqref{Eq20} that 
$$
\frac{\tilde{c}^*}{\tau}-m=(1-\tau)N-\tau m=:h\quad \text{and}\quad  \frac{\tilde{c}^*}{1-\tau}=\tau(N+m)=N-h.
$$
Hence, by \eqref{Eq17}, $(\bm S^{\infty}(t),\bm I^{\infty}(t))=((N-h)\bm\alpha, h\bm \alpha)$ for all $t\in\mathbb{R}$. Noting that $((N-h)\bm\alpha, h\bm \alpha) $ is independent of the sequence $\{t_n\}_{n\ge 1}$, which was arbitrary chosen, we conclude that $(\bm S(t),\bm I(t))\to ((N-h)\bm\alpha, h\bm \alpha)$ as $t\to\infty$. This completes the proof of sub-case 2.

\medskip

Now thanks to \eqref{Eq18}, we see that if $\mathcal{R}_0\le 1$, then either case 1 or sub-case 1 holds, and thus $(\bm S(t),\bm I(t))\to (N\bm\alpha,\bm 0)$ as $t\to\infty$. So {\rm (i)} is proved.  Note again from \eqref{Eq18} that if $\mathcal{R}_0>1$, then sub-case 2 holds and then $(\bm S(t),\bm I(t))\to ((N-h)\bm \alpha, h\bm \alpha)$ as $t\to\infty$, where $h=(1-\tau)N-\tau m$. In the later case we have that $ ((N-h)\bm \alpha, h\bm \alpha)$ is the unique EE solution of \eqref{model}, which completes the proof of {\rm (ii)}.
    
\end{proof}

\begin{proof}[Proof of Theorem \ref{TH4}] Suppose that $d:=d_S=d_I$. Set $\bm Q:=\bm S+\bm I$. Summing up the two equations of \eqref{model}, we get 
\begin{equation*}
    \bm Q'(t)=d\mathcal{L}\bm Q(t)\quad t>0.
\end{equation*}
Hence, since $\sum_{j\in\Omega}Q_j(t)=\sum_{j\in\Omega}(S_j(t)+I_j(t))=N$ for all $t\ge 0$, we conclude from Lemma \ref{lem2} that 
\begin{equation}\label{Eqa2-4}
    \lim_{t\to\infty}\bm Q(t)=N\bm\alpha.
\end{equation}
Observe that 
\begin{equation}\label{Eqa2-6}
    \bm I'(t)=d\mathcal{L}\bm I(t)+(\bm \beta\circ(\bm Q(t)-\bm I )/(\bm \zeta +\bm Q(t)) -\bm \gamma)\circ \bm I\quad t>0.
\end{equation}
Next, for every $|\varepsilon|\ll N$, let $\hat{\bm I}^{(\varepsilon)}$ denote the unique nonnegative stable solution of 
\begin{equation}\label{Eqa2-7}
0=d\mathcal{L}\hat{\bm I}+(\bm\beta\circ((N+\varepsilon)\bm\alpha -\hat{\bm I})/(\bm\zeta+(N+\varepsilon)\bm\alpha)-\bm \gamma)\circ\hat{\bm I}. 
\end{equation}
Note that $\hat{\bm I}^{(\varepsilon)}$ is nondecreasing and continuous in $|\varepsilon|\ll N$.  Thanks to \eqref{Eqa2-4} and \eqref{Eqa2-6} and the comparison principle for cooperative systems, we have that 
$$
\hat{\bm I}^{(-\varepsilon)}\le \liminf_{t\to\infty}\bm I(t)\le \limsup_{t\to\infty}\bm I(t)\le \hat{\bm I}^{(\varepsilon)}\quad \forall\ 0<\varepsilon\ll N.
$$
Letting $\varepsilon\to 0^+$, then $\bm I(t)\to \hat{\bm I}^{(0)}$ as $t\to\infty$. Now, if $\mathcal{R}_0\le 1$, then $\sigma_*(d_I\mathcal{L}+{\rm diag}(N\bm\beta\circ\bm\alpha/(\bm\zeta+N\bm\alpha)-\bm \gamma))\le 0$, in which case $\hat{\bm I}=\bm 0$ is the unique nonnegative solution of \eqref{Eqa2-7} for $\varepsilon=0$. Thus $\hat{\bm I}^{(0)}=\bm 0$ and $\|\bm I(t)\|_{\infty}\to0$ as $t\to\infty$ when $\mathcal{R}_0\le 1$. In this case, it follows from \eqref{Eqa2-4} that $\bm S(t)\to N\bm \alpha$ as $t\to\infty$, which proves  {\rm(i)}.

\medskip

\noindent Next, suppose that $\mathcal{R}_0>1$ and hence $\sigma_*(d_I\mathcal{L}+{\rm diag}(N\bm\beta\circ\bm\alpha/(\bm\zeta+N\bm\alpha)-\bm r))>0$. It follows from classical results on logistic-type reaction equations that system \eqref{Eqa2-7} has a unique positive solution. Note also from Theorem \ref{TH1-2}-{\rm (ii)} that  $\hat{\bm I}^{(0)}\gg \bm 0 $  since $\mathcal{R}_0>1$, $\bm  I (t)\to \hat{\bm I}^{(0)} $ and \eqref{TH1-2-eq1} holds.  Therefore, $\bm I(t)$ converges to the unique positive solution $\hat{\bm I}^{(0)}$ of \eqref{Eqa2-7} with $\varepsilon=0$ as $t\to\infty$. In this case, thanks again to \eqref{Eqa2-4}, we have that $(\bm S(t),\bm I(t))\to (N\bm\alpha-\hat{\bm I}^{(0)},\hat{\bm I}^{(0)})$ as $t\to\infty$, where $\hat{\bm I}^{(0)}$ is the unique positive solution of \eqref{Eqa2-7}. Since $(N\bm\alpha-\hat{\bm I}^{(0)},\hat{\bm I}^{(0)})$ is independent of the initial data, then it is the unique EE solution of \eqref{model}. 
    
\end{proof}

\subsection{Proof of Theorems \ref{TH5} and \ref{TH6}}
For every $l>0$, consider the system of algebraic equations in $\bm P\ge \bm 0$:

\begin{equation}\label{transform equation}
    0=d_I\mathcal{L}\bm P +\big(l\bm\beta\circ(\bm \alpha-d_I\bm P)/(\bm\zeta +l(\bm\alpha-d_I\bm P)+d_Sl\bm P) -\bm\gamma\big)\circ\bm P.
\end{equation}

\begin{lem} \label{lem3} Fix $d_I>0$ and $d_S>0$.  Let $\hat{\mathcal{R}}_0$ be defined by \eqref{R-hat-0}.
    \begin{itemize}
    \item[\rm (i)] If $\hat{\mathcal{R}}_0\le 1$, then system \eqref{transform equation} has no positive solution for any $l>0$.
        \item[\rm (ii)] Suppose that $\hat{\mathcal{R}}_0>1$ and let  $\mathcal{N}_0=\mathcal{N}_0(d_I,\bm\zeta)$ be as in Proposition \ref{prop2}-{\rm (iii)}.
        \begin{itemize}
            \item[\rm (ii-1)] System \eqref{transform equation} has no positive solution if $l\le \mathcal{N}_0$.
            \item[\rm (ii-2)] If $l>\mathcal{N}_0$, then system \eqref{transform equation} has exactly one positive solution ${\bm P}^{(l)}\gg\bm 0$. Furthermore, the function $(\mathcal{N}_0,\infty)\ni l\mapsto {\bm P}^{(l)}$ is analytic, strictly increasing, 
            \begin{equation}\label{Eqb-1}
            \bm 0\ll d_I{\bm P}^{(l)}\ll\bm \alpha, \quad    \lim_{l\to \mathcal{N}_0^+}{\bm P}^{(l)}=\bm 0\quad \text{and}\quad \lim_{l\to\infty}{\bm P}^{(l)}=\bm P^*,
            \end{equation}
            where $\bm0\ll \bm P^*\ll \frac{1}{d_{I}}\bm \alpha$ is the unique  positive solution of 
            \begin{equation}\label{Eqb-2}
                0=d_I\mathcal{L}\bm P+\big(\bm\beta\circ(\bm\alpha-d_I\bm P)/(\bm\alpha-d_I\bm P+d_S\bm P)-\bm\gamma\big)\circ\bm P
            \end{equation}
            \item[\rm (ii-3)] Let  $\bm\eta\gg 0$ satisfying $\|\bm \eta\|_1=1$ be the positive eigenvector associated with $\sigma_*\big(d_I\mathcal{L}+{\rm diag}(\big(\mathcal{N}_0\bm\beta\circ\bm\alpha/(\bm \zeta+\mathcal{N}_0\bm\alpha )-\bm\gamma\big))\big)$. Let also $\bm\eta^*\gg 0$ satisfying $\|\bm \eta^*\|_1=1$ be the positive eigenvector associated with $\sigma_*\big(d_I\mathcal{L}^T+{\rm diag}(\big(\mathcal{N}_0\bm\beta\circ\bm\alpha/(\bm \zeta+\mathcal{N}_0\bm\alpha )-\bm\gamma\big))\big)$. Then
            \begin{equation}\label{Eqb-3}
                \lim_{l\to \mathcal{N}_0^+}\frac{{\bm P}^{(l)}}{l-\mathcal{N}_0}=\lim_{l\to \mathcal{N}_0^+}\frac{d{\bm P}^{(l)}}{dl}=\left(\frac{\sum_{j\in\Omega}\frac{\beta_j\alpha_j\zeta_j\eta_j\eta^*_j}{(\zeta_j+\mathcal{N}_0\alpha_j)^2}}{\sum_{j\in\Omega}\frac{\mathcal{N}_0\eta_j^2\beta_j(d_I\zeta_j+\mathcal{N}_0d_S\alpha_j){\eta^*_j}}{(\zeta_j+\mathcal{N}_0\alpha_j)^2}}\right)\bm\eta.
            \end{equation}
        \end{itemize}
    \end{itemize}
\end{lem}
\begin{proof} Define 
\begin{equation*}
    \mathcal{F}(l,\bm P)=l\bm\beta\circ(\bm\alpha-d_I\bm P)/(\bm\zeta +l(\bm\alpha-d_I\bm P)+d_Sl\bm P) -\bm\gamma\quad \quad -\frac{1}{d_Sl}\bm\zeta<\bm P<\frac{1}{d_I}\bm \alpha ,\ l>0, 
\end{equation*}
so that  \eqref{transform equation} can be written as 
$$
0=d_I\mathcal{L}\bm P+\bm P\circ\mathcal{F}(l,\bm P).
$$
Note that $\mathcal{F}$ is analytic and 

\begin{equation}\label{Eqb-7}
    \partial_{ p_j}\mathcal{F}_i(l,\bm P)=\begin{cases}-\frac{ l\beta_i(d_I(\zeta_i +d_Sl P_i)+ld_S(\alpha_i -d_I P_i))}{(\zeta_i +l(\alpha_i-d_IP_i)+ld_SP_i)^2} & \text{if}\ j=i,\cr 
    0 & \text{if}\ j\ne i
    \end{cases}\quad -\frac{1}{d_Sl}\bm\zeta<\bm P<\frac{1}{d_I}\bm \alpha,\  l>0.
\end{equation}
Therefore, $\mathcal{F}(l,\bm P)<\mathcal{F}(l,\tilde{\bm P} )$ for all $l>0$, $\bm 0\le \tilde{\bm P}\ll\bm P< \frac{1}{d_I}\bm \alpha$. Thus, by the classical results on the logistic-type equations, \eqref{transform equation} has a (unique) positive solution ${\bm P}^{(l)}$ if and only if $\sigma^{(l)}_*:=\sigma_*(d_I\mathcal{L}+{\rm diag}(l\bm\beta\circ\bm\alpha/(\bm \zeta+l\bm\alpha)-\bm\gamma))>0$. Taking $l=N$ in \eqref{R-0}, and then set $\mathcal{R}_0^{(l)}:=\mathcal{R}_0$ to indicate the dependence of $\mathcal{R}_0$ on $l>0$, we  have from Proposition \ref{prop2}-{\rm (i)} that $\mathcal{R}_0^{(l)}-1$ and $\sigma_*^{(l)}$ have the same sign. Therefore, \eqref{transform equation} has a (unique) positive solution ${\bm P}^{(l)}$ if and only if $\mathcal{R}_0^{(l)}>1$.

\medskip

{\rm (i)} Suppose that $\hat{\mathcal{R}}_0\le 1$. Since by Proposition \eqref{prop2}-{\rm (iii)} $\mathcal{R}_0^{(l)}<\hat{\mathcal{R}}_0$ for all $l>0$, then by the previous development, we have that \eqref{transform equation} has no positive solution for every $l>0$. 

\medskip

{\rm (ii)} Suppose that $\hat{\mathcal{R}}_0>1$ and let $\mathcal{N}_0$ be given by Proposition \eqref{prop2}-{\rm (iii)}. Hence, if $0<l\le \mathcal{N}_0$, we have that $\mathcal{R}_0^{(l)}\le 1$ and system \eqref{transform equation} has no positive solution. However, if $l>\mathcal{N}_0$, we have that $\mathcal{R}_0^{(l)}>1$, and hence \eqref{transform equation} has a unique positive solution ${\bm P}^{(l)}$. Furthermore, since {\bf (A1)} holds, we have that ${\bm P}^{(l)}\gg \bm 0$. It is easy to see that $\bm P:=\frac{1}{d_I}\bm \alpha$ is a strict super-solution of \eqref{transform equation}, and hence ${\bm P}^{(l)}\ll \frac{1}{d_I}\bm \alpha$  for all $l>\mathcal{N}_0$. Here, we have used the fact $\bm P^{(l)}$, $l>\mathcal{N}_0$,  is the unique positive and linearly/globally stable solution of the initial value problem associated with \eqref{transform equation} with respect to solutions with positive initials. 
Observe that 
\begin{equation}\label{Eqb-9}
    \partial_l\mathcal{F}(l,\bm P)=\bm\beta\circ\bm\zeta\circ(\bm\alpha-d_I\bm P)/((\bm\zeta+l(\bm\alpha-d_I\bm P)+d_Sl\bm P)^2)\gg 0 \quad \forall\ l>0,\,\ \ \bm 0\le \bm P\ll \frac{1}{d_I}\bm \alpha.
\end{equation}
Therefore, by the comparison principle, we have that ${\bm P}^{(l)}\ll \bm P^{(\tilde{l})}$ whenever $\mathcal{N}_0<l<\tilde{l}$. It follows from classical theory of the logistic-type equations that ${\bm P}^{(l)}$ is linearly stable, and by the implicit function theorem and the analycity of $\mathcal{F}$, we have that ${\bm P}^{(l)}$ is an analytic function of $l>\mathcal{N}_0$. Since \eqref{transform equation} has no positive solution for $l=\mathcal{N}_0$, and $\bm 0\ll {\bm P}^{(l)}\ll \frac{1}{d_I}\bm \alpha$ for all $l>\mathcal{N}_0$, we deduce that the limit in the middle term of \eqref{Eqb-1} holds. Next, since ${\bm P}^{(l)}$ is strictly increasing in $l>\mathcal{N}_0$ and bounded above by $\frac{1}{d_I}\bm\alpha$, it converges to some $\bm P^*\in (\bm 0,\frac{1}{d_I}\bm\alpha]$ as $l\to\infty$. Observing that 
\begin{align*}
l\bm\beta\circ(\bm\alpha-d_I{\bm P}^{(l)})/(\bm\zeta +l(\bm\alpha-d_I{\bm P}^{(l)})+d_Sl{\bm P}^{(l)})=&\bm\beta(\bm \alpha-d_I{\bm P}^{(l)})/(\frac{1}{l}\bm\zeta+(\bm \alpha-d_I{\bm P}^{(l)})+d_S{\bm P}^{(l)})\cr 
\to& \bm\beta(\bm\alpha-d_I\bm P^*)/(\bm\alpha-d_I\bm P^*+d_S\bm P^*) \quad \text{as}\quad l\to\infty,
\end{align*}
taking limit as $l\to\infty$ in \eqref{transform equation}, we have that $\bm P^*$ solves \eqref{Eqb-2}.  It is easy to see that $\bm P=\frac{1}{d_I}\bm\alpha$ is a strict super-solution of \eqref{Eqb-2}, hence $\bm P^*\ll \frac{1}{d_I}\bm \alpha$ by {\bf (A1)}.

\medskip

To prove \eqref{Eqb-3}, we use the bifurcation theory. To this end, set 

$$
\mathcal{H}(l,\bm P):=d_I\mathcal{L}\bm P+\big(l\bm\beta\circ(\bm \alpha-d_I\bm P)/(\bm\zeta +l(\bm\alpha-d_I\bm P)+d_Sl\bm P) -\bm\gamma\big)\circ\bm P\quad l>0, \ -\frac{1}{d_Sl}\bm\zeta<\bm P<\frac{1}{d_I}\bm \alpha.
$$
Thus $\mathcal{H}$ is an analytic function and $\mathcal{H}(l,{\bm P}^{(l)})=\bm 0$ for all $l>\mathcal{N}_0$. Note also that $\mathcal{H}(l,\bm 0)=\bm 0$ for all $l>0$. By computations, 
 for every $l>0$, $-\frac{1}{d_Sl}\bm\zeta<\bm P<\frac{1}{d_I}\bm \alpha$, we have that
\begin{equation}\label{Eqb-8}
\partial_{\bm P}\mathcal{H}(l,\bm P)[\bm Q]=d_I\mathcal{L}\bm Q+\bm Q\circ\mathcal{F}(l,\bm P)+ \bm P\circ\partial_{\bm P}\mathcal{F}(l,\bm P)[\bm Q] \quad \forall\ \bm Q\in \mathbb{R}^n. 
\end{equation}
Thus 
$$
\partial_{\bm P}\mathcal{H}(\mathcal{N}_0,\bm 0)[\bm Q]=d_I\mathcal{L}\bm Q+\big(\mathcal{N}_0\bm\beta\circ\bm\alpha/(\bm \zeta+\mathcal{N}_0\bm\alpha )-\bm\gamma\big)\circ\bm Q\quad  \forall\  \bm Q\in \mathbb{R}^n.
$$
By Proposition \ref{prop2}-{\rm (iii)} and the Perron-Frobeinus Theorem, we have that $\sigma_*(\partial_{\bm P}\mathcal{H}(\mathcal{N}_0,\bm 0))=0$ is a simple eigenvalue of $\partial_{\bm P}\mathcal{H}(\mathcal{N}_0,\bm 0)$ and 
\begin{equation}\label{Eqb-11}
{\rm Ker}(\partial_{\bm P}\mathcal{H}(\mathcal{N}_0,\bm 0))={\rm span}(\bm \eta)\quad \text{and}\quad {\rm Range}(\partial_{\bm P}\mathcal{H}(\mathcal{N}_0,\bm 0))={\rm span}(\bm\eta^*)^T,
\end{equation}
where ${\rm span}(\bm\eta^*)^T$ is the orthogonal complement of ${\rm span}(\bm\eta^*)$. Taking the partial derivative of \eqref{Eqb-8} with respect to $l>0$, we have that
$$
\partial_{(l,\bm P)}\mathcal{H}(l,\bm P)[\bm Q]=\bm Q\circ\partial_l\mathcal{F}(l,\bm P)+\bm P\circ\partial_{(l,\bm P)}\mathcal{F}(l,\bm P)[\bm Q],\ \quad \ l>0,\ \ \bm Q\in\mathbb{R}^n,\ -\frac{1}{d_Sl}\bm\zeta<\bm P<\frac{1}{d_I}\bm \alpha. 
$$
Thus, recalling \eqref{Eqb-9}, we have that 
$$
\partial_{(l,\bm P)}\mathcal{H}(\mathcal{N}_0,\bm 0)[\bm \eta]=\bm\beta\circ\bm\zeta\circ\bm\alpha\circ\bm\eta/((\bm \zeta+\mathcal{N}_0\bm\alpha)\circ(\bm \zeta+\mathcal{N}_0\bm\alpha)),
$$
from which it follows 
\begin{equation}\label{Eqb-14}
<\partial_{(l,\bm P)}\mathcal{H}(\mathcal{N}_0,\bm 0)[\bm \eta],\bm\eta^*>=\sum_{j\in\Omega}\frac{\beta_j\alpha_j\zeta_j\eta_j\eta^*_j}{(\zeta_j+\mathcal{N}_0\alpha_j)^2}>0,
\end{equation}
and hence, by \eqref{Eqb-11}, $\partial_{(l,\bm P)}\mathcal{H}(\mathcal{N}_0,\bm 0)[\bm \eta]\notin {\rm Range}(\partial_{\bm P}\mathcal{H}(\mathcal{N}_0,\bm 0))$. Fix a complement subspace $\mathbb{X}\subset\mathbb{R}^n$ of ${\rm Ker}(\partial_{\bm P}\mathcal{H}(\mathcal{N}_0,\bm 0))$. Then, by \cite[Theorem 1.7]{CR1971}, the solution set of $\mathcal{H}(l,\bm P)=\bm 0$ near $(\mathcal{N}_0,\bm 0)$ consists of precisely the two curves $\mathcal{C}_0:=\{(\mathcal{N}_0,\bm 0): \ l>0\}$ and  $\mathcal{C}_1:=\{(h(s),s\bm\eta+s \tilde{\bm P}(s)) : |s|<\varepsilon\}$, where $h : (-\varepsilon,\varepsilon)\to \mathbb{R}$ and $ \tilde{\bm P} : (-\varepsilon,\varepsilon)\to \mathbb{X}$ are analytic functions satisfying $h(0)=\mathcal{N}_0$, $\tilde{\bm P}(0)=\bm 0$. Furthermore, 
 \begin{equation}\label{Eqb-13}
     \dot{h}(0)=-\frac{<\partial_{(\bm P,\bm P)}\mathcal{H}\big(\mathcal{N}_0,\bm 0\big)[\bm\eta,\bm\eta],\eta^*>}{2<\partial_{(l,\bm P)}\mathcal{H}\big(\mathcal{N}_0,\bm 0\big)[\bm \eta],\eta^*>}.
 \end{equation} 
 Note from \eqref{Eqb-8} that
 \begin{equation*}
\partial_{(\bm P,\bm  P)}\mathcal{H}(l,\bm P)[\bm Q,\tilde{\bm Q}]=\bm Q\circ\partial_{\bm P}\mathcal{F}(l,\bm P)[\tilde{\bm Q}] +\tilde{\bm Q}\circ\partial_{\bm P}\mathcal{F}(l,\bm P)[\bm Q]+\bm P\circ \partial_{(\bm P,\bm P)}\mathcal{F}(l,\bm P)[\bm Q,\tilde{\bm Q}] \quad \forall\ \bm Q,\ \tilde{\bm Q}\in \mathbb{R}^n,
\end{equation*}
so that 
$$
\partial_{(\bm P,\bm  P)}\mathcal{H}(\mathcal{N}_0,\bm 0)[\bm \eta,{\bm \eta}]= 2\bm \eta\circ\partial_{\bm P}\mathcal{F}(\mathcal{N}_0,\bm 0)[\bm \eta].
$$
We use \eqref{Eqb-7} to get
$$
\partial_{\bm P}\mathcal{F}(\mathcal{N}_0,\bm 0)[\bm \eta]=-\mathcal{N}_0\bm\beta\circ\bm \eta\circ (d_I\bm\zeta +\mathcal{N}_0d_S\bm \alpha)/((\bm \zeta+\mathcal{N}_0\bm\alpha)\circ(\bm \zeta+\mathcal{N}_0\bm\alpha)),
$$
which in turn gives
$$
\partial_{(\bm P,\bm  P)}\mathcal{H}(\mathcal{N}_0,\bm 0)[\bm \eta,{\bm \eta}]=-2\mathcal{N}_0\bm\beta\circ\bm \eta\circ\bm\eta\circ (d_I\bm\zeta +\mathcal{N}_0d_S\bm \alpha)/((\bm \zeta+\mathcal{N}_0\bm\alpha)\circ(\bm \zeta+\mathcal{N}_0\bm\alpha)).
$$
This along with \eqref{Eqb-14} and \eqref{Eqb-13} yields that
$$
\dot{h}(0)=\frac{\sum_{j\in\Omega}\frac{\mathcal{N}_0\eta_j^2\beta_j(d_I\zeta_j+\mathcal{N}_0d_S\alpha_j){\eta^*_j}}{(\zeta_j+\mathcal{N}_0\alpha_j)^2}}{\sum_{j\in\Omega}\frac{\beta_j\alpha_j\zeta_j\eta_j\eta^*_j}{(\zeta_j+\mathcal{N}_0\alpha_j)^2}}>0.
$$
Taking $\varepsilon>0$ sufficiently small, we have that {$h$} is strictly increasing on $[-\varepsilon,-\varepsilon]$ and $s\bm\eta+s\tilde{\bm P}(s)\gg \bm 0$ for ${0<s}\le \varepsilon$. Denoting by ${h}^{-1}$, the inverse function of {$h$}, we have that 
 $$
 \bm P^{(l)}=h^{-1}(l)(\bm\eta+\tilde{\bm P}(h^{-1}(l)))\quad \mathcal{N}_0<l<h({\varepsilon}),
 $$
 is the unique positive solution of \eqref{transform equation}. This shows that the function $[\mathcal{N}_0,{h}(\varepsilon))\ni l\mapsto {\bm P}^{(l)}$ is analytic and 
 \begin{equation*}
     \lim_{l\to \mathcal{N}_0^+}\partial_l{\bm P}^{(l)}=\lim_{l\to \mathcal{N}_0^+}\frac{{\bm P}^{(l)}}{l-\mathcal{N}_0}=\lim_{l\to\mathcal{N}_0^+}\frac{h^{-1}(l)}{l-\mathcal{N}_0}(\bm\eta+\tilde{\bm P}(h^{-1}(l)))=\frac{1}{\dot{h}(0)}\bm\eta=\frac{\sum_{j\in\Omega}\frac{\beta_j\alpha_j\zeta_j\eta_j\eta^*_j}{(\zeta_j+\mathcal{N}_0\alpha_j)^2}}{\sum_{j\in\Omega}\frac{\mathcal{N}_0\eta_j^2\beta_j(d_I\zeta_j+\mathcal{N}_0d_S\alpha_j){\eta^*_j}}{(\zeta_j+\mathcal{N}_0\alpha_j)^2}}\bm\eta,
 \end{equation*}
which gives the desired result.
    
\end{proof}

Thanks to Lemma \ref{lem3}, when $\hat{\mathcal{R}}_0>1$, we introduce the function
\begin{equation}\label{mathcal-G-de}
    \mathcal{G}(l)=l\sum_{j\in\Omega}((\alpha_j-d_IP^{(l)}_{j})+d_SP_{j}^{(l)})\quad \forall\ l\ge \mathcal{N}_0,
\end{equation}
where $\bm P^{(\mathcal{N}_0)}:=\bm 0$.   
\begin{lem}\label{lem4}  Fix $d_I>0$ and $d_S>0$. Suppose that $\hat{\mathcal{R}}_0>1$, and let $\mathcal{G}$ be defined by \eqref{mathcal-G-de}. Then $\mathcal{G}$ is continuously differentiable (in fact, analytic on $(\mathcal{N}_0,\infty)$). Furthermore, 
\begin{equation}\label{Eqc-4}
    \frac{d\mathcal{G}(l)}{dl}=\sum_{j\in\Omega}((\alpha_j-d_IP^{(l)}_j)+d_SP^{(l)}_j)+l(d_S-d_I)\sum_{j\in\Omega}\frac{dP^{(l)}_j}{dl}\qquad l\ge \mathcal{N}_0,
\end{equation}
\begin{align}\label{Eqc-5}
    \frac{d\mathcal{G}(\mathcal{N}_0)}{dl}=&1+\frac{(d_S-d_I)\sum_{j\in\Omega}\frac{\beta_j\alpha_j\zeta_j\eta_j\eta^*_j}{(\zeta_j+\mathcal{N}_0\alpha_j)^2}}{\sum_{j\in\Omega}\frac{\eta_j^2\beta_j(d_I\zeta_j+\mathcal{N}_0d_S\alpha_j)\eta^*_j}{(\zeta_j+\mathcal{N}_0\alpha_j)^2}}\cr
=&\frac{d_I\sum_{j\in\Omega}\frac{\zeta_j\eta_j\eta_j^*\beta_j(\eta_j-\alpha_j)}{(\zeta_j+\mathcal{N}_0\alpha_j)^2}+d_S\sum_{j\in\Omega}\frac{\beta_j\eta_j\eta_j^*\alpha_j(\mathcal{N}_0\eta_j+\zeta_j)}{(\zeta_j+\mathcal{N}_0\alpha_j)^2}}{\sum_{j\in\Omega}\frac{\eta_j^2\beta_j(d_I\zeta_j+\mathcal{N}_0d_S\alpha_j)\eta^*_j}{(\zeta_j+\mathcal{N}_0\alpha_j)^2}},
\end{align}
where $\bm \eta$ and $\bm\eta^*$ are as in Lemma \ref{lem3}-{\rm (ii-3)}, and 

\begin{equation}\label{Eqc-6}
    \mathcal{G}(\mathcal{N}_0)=\mathcal{N}_0\quad \text{and}\quad \lim_{l\to\infty}\mathcal{G}(l)=\infty.
\end{equation}

\end{lem}
\begin{proof} The regularity of $\mathcal{G}$ in $l\ge \mathcal{N}_0$ follows from that of the mapping ${\bm P}^{(l)}$ on $l\ge \mathcal{N}_0$. Taking the derivative of \eqref{mathcal-G-de} with respect to $l$ yields \eqref{Eqc-4}. Direct computations from \eqref{Eqc-4} along with \eqref{Eqb-3} yields \eqref{Eqc-5}. Since $\bm P^{(\mathcal{N}_0)}=\bm 0$, then $\mathcal{G}(\mathcal{N}_0)=\mathcal{N}_0$. Finally, since by Lemma \ref{lem3}-{\rm (ii)}, ${\bm P}^{(l)}\to \bm P^*$  as $l\to \infty$ where $\bm 0\ll\bm P^*\ll \frac{1}{d_{I}}\bm \alpha$ is the unique positive solution of \eqref{Eqb-2}, then 
$$
\mathcal{G}(l)\ge l\|\bm\alpha-d_I{\bm P}^{(l)}\|_1\to \infty \quad \text{as} \ l \to \infty, 
$$
which completes the proof of the result.
    
\end{proof}

\noindent The next result connects the last two lemmas \ref{lem3} and \ref{lem4} to the existence of EE solutions of system \eqref{model}.

\begin{lem}\label{lem5} Fix $d_I>0$, $d_S>0$ and $N>0$. Let $\mathcal{G}$ be defined as in \eqref{mathcal-G-de}.
\begin{itemize}
    \item[\rm(i)] If $(\bm S,\bm I)$ is an EE solution of system \eqref{model}, then there is a positive constant $\kappa>0$ such that 
    \begin{equation}\label{Eqc-1}
        d_S\bm S+d_I\bm I=\kappa\bm\alpha.
    \end{equation}
    Furthermore, setting 
    \begin{equation}\label{Eqc-2}
        \tilde{\bm S}=\frac{1}{\kappa}\bm S\quad \text{and}\quad \tilde{\bm I}=\frac{1}{\kappa}\bm I,
    \end{equation}
    then 
    \begin{equation}\label{Eqc-3}
        \tilde{\bm S}=\frac{1}{d_S}(\bm\alpha -d_{I}\tilde{\bm I}),
    \end{equation} 
    $\tilde{\bm I}\gg \bm 0$ solves \eqref{transform equation} with $l=\frac{\kappa}{d_S}$, and $\mathcal{G}(\frac{\kappa}{d_S})=N$. In particular, $\frac{\kappa}{d_S}>\mathcal{N}_0$.

    \item[\rm (ii)] Suppose that for some $l>\mathcal{N}_0$, $\mathcal{G}(l)=N$. Then $(\bm S,\bm I)=(l(\bm \alpha-d_I{\bm P}^{(l)}),d_Sl{\bm P}^{(l)} )$ is an EE  of system \eqref{model}.
    
\end{itemize}
    
\end{lem}
\begin{proof}{\rm(i)} Let $(\bm S,\bm I)$ be an EE solution of \eqref{model}. Adding up the two equations of \eqref{EE-system}, we get 
$$
0=\mathcal{L}(d_S\bm S+d_I\bm I).
$$
Hence, zero is a simple eigenvalue of $\mathcal{L}$ and $\mathcal{L}\bm \alpha=0$ with $\bm\alpha\ne\bm 0$, there is a $\kappa\in \mathbb{R}$ such that $d_S\bm S+d_I\bm I=\kappa\bm\alpha$. Clearly, $\kappa>0$ since $\bm S\gg \bm 0$ and $\bm I\gg 0$. This shows that \eqref{Eqc-1} holds. Defining $(\tilde{\bm S},\tilde{\bm I})$ as in \eqref{Eqc-2} , then \eqref{Eqc-3} is obtained by \eqref{Eqc-2}. Furthermore, replacing $\bm S=\frac{\kappa}{d_S}(\bm \alpha-d_I\tilde{\bm I})$ in the second equation of \eqref{EE-system}, and then divide the resulting equation by $\kappa$, we see that $\tilde{\bm I}$ is a positive solution of \eqref{transform equation} for $l=\frac{\kappa}{d_S}$, and hence $\tilde{\bm I}=\bm P^{(\frac{\kappa}{d_S})}$. This shows that $\frac{\kappa}{d_S}>\mathcal{N}_0$ by Lemma \ref{lem3}-{\rm (ii)}. Finally, since $\sum_{j\in\Omega}(S_j+I_j)=N$, then
$$
N=\sum_{j\in\Omega}(S_j+I_j)=\sum_{j\in\Omega}(\frac{\kappa}{d_S}(\alpha_j-d_I\tilde{I}_j)+\kappa\tilde{I}_j)=\frac{\kappa}{d_S}\sum_{j\in\Omega}((\alpha_j-d_I P_j^{(\frac{\kappa}{d_S})})+d_S P^{(\frac{\kappa}{d_S})}_j)=\mathcal{G}(\frac{\kappa}{d_S}).
$$
This completes the proof of {\rm(i)}.

\medskip

\noindent{\rm (ii)} It can easily be verified.

\end{proof}

Thanks to Lemma \ref{lem5}, we see that system \eqref{model} has an EE solution if and only if there is $l>\mathcal{N}_0$ such that $\mathcal{G}(l)=N$. Using this fact, we can now present the proofs of Theorem \ref{TH5}.

\begin{proof}[Proof of Theorem \ref{TH5}] Fix $N>0$, $d_S>0$, $d_I>0$ and suppose that $\bm\mu=\bm 0$, {\bf (A1)-(A2)} holds, and $\hat{\mathcal{R}}_0>1$.

\medskip

\noindent{\rm (i)} Since by Lemma \ref{lem5}, system \eqref{model} has an EE solution if and only if there is $l>\mathcal{N}_0$ such that $\mathcal{G}(l)=N$, (in which case $(l(\bm\alpha-d_I{\bm P}^{(l)}),d_Sl{\bm P}^{(l)})$ is an EE solution), and the mapping $(\mathcal{N}_0,\infty)\ni l\mapsto {\bm P}^{(l)}$ is strictly increasing by Lemma \ref{lem3}-{\rm (ii-2)}, then to show the uniqueness of  EE solution of system \eqref{model}, whenever it exists,  is equivalent to establishing that $\mathcal{G}$ is strictly increasing. Now, by \eqref{Eqc-4}, we have that $\frac{d\mathcal{G}(l)}{dl}>0$ for all $l>\mathcal{N}_0$ if $d_S\ge d_I$. Therefore, when $d_S\ge d_I$, system \eqref{model} has no EE if $N\le \mathcal{G}(\mathcal{N}_0)=\mathcal{N}_0$, and has a unique EE if $N>\mathcal{G}(\mathcal{N}_0)=\mathcal{N}_0$. 
This shows that assertion {\rm (i)} holds. 

\medskip

 \noindent {\rm (ii)} Next, suppose that  $N(\bm 1-2\bm r)\circ\bm\alpha\ge \bm r\circ\bm \zeta$. Then $\min_{j\in\Omega}\frac{N\alpha_j}{r_j(\zeta_j+N\alpha_j)}>1$, and hence 
 \begin{align}
\frac{\sum_{j\in\Omega}\frac{N\beta_j\alpha_j^2}{\zeta_j+N\alpha_j}}{\sum_{j\in\Omega}\gamma_j\alpha_j}{= \frac{\sum_{j\in\Omega}\Big(\frac{N\alpha_j}{r_j(\zeta_j+N\alpha_j)}\Big)\gamma_j\alpha_j}{\sum_{j\in\Omega}\gamma_j\alpha_j}
\ge\frac{\min_{k\in\Omega}\Big(\frac{N\alpha_k}{r_k(\zeta_k+N\alpha_k)}\Big)\sum_{j\in\Omega}\gamma_j\alpha_j}{\sum_{j\in\Omega}\gamma_j\alpha_j}} =\min_{j\in\Omega}\frac{N\alpha_j}{r_j(\zeta_j+N\alpha_j)}>1.
 \end{align}
It then follows from  \eqref{R-0-limit}  that $\mathcal{R}_0>1$ for all $d_I>0$. Now, fix $d_S>0$ and $d_I>0$. If $d_S\ge d_I$ we know from {\rm (i)} that system \eqref{model} has a unique EE solution. So, we focus on the case of $0<d_S<d_I$. Thanks to Theorem \ref{TH1-2}-{\rm(ii)}, we know that system \eqref{model} has at least one EE solution. We proceed by contradiction to establish that the EE solution is unique. To this end, thanks to Lemma \ref{lem5}, we suppose to the contrary that there exist $l_2>l_1>\mathcal{N}_0$ such that ${N}:=\mathcal{G}(l_1)=\mathcal{G}(l_2)$. Set $\bm I^{(i)}=d_Sl_i\bm P^{(l_i)}$ and $\bm S^{(i)}=l_i(\bm\alpha-d_I\bm P^{(l_i)})$ for $i=1,2$. Since $l_1<l_2$ and the mapping ${\bm P}^{(l)}$ is strictly increasing in $l\ge \mathcal{N}_0$, then $\bm0\ll\bm I^{(1)}\ll \bm I^{(2)}$. Define 
$$
\nu:=(\bm I^{(1)}/\bm I^{(2)})_m.
$$
Then $0<\nu<1$ and $\bm 0\ll \nu\bm I^{(2)}\le \bm I^{(1)}$. Note  that 
$$
d_S\bm S^{(i)}+d_I\bm I^{(i)}=d_Sl_i\bm\alpha\quad i=1,2,
$$
and hence 
\begin{equation}\label{Eqd-8-3}
d_Sl_i=d_I\|\bm I^{(i)}\|_1+d_S\|\bm S^{(i)}\|_1=(d_I-d_S)\|\bm I^{(i)}\|_1+d_S{N}\quad i=1,2.
\end{equation}
From the last two equations, we get
$$
\bm S^{(i)}={N}\bm\alpha -\frac{d_I}{d_S}\Big(\bm I^{(i)}- \Big(1-\frac{d_S}{d_I}\Big)\|\bm I^{(i)}\|_1\bm\alpha \Big)\quad i=1,2.
$$
As a result, for each $i=1,2$, and $j\in\Omega$,
\begin{equation}\label{Eqd-7}
    0=d_I\sum_{k\in\Omega}L_{jk} I^{(i)}_k+\beta_j\left(\frac{\Big({N}\alpha_j -\frac{d_I}{d_S}\Big( I^{(i)}_j- \Big(1-\frac{d_S}{d_I}\Big)\|\bm I^{(i)}\|_1\alpha_j \Big)\Big)}{\Big( \zeta_j+{N}\alpha_j -\frac{d_I}{d_S}\Big( I^{(i)}_j- \Big(1-\frac{d_S}{d_I}\Big)\|\bm I^{(i)}\|_1\alpha_j \Big)+ I^{(i)}_j\Big)}- r_j\right) I^{(i)}_j.
\end{equation}
First, we claim that 
\begin{equation}\label{Eqd-8-0}
    \bm I^{(i)}\ll N\bm\alpha \quad i=1,2.
\end{equation}
Indeed, since for each $i=1,2$, $d_I\bm P^{(l_i)}\ll\bm\alpha$, then  by \eqref{Eqd-8-3}
\begin{align*}
\bm I^{(i)}=\frac{l_id_S}{d_I}d_I\bm P^{(l_i)}\ll\frac{l_id_S}{d_I}\bm\alpha =\Big(\Big(1-\frac{d_S}{d_I}\Big)\|\bm I^{(i)}\|_1+\frac{d_S}{d_I}N\Big)\bm\alpha\ll \Big(\Big(1-\frac{d_S}{d_I}\Big)N+\frac{d_S}{d_I}N\Big)\bm\alpha=N\bm\alpha.
\end{align*}
Secondly, we claim that
\begin{equation}\label{Eqd-8}
\bm I^{(i)}\gg \Big(1-\frac{d_S}{d_I}\Big)\|\bm I^{(i)}\|_1\bm\alpha\quad \quad \quad  i=1,2.
\end{equation}
Indeed, for clarity, we fix $i=1,2$, set $\tau_i=(1-\frac{d_S}{d_I})\|\bm I^{(i)}\|_1>0$. We treat $\tau_i>0$ as given in \eqref{Eqd-7}.  We rewrite \eqref{Eqd-7} as 
\begin{equation}\label{Eqd-13}
    \bm 0=d_I\mathcal{L}\bm I^{(i)}+\bm\beta\circ\Big(\Big(\tilde{\bm A}+\frac{d_I}{d_S}\tau_i\bm\alpha-\frac{d_I}{d_S}\bm I^{(i)}\Big)/\Big(\tilde{\bm B}^{(i)}+\frac{d_I}{d_S}\tau_i\bm\alpha-\frac{d_I}{d_S}\bm I^{(i)}\Big)-\bm r\Big)\circ\bm I^{(i)}
\end{equation}
where $\tilde{\bm A}={N}\bm \alpha \gg \bm0$ and $\tilde{\bm B}^{(i)}=\tilde{\bm A}+\bm\zeta+\bm I^{(i)}\gg \tilde{\bm A}$. 
Noting that for every $0<a<b$, the function 
\begin{equation*}
(0,a)\ni x\mapsto g(x):= \frac{a-x}{b-x}
\end{equation*}
is strictly decreasing, and $\mathcal{L}$ is quasipositive and irreducible matrix, it follows from classical results on logistic-type monotone dynamical systems that $\bm I^{(i)}$ is the unique strictly positive and linearly/globally stable solution of \eqref{Eqd-13}. Therefore, any strict subsolution of \eqref{Eqd-13} must be strictly less than $\bm I^{(i)}$. Hence, we just need to establish that $\tau_i\bm\alpha$ is a strict subsolution of \eqref{Eqd-13}. Replace $\bm I^{(i)}$ by $\tau_i \bm\alpha$ in the right hand side of \eqref{Eqd-13} gives 
\begin{align*}
    &d_I\mathcal{L}(\tau_i\bm \alpha)+\bm\beta\circ\Big(\Big(\tilde{\bm A}+\frac{d_I}{d_S}\tau_i\bm\alpha-\frac{d_I}{d_S}\tau_i\bm \alpha\Big)/\Big(\tilde{\bm B}^{(i)}+\frac{d_I}{d_S}\tau_i\bm\alpha-\frac{d_I}{d_S}\tau_i\bm \alpha\Big)-\bm r\Big)\circ(\tau_i\bm \alpha) 
    = \bm\beta\circ\Big(\Big(\tilde{\bm A}/ \tilde{\bm B}^{(i)}\Big)-\bm r\Big)\circ(\tau_i\bm \alpha).
\end{align*}
Hence, noting by \eqref{Eqd-8-0} that 
\begin{align*}
    \tilde{\bm A}-\bm r\circ\tilde{\bm B}^{(i)}=&  {N}(\bm 1-\bm r)\circ\bm\alpha -\bm r\circ\bm\zeta -\bm r\circ\bm I^{(i)}\cr 
    \gg & {N}(\bm 1-\bm r)\circ\bm\alpha -\bm r\circ\bm\zeta -{N}\bm r\circ\bm \alpha={N}(\bm 1-2\bm r)\circ\bm\alpha-\bm r\circ\bm\zeta\ge \bm 0,
\end{align*}
then $\tau_i\bm\alpha$ is a strict subsolution of \eqref{Eqd-13}. Hence \eqref{Eqd-8} holds.

\medskip

Next, recalling that the function $g$ above is strictly decreasing (for any choice of $0<a<b$), it follows from \eqref{Eqd-7}, \eqref{Eqd-8}, and the fact that $0<\nu<1$ that, for each $j=1,2$,  
\begin{align}\label{Eqd-10}
     0< & d_I\sum_{k\in\Omega}L_{jk}(\nu I^{(2)}_k) +\beta_j\left(\frac{\Big({N}\alpha_j -\nu\frac{d_I}{d_S}\Big( I^{(2)}_j- \Big(1-\frac{d_S}{d_I}\Big)\|\bm I^{(2)}\|_1\alpha_j \Big)\Big)}{\Big( \zeta_j+{N}\alpha_j -\nu\frac{d_I}{d_S}\Big( I^{(2)}_j- \Big(1-\frac{d_S}{d_I}\Big)\|\bm I^{(2)}\|_1\alpha_j \Big)+ I^2_j\Big)}-r_j\right)(\nu I^{(2)}_j)\cr 
    =& d_I\sum_{k\in\Omega}{L}_{jk}(\nu I^{(2)}_k)
    +\beta_j\left(\frac{\Big(\Big({N}\alpha_j -\nu\frac{d_I}{d_S} I^{(2)}_j\Big)+ \Big(\frac{d_I}{d_S}-1\Big)\|\nu\bm I^{(2)}\|_1\alpha_j \Big)}{\Big( \zeta_j+\Big({N}\alpha_j -\nu\frac{d_I}{d_S} I^{(2)}_j\Big)+ \Big(\frac{d_I}{d_S}-1\Big)\|\nu\bm I^{(2)}\|_1\alpha_j + I^{(2)}_j\Big)}-r_j\right)(\nu I^{(2)}_j).
\end{align}
Observing also that, for $j=1,2$, 
\begin{align*}
\frac{\Big(\Big({N}\alpha_j -\nu\frac{d_I}{d_S} I^{(2)}_j\Big)+ \Big(\frac{d_I}{d_S}-1\Big)\|\nu\bm I^{(2)}\|_1\alpha_j\Big) }{\Big( \zeta_j+\Big({N}\alpha_j -\nu\frac{d_I}{d_S} I^{(2)}_j\Big)+ \Big(\frac{d_I}{d_S}-1\Big)\|\nu\bm I^{(2)}\|_1\alpha_j + I^{(2)}_j\Big)}
\le \frac{\Big(\Big({N}\alpha_j -\nu\frac{d_I}{d_S} I^{(2)}_j\Big)+ \Big(\frac{d_I}{d_S}-1\Big)\|\bm I^{(1)}\|_1\alpha_j\Big) }{\Big( \zeta_j+\Big({N}\alpha_j -\nu\frac{d_I}{d_S} I_j^{(2)}\Big)+ \Big(\frac{d_I}{d_S}-1\Big)\|\bm I^{(1)}\|_1\alpha_j + I^{(2)}_j\Big)}
\end{align*}
since $\|\nu\bm I^{(2)}\|_1\le \|\bm I^{(1)}\|_1$, we conclude from \eqref{Eqd-10} and the fact that $\bm 0\ll \bm I^{(1)}\ll \bm I^{(2)}$ that
\begin{align*}
    0<&   d_I\sum_{k\in\Omega}{L}_{jk}(\nu  I^{(2)}_k )+\beta_{j}\left(\frac{\Big({N}\alpha_j -\nu\frac{d_I}{d_S} I^{(2)}_j\Big)+ \Big(\frac{d_I}{d_S}-1\Big)\|\bm I^{(1)}\|_1\alpha_j }{ \zeta_j+\Big({N}\alpha_{j} -\nu\frac{d_I}{d_S} I^{(2)}_j\Big)+ \Big(\frac{d_I}{d_S}-1\Big)\|\bm I^{(1)}\|_1\alpha_j + I^{(2)}_j}-r_j\right)(\nu I^{(2)}_j)\cr 
    < &  d_I\sum_{k\in\Omega}{L}_{jk}(\nu  I^{(2)}_k )+\beta_j\left(\frac{\Big({N}\alpha_j -\nu\frac{d_I}{d_S} I^{(2)}_j\Big)+ \Big(\frac{d_I}{d_S}-1\Big)\|\bm I^{(1)}\|_1\alpha_j }{ \zeta_j+\Big({N}\alpha_j -\nu\frac{d_I}{d_S} I^{(2)}_j\Big)+ \Big(\frac{d_I}{d_S}-1\Big)\|\bm I^{(1)}\|_1\alpha_j + I^{(1)}_j}- r_j\right)(\nu I^{(2)}_j) \quad  j=1,2.
\end{align*}
Therefore, setting 
$$\bm B:= \bm \zeta+{N}\bm\alpha + \Big(\frac{d_I}{d_S}-1\Big)\|\bm I^{(1)}\|_1\bm\alpha +\bm I^{(1)}\gg\bm 0\quad 
\text{and} \quad 
\bm A:={N}\bm\alpha + \Big(\frac{d_I}{d_S}-1\Big)\|\bm I^{(1)}\|_1\bm\alpha \gg \bm 0,
$$
we have that 
\begin{equation}\label{Eqd-11}
0=d_I\mathcal{L}\bm I^{(1)} +\bm\beta\circ\Big((\bm A-{\frac{d_{I}}{d_{S}}}\bm I^{(1)})/(\bm B-{\frac{d_{I}}{d_{S}}}\bm I^{(1)})-\bm r\Big)\circ\bm I^{(1)}
\end{equation}
and 
$$
\bm 0\ll d_I\mathcal{L}(\nu \bm I^{(2)}) +\bm\beta\circ\Big((\bm A-{\frac{d_{I}}{d_{S}}}\nu\bm I^{(2)})/(\bm B-{\frac{d_{I}}{d_{S}}}\nu \bm I^{(2)})-\bm r\Big)\circ(\nu\bm I^{(2)}).
$$
Therefore, since $\mathcal{L}$ is quasipositive and irreducible, $\bm I^{(1)}$ is the positive and locally/globally stable positive solution of \eqref{Eqd-11},  we must have that $\nu\bm I^{(2)}\ll \bm I^{(1)}$. Hence $\nu<(\bm I^{(1)}/\bm I^{(2)})_m$, which yields a contradiction to the definition of $\nu$.  Thus, there is a unique $l>\mathcal{N}_0$ satisfying $\mathcal{G}(l)=N$;  hence system \eqref{model} has a unique EE.



 \end{proof}

\medskip

\noindent Next, we give a proof of Theorem \ref{TH6}.

\begin{proof}[Proof of Theorem \ref{TH6}] Let $\mathcal{G}$ be defined as in \eqref{mathcal-G-de}, and define 
$$
\mathcal{R}_{\min}=\rho({\rm diag}(N_{\min}\bm\alpha\circ\bm\beta/(\bm\zeta+N_{\min}\bm\alpha))V^{-1})\quad \text{where}\quad N_{\min}=\min_{l\ge \mathcal{N}_0}\mathcal{G}(l).
$$
Since $\mathcal{G}$ is continuous and \eqref{Eqc-6} holds, then it achieves it minimal value. Hence $N_{\min}>0$ is well defined, and so is $\mathcal{R}_{\min}$ and $\mathcal{R}_{\min}>0$.  It is clear from $N_{\min}\le \mathcal{G}(\mathcal{N}_0)=\mathcal{N}_0$. Hence, by Proposition \ref{prop2}-{\rm (iii)}, we have that $\mathcal{R}_{\min}\le 1$. Thanks to Lemma \ref{lem5}, system \eqref{model} has no EE solution if $N<N_{\min}$. Thus, since $\mathcal{R}_0$ is strictly increasing and continuous in $N>0$,  system \eqref{model} has no EE solution if $\mathcal{R}_{0}<\mathcal{R}_{\min}$.  Next, by \eqref{Eqc-6} and the intermediate value theorem, for every $N>N_{\min}$, there is $l_{N}>\mathcal{N}_0$ such that $\mathcal{G}(l_{N})=N$, and hence system \eqref{model} has an EE solution.  Therefore, recalling also that $\mathcal{R}_0\to \hat{\mathcal{R}}_0$ as $N\to\infty$, and $\mathcal{R}_0$ is continuous in $N$, we deduce that for every $\mathcal{R}_0\in (\mathcal{R}_{\min},\hat{\mathcal{R}}_0)$, system \eqref{model} has at least one EE solution.

\medskip

\noindent Thanks to Lemma \ref{lem5} again, we see that 
$$ 
\mathcal{C}_*:=\{(\mathcal{R}_0,\bm S,\bm I)=(\sigma_*^{(\mathcal{G}(l))}, l(\bm\alpha-d_I{\bm P}^{(l)}),d_Sl{\bm P}^{(l)}) :\ l>\mathcal{N}_0\}
$$
is a simple curve of EE solution of system \eqref{model}. Here, recall that $\sigma^{(N)}_*:=\sigma_*(d_I\mathcal{L}+{\rm diag}(N\bm\beta\circ\bm\alpha/(\bm \zeta+N\bm\alpha)-\bm\gamma))>0$ for all $N>0$. Moreover, by Lemma \ref{lem5}-{\rm (i)} again, any EE solution of system \eqref{model} belongs to the curve $\mathcal{C}_*$. Observing that 
$$
(\sigma_*^{(\mathcal{G}(l))}, l(\bm\alpha-d_I{\bm P}^{(l)}),d_Sl{\bm P}^{(l)})\to (1,\mathcal{N}_0\bm\alpha,\bm 0)\quad \text{as}\quad l\to\mathcal{N}_0^+,
$$
then the curve $\mathcal{C}_*$ bifurcates of the set of DFE solutions at $\mathcal{R}_0=1$. By Lemma \ref{lem3}-{\rm (ii-2)}, we have  $\sum_{j\in\Omega}l P_j^{(l)}\to \infty$ and $\sum_{j\in\Omega}l(\alpha_j-d_I p^{(l)}_j)\to \infty$ as $l\to\infty$.  Recalling from \eqref{Eqc-6} that $\mathcal{G}(l)\to \infty$ as $l\to\infty$, we conclude from Proposition \eqref{prop2}-{\rm (iii)} that $f(l):=\sigma_*^{(\mathcal{G}(l))}\to \hat{\mathcal{R}}_0$ as $l\to\infty$. This  proves  \eqref{Th6-eq1}.

\medskip

\noindent Next, by the perturbation theory for the principal eigenvalue and the fact that $F$ is analytic in $N>0$,  we have that   $\sigma_*^{(N)}$ is analytic in $N$ (since $F$ defined by \eqref{R-0} is analytic in $N>0$), and hence $f(l)=\sigma_{*}^{(\mathcal{G}(l))}$ is analytic in $l>\mathcal{N}_0$ as the composition of two such functions. Observe also that since $F$ defined by \eqref{R-0} is strictly increasing in $N$ with $\partial_{N}F\gg\bm0 $ for all $N>0$, then    $\sigma_*^{(N)}$ is strictly  increasing in $N>0$,  with $\frac{d\sigma_*^{(N)}}{dN}>0$ for all $N>0$. Therefore,  the bifurcation direction at $\mathcal{R}_0=1$ is completely determined by the sign of $\frac{d\mathcal{G}(\mathcal{N}_0)}{dl}$.  Therefore, the conclusions {\rm (i)} and {\rm(ii)} easily follow from \eqref{Eqc-5}. It is also clear that \eqref{TH6-eq1} holds when $d_S\ge d_I$. Thus, $\mathcal{R}_0=1$ is a forward transcritical bifurcation point when $d_S\ge d_I$.

\end{proof}

{
\begin{proof}[Proof of Proposition \ref{prop3}]We suppose that $|\Omega|=2$, that is we have two patches, and $\mathcal{L}$ is symmetric. Then there is $L>0$ such that 
\begin{equation*}
    \mathcal{L}=\left(\begin{array}{cc}
      -L   & L \\
        L &  -L
    \end{array}\right),
\end{equation*}
and an easy computation gives $\alpha:=\alpha_1=\alpha_2=\frac{1}{2}.$
Next, suppose also that $\bm\zeta\in{\rm span}(\bm 1)$, $\bm\zeta\gg 0$, $\bm r\notin{\rm span}(\bm 1)$, $\gamma_1<\beta_1$, $\|\bm\gamma\|_1<\|\bm\beta\|_1$, and $\mathcal{N}_{\rm up}^*=\frac{\gamma_1\zeta_1}{(\beta_1-\gamma_1)\alpha_1}$, where $\mathcal{N}_{\rm up}^*$ is defined by \eqref{N-star-def}. Then $\|\bm \beta/\bm\gamma\|_{\infty}>1$, $1\in\tilde{\Omega}$ and there is $\zeta>0$ such that $\bm \zeta=\zeta{\bm 1}$. Moreover, since $\|\bm\gamma\|_1<\|\bm\beta\|_1$, then  $d_*=\infty$ in Proposition \ref{prop2}-{\rm (iv)}. Note that since $\bm r \notin{\rm span}(\bm 1)$, then by Remark \ref{Rk0}, the mapping $\mathcal{N}_0$ is strictly increasing in $d_I>0$. Recalling from Proposition \ref{prop2}-{\rm (iv-2)} that 
$$
\lim_{d_I\to 0^+}\mathcal{N}_0=\mathcal{N}^*_{\rm up}=\frac{\gamma_1\zeta}{(\beta_1-\gamma_1)\alpha},
$$
then 
\begin{equation*}
    \mathcal{N}_0>\mathcal{N}^*_{\rm up}\quad \forall\ d_I>0.
\end{equation*}
By Proposition \ref{prop2}-{\rm (iv)}, we have that 
$$
\lim_{d_I\to 0^+}\frac{\mathcal{N}_0\beta_{1}\alpha}{\zeta+\mathcal{N}_0\alpha}=\frac{\Big(\frac{\gamma_{1}\zeta}{(\beta_{1}-\gamma_{1})\alpha}\Big)\beta_{1}\alpha}{\zeta+\Big(\frac{\gamma_{1}\zeta}{(\beta_{1}-\gamma_{1})\alpha}\Big)\alpha}=\gamma_{1}.
$$
As a result, 
\begin{equation}\label{Appedndix-eq5-0}
    \gamma_{1}-\frac{\mathcal{N}_0\beta_{1}\alpha}{\zeta+\mathcal{N}_0\alpha}< 0\quad \forall\ d_I>0.
\end{equation}
Next, fix $d_I>0$. Then $ \mathcal{N}_0>0$ satisfies   $\sigma_*(d_I\mathcal{L}+{\rm diag}(\mathcal{N}_0\bm\beta\circ\bm\alpha/(\bm\zeta+\mathcal{N}_0\bm\alpha)-\bm \gamma))=0$, where  
$$
d_I\mathcal{L}+{\rm diag}(\mathcal{N}_0\bm\beta\circ\bm\alpha/(\bm\zeta+\mathcal{N}_0\bm\alpha)-\bm \gamma)=\left(
\begin{array}{cc}
    \frac{\mathcal{N}_0\beta_1\alpha}{\zeta+\mathcal{N}_0\alpha}-(d_IL+\gamma_1) & d_IL \\
     d_IL& \frac{\mathcal{N}_0\beta_2\alpha}{\zeta+\mathcal{N}_0\alpha} -(d_IL+\gamma_2)
\end{array}
\right).
$$
Therefore,  \begin{equation*}
    \Big(\frac{\mathcal{N}_0\beta_1\alpha}{\zeta+\mathcal{N}_0\alpha}-(d_IL+\gamma_1)\Big)\Big(\frac{\mathcal{N}_0\beta_2\alpha}{\zeta+\mathcal{N}_0\alpha} -(d_IL+\gamma_2)\Big)-(d_IL)(d_IL)=0,
\end{equation*}
\begin{equation}\label{App-eq2-0}
    \frac{\mathcal{N}_0\beta_1\alpha}{\zeta+\mathcal{N}_0\alpha} -(d_IL+\gamma_1)<0\quad \text{and}\quad \frac{\mathcal{N}_0\beta_2\alpha}{\zeta+\mathcal{N}_0\alpha} -(d_IL+\gamma_2)<0.
\end{equation}
Recall that $\bm\eta\gg 0$ is uniquely determined by  
\begin{equation}\label{Appedndix-eq4-0}
(d_I\mathcal{L}+{\rm diag}(\mathcal{N}_0\bm\beta\circ\bm\alpha/(\bm\zeta+\mathcal{N}_0\bm\alpha)-\bm \gamma))\bm\eta=\bm 0 \quad \text{and}\quad \sum_{i=1}^2\eta_i=1.
\end{equation}
Solving for $\bm \eta$ from \eqref{Appedndix-eq4-0}, we obtain that
\begin{equation}\label{Eqp-1}
    \eta_{1}=\frac{L}{{2L+\frac{1}{d_I}\Big(\gamma_{1}-\frac{\mathcal{N}_0\beta_{1}\alpha}{\zeta+\mathcal{N}_0\alpha}}\Big)}\quad \text{and}\quad \eta_{2}=\frac{{L+\frac{1}{d_I}\Big(\gamma_{1}-\frac{\mathcal{N}_0\beta_{1}\alpha}{\zeta+\mathcal{N}_0\alpha}\Big)}}{{2L+\frac{1}{d_I}\Big(\gamma_{1}-\frac{\mathcal{N}_0\beta_{1}\alpha}{\zeta+\mathcal{N}_0\alpha}}\Big)}.
\end{equation}
Now, since \eqref{Appedndix-eq5-0} holds, then 
$$
\eta_{1}=\frac{L}{{2L+\frac{1}{d_I}\Big(\gamma_{1}-\frac{\mathcal{N}_0\beta_{1}\alpha}{\zeta+\mathcal{N}_0\alpha}}\Big)}> \frac{L}{2L}=\alpha_{1}.
$$
Note also that since $0<L<2L$, the mapping $H:\ (-\infty,L)\ni x\mapsto \frac{L-x}{2L-x}$ is strictly decreasing. By \eqref{Appedndix-eq5-0},
\begin{align*}
\eta_{2}=&\frac{{L+\frac{1}{d_I}\Big(\gamma_{1}-\frac{\mathcal{N}_0\beta_{1}\alpha}{\zeta+\mathcal{N}_0\alpha}\Big)}}{{2L+\frac{1}{d_I}\Big(\gamma_{1}-\frac{\mathcal{N}_0\beta_{1}\alpha}{\zeta+\mathcal{N}_0\alpha}}\Big)}
= H\left(\frac{1}{d_I}\Big(\frac{\mathcal{N}_0\beta_{1}\alpha}{\zeta+\mathcal{N}_0\alpha}-\gamma_{1}\Big)\right) <H(0)= \frac{L}{2L}=\alpha_{2}.
\end{align*}
Note that $\bm\eta^*=\bm\eta$ since $\mathcal{L}$ is symmetric. Therefore,

\begin{equation*}
    \sum_{i=1}^2\frac{\zeta_i\eta_i\eta_i^*\beta_i(\eta_i-\alpha_i)}{(\zeta_i+\mathcal{N}_0\alpha_i)^2}=\frac{\zeta\eta_1^2{\beta_1}}{(\zeta+\mathcal{N}_0\alpha)^2}\Big((\eta_1-\alpha)+\Big(\frac{\eta_2}{\eta_1}\Big)^2\Big(\frac{\beta_2}{\beta_1}\Big)(\eta_2-\alpha)\Big).
\end{equation*}
Hence, if $\Big(\frac{\eta_2}{\eta_1}\Big)^2\Big(\frac{\beta_2}{\beta_1}\Big)\le 1$, then  since $\eta_2<\alpha$, we have 
$$
 \sum_{i=1}^2\frac{\zeta_i\eta_i\eta_i^*\beta_i(\eta_i-\alpha_i)}{(\zeta_i+\mathcal{N}_0\alpha_i)^2}\ge \frac{\zeta\eta_1^2{\beta_1}}{(\zeta+\mathcal{N}_0\alpha)^2}\Big((\eta_1-\alpha)+(\eta_2-\alpha)\Big)=\frac{\zeta\eta_1^2{\beta_1}}{(\zeta+\mathcal{N}_0\alpha)^2}\Big((\eta_1+\eta_2)-2\alpha)\Big)=0.
$$
In this case, \eqref{TH6-eq1} holds for any $d_S>0$. Hence $\mathcal{R}_0=1$ is a forward transcritical bifurcation point, which proves {\rm (i)}.

\medskip

Next, suppose that $\Big(\frac{\eta_2}{\eta_1}\Big)^2\Big(\frac{\beta_2}{\beta_1}\Big)> 1$. Then since $\eta_2<\alpha$, we have 
$$
 \sum_{i=1}^2\frac{\zeta_i\eta_i\eta_i^*\beta_i(\eta_i-\alpha_i)}{(\zeta_i+\mathcal{N}_0\alpha_i)^2}< \frac{\zeta\eta_1^2{\beta_1}}{(\zeta+\mathcal{N}_0\alpha)^2}\Big((\eta_1-\alpha)+(\eta_2-\alpha)\Big)=\frac{\zeta\eta_1^2{\beta_1}}{(\zeta+\mathcal{N}_0\alpha)^2}\Big((\eta_1+\eta_2)-2\alpha)\Big)=0.
$$
In this case, we define 
$$
d^*_{\rm up}=\frac{\zeta\sum_{i=1}^2\eta_i^2\beta_i(\alpha-\eta_i)}{\alpha\sum_{j=1}^2\beta_j\eta_j^2(\mathcal{N}_0\eta_j+\zeta)}d_I.
$$
Hence $0<d_{\rm up}^*<d_I$, and  the assertion {\rm (ii)} follows from \eqref{TH6-eq1} and \eqref{TH6-eq2}.

\medskip

Finally, it is clear from \eqref{Eqp-1} and \eqref{Appedndix-eq5-0} that 
$$ 
\eta_1>\eta_2 \quad \text{and}\quad  \frac{\eta_2}{\eta_1}=1-\frac{1}{d_IL}\Big(\frac{\mathcal{N}_0\beta_{1}\alpha}{\zeta+\mathcal{N}_0\alpha}-\gamma_{1}\Big)>1-\frac{1}{d_IL}(\beta_1-\gamma_1).
$$
Then \eqref{prop3-eq1} holds.

\medskip

\end{proof}

}

\subsection{Proofs of Theorems \ref{TH7} and \ref{TH8}}


    

\begin{proof}[Proof of Theorem \ref{TH7}] Assume that the hypotheses of the theorem hold. For every $d_S>0$, let $(\bm S,\bm I)$ be an EE solution of \eqref{model}. Then, for every $d_S>0$, by Lemma \ref{lem5}-{\rm (i)}, there is $\kappa>0$ such that \eqref{Eqc-1} holds. Thus 
since $\|\bm\alpha\|_1=\sum_{j\in\Omega}\alpha_j=1$, we have 
\begin{align}\label{Eqd-1}
\Big\|\bm I-(\sum_{j\in\Omega}I_j)\bm\alpha\Big\|_1\le& \Big\|\bm I-\frac{\kappa}{d_I}\bm\alpha\Big\|_1+\Big\|(\sum_{j\in\Omega}I_j-\frac{\kappa}{d_I})\bm\alpha\Big\|_1
=\|\bm I-\frac{\kappa}{d_I}\bm\alpha\|_1+\Big|\sum_{j\in\Omega}I_j-\frac{\kappa}{d_I}\Big| \cr
=&\Big\|\bm I-\frac{\kappa}{d_I}\bm\alpha\big\|_1+\Big|\|\bm I\|_{1}-\Big\|\frac{\kappa}{d_I}\bm\alpha\Big\|_1\Big| 
\le   2\|\bm I-\frac{\kappa}{d_I}\bm\alpha\|_1 
= 2 \frac{d_S}{d_I}\|\bm S\|_1\le \frac{2d_S}{d_I}N\to 0\ \text{as}\ d_S\to0^+.\cr
\end{align}
This proves the first result of the theorem.  Next, set 
$ M^*:=\limsup_{d_S\to0^+}\sum_{j\in\Omega}I_j.$  Note that $0\le M^*\le N$.  

\medskip

\noindent {\bf Claim 1.} If $M^*>0$, then $\bm r_M<1$, $N>\|\bm\zeta\circ\bm r/(\bm 1-\bm r)\|_1$, $M^*=\frac{(N-\|\bm\zeta\circ\bm r/(\bm 1-\bm r)\|_1)}{(1+\|\bm\alpha\circ\bm r/(\bm 1-\bm r)\|_1)}$. 

\medskip

\noindent So, suppose that $M^*>0$. Hence, possibly after passing to a subsequence, we may suppose that $\sum_{j\in\Omega}I_j\to M^*>0$ as $d_S\to 0^+$. Thus, from the first equation of \eqref{EE-system}, we have 
\begin{align}\label{Eqd-2}
   \| {\bm S}/(\bm \zeta +\bm S+\bm I)-{\bm r}\|_1=d_S\|(\mathcal{L}\bm S)/(\bm I\circ\bm \beta)\|_1\le d_S\frac{\|\mathcal{L}\|\|\bm S\|_1}{\bm I_m\bm\beta_m} \le d_S \frac{\|\mathcal{L}\|N}{\bm I_m\bm \beta_m}\to 0 \quad \text{as}\ d_S\to 0^+.
\end{align} 
This along with the fact that
$$\|\bm S/(\bm \zeta+\bm S+\bm I)\|_{\infty}\le \max_{j\in\Omega}\frac{N}{\zeta_j+N+I_j}\to \max_{j\in\Omega}\frac{N}{\zeta_j+N+M^*\alpha_j}<1 \quad \text{as}\ d_S\to 0^+, $$
implies that ${\bm r}_M<1$. 
Next, since $\bm r_M< 1$, it follows from \eqref{Eqd-1} and \eqref{Eqd-2} that 
\begin{equation}\label{Eqc-7}
\bm S\to (\bm\zeta +M^*\bm\alpha)\circ (\bm r/(\bm 1-{\bm r}))\quad \text{as} \ d_S\to 0^+, 
\end{equation}
from which it follows that

$$
N=\lim_{d_S\to 0^+}\sum_{j\in\Omega}(S_j+I_j)=\sum_{j\in\Omega}\frac{(\zeta _j+M^*\alpha_j)r_j}{1-r_j}+\sum_{j\in\Omega}M^*\alpha_j=\|\bm \zeta\circ \bm r/(\bm 1-\bm r)\|_1+M^*\Big(1+\|{\bm r\circ\bm\alpha}/({\bm 1-\bm r})\|_1\Big).
$$
Solving for $M^*$ in the last equation yields
$ M^*=({N-\|{\bm\zeta\circ\bm r}/{(\bm 1-\bm r)}}\|_1)/({ 1+\|{\bm \alpha\circ\bm r}/({\bm 1-\bm r})}\|_1).$  Recalling from our initial hypothesis that $M^*>0$, then we must have $N>\|{\bm\zeta\circ\bm r}/(\bm 1-\bm r)\|_1$. This completes the proof of Claim 1. 
Now, we proceed to prove {\rm (i)} and {\rm (ii)}.

\medskip

\noindent{\rm (i)} It is clear from Claim 1 that if either $\bm r_M\ge 1$ or $\bm r_M<1$ and $N\le\|{\bm\zeta\circ\bm r}/{(\bm 1-\bm r)}\|_1$, then $M^*=0$, which implies that $\|\bm I\|_1\to 0$ and $\|\bm S\|_1\to N$ as $d_S\to 0^+$. Next, we show that $(\bm S,\bm I)$ has the asymptotic profiles described in {\rm (i-1)} and {\rm (i-2)}.

\medskip

\noindent{\rm (i-1)} Next, suppose that $\bm r_M\ge 1$ or $\bm r_M<1$ and $N<\|{\bm\zeta\circ\bm r}/{(\bm 1-\bm r)}\|_1$.  By Lemma \ref{lem5}, for every $d_S>0$, there is $l>\mathcal{N}_0$ such that $(\bm S,\bm I)=(l(\bm \alpha-d_I{\bm P}^{(l)}),d_Sl{\bm P}^{(l)})$, where ${\bm P}^{(l)}$ is the unique positive solution of \eqref{transform equation}. We first claim that 
\begin{equation}\label{Eqc-8}
    \limsup_{d_S\to 0^+}l<\infty.
\end{equation}
If this is false, then possible after passing to a subsequence, we may suppose that $l\to \infty$ as $d_S\to 0^+$. Furthermore, by the Bolzano-Wierestrass theorem, possible after passing to a further subsequence, we may suppose that $l(\bm \alpha-d_I{\bm P}^{(l)})=\bm S\to \bm S^*$ as $l\to\infty$. Then 
\begin{equation}\label{Eqc-9}
\|\bm \alpha-d_I{\bm P}^{(l)}\|_1=\frac{1}{l}\|\bm S\|_1\le \frac{N}{l}\to 0 \quad \text{as}\ l\to\infty.
\end{equation}
Therefore, since $d_Sl\bm P^{(l)}=\bm I\to \bm 0$ as $d_S\to 0^+$, letting $l\to\infty$ in \eqref{transform equation}, we obtain that 
\begin{equation}\label{Eqc-10}
0=d_I\mathcal{L}(\frac{1}{d_I}\bm\alpha)+\bm \beta\circ(\bm S^*/(\bm \zeta +\bm S^*)-\bm r)\circ(\frac{1}{d_I}\bm \alpha),
\end{equation}
from which we deduce that $ \bm S^*/(\bm \zeta+\bm S^*)=\bm r$ since $\mathcal{L}\bm\alpha=0$ and $\bm \alpha\gg \bm 0$. Solving for $\bm S^*$, we obtain $\bm S^*=\bm\zeta\circ\bm r/(\bm 1-\bm r)$. Hence, we must have $\bm r_M<1$ and $N=\|\bm\zeta\circ\bm r/(\bm 1-\bm r)\|_{1}$, which is contrary to our initial assumption. Therefore \eqref{Eqc-8} holds. 

Since \eqref{Eqc-8} holds, after passing to a subsequence, we may suppose that $l\to l^*\in [\mathcal{N}_0,\infty)$ as $d_S\to 0^+$. Hence $(\bm S,\frac{1}{d_S}\bm I)=(l(\bm\alpha-d_I{\bm P}^{(l)}),l{\bm P}^{(l)})\to (l^*(\bm\alpha-d_I\bm P^{(l^*)}), l^*\bm P^{(l^*)})$ as $d_S\to 0^{+}$. To complete the proof of the result, it remains to argue that $l^*>\mathcal{N}_0$. If this were false, we would have that $\bm S\to \mathcal{N}_0\bm\alpha$, which yields $N=\|\mathcal{N}_0\bm\alpha\|_1=\mathcal{N}_0$. As a result, we get $\mathcal{R}_0=1$, so we get a contradiction. Hence, 
$l^*>\mathcal{N}_0$.

\medskip

\noindent{\rm (i-2)} Suppose that $\bm r_M<1$ and $N=\|\bm\zeta\circ\bm r/(\bm 1-\bm r)\|_1$. If \eqref{Eqc-8} holds, then we can proceed as above to establish that $(\bm S,\frac{1}{d_S}\bm I)$ has the asymptotic profiles described in {\rm(i-1)}. Now, suppose that \eqref{Eqc-8} is false. Thus, by the similar arguments leading to \eqref{Eqc-9}-\eqref{Eqc-10},  after passing to a further subequence,  $\bm S\to \bm S^*$ as $d_S\to 0^+$, where $\bm  S^*>\bm 0$ and satisfies $\bm S^*/(\bm \zeta+\bm S^*)=\bm r$. Solving for $\bm S^*$, we get $\bm S^*=\bm\zeta\circ\bm r/(\bm 1-\bm r)$.

\medskip

\noindent {\rm (ii)} Suppose that $\bm r_M<1$ and $N>\|\bm\zeta\circ\bm r/(\bm 1-\bm r)\|_1$. We proceed in two cases.\\
{\bf Case 1.} Here we suppose that  $M^*>0$. Then it follows from Claim 1,  the arguments leading to \eqref{Eqc-7}, and \eqref{Eqd-1} that, possible after passing to a subsequence, 
$$
\bm S\to \Big(\bm\zeta+\frac{(N-\|\bm\zeta\circ\bm r/(\bm 1-\bm r)\|_1)}{(1+\|\bm\alpha\circ\bm r/(\bm 1-\bm r)\|_1)}\bm\alpha\Big)\circ(\bm r/(\bm 1-\bm r))\quad \text{and} \quad \bm I\to \frac{(N-\|\bm\zeta\circ\bm r/(\bm 1-\bm r)\|_1)}{(1+\|\bm\alpha\circ\bm r/(\bm 1-\bm r)\|_1)}\bm\alpha$$
as $d_S\to 0^+$. In this case, we see that $(\bm S,\bm I)$ has the asymptotic profiles described in {\rm(ii-1)}.\\
{\bf Case 2.} Next, we suppose that $M^*=0$. Then $\bm I\to \bm 0$ and $\|\bm S\|_1\to N$ as $d_S\to 0^+$. Furthermore, since $N\ne \|\bm\zeta\circ\bm r/(\bm 1-\bm r)\|_1$, then $\bm \zeta\circ\bm r/(\bm 1-\bm r)$ is not a limit point of $\{\bm S\}_{d_S>0}$ as $d_S$ tends  to zero.  This shows that $\bm S$ doesn't have the asymptotic profiles in {\rm (i-2)} for any subsequence of $d_S$ converging to zero. Therefore, \eqref{Eqc-8} must hold and hence, up to a subsequence, $(\bm S,\frac{1}{d_S}\bm I)$ has the asymptotic profiles described in {\rm(i-1)} as $d_S\to 0^+$. In the current case, we see that $(\bm S,\bm I)$ has the asymptotic profiles described in {\rm (ii-2)} as $d_S\to 0^{+}$.

\medskip

 It follows from Case 1 and Case 2 that up to a subsequence, $(\bm S,\bm I)$ has one of the asymptotic profiles in {\rm (ii-1)} or {\rm (ii-2)} as $d_S$ tends to zero.   Finally, suppose in addition that either  $N>\|\bm r\circ\bm\zeta/((\bm 1-\bm r)\circ\bm\alpha)\|_{\infty}$ or $N=\|\bm r\circ\bm\zeta/((\bm 1-\bm r)\circ\bm\alpha)\|_{\infty}$ and $\bm \zeta\circ\bm r/((\bm 1-\bm r )\circ\bm\alpha){\notin}{\rm span}(\bm 1)$. We claim that 
$$
M_*:=\liminf_{d_S\to 0^+}\sum_{j\in\Omega}I_j>0.
$$
Suppose to the contrary that $M_*=0$. Hence, possibly after passing to a subsequence, we may suppose that $\sum_{j\in\Omega}I_j\to 0$ as $d_S\to 0^+$. Hence, $(\bm S,\bm I)$ has the asymptotic profiles described in {\rm (ii-2)}. Consequently,  there is $l^*>0$ such that $(\bm S,\frac{1}{d_S}\bm I)=(l(\bm\alpha-d_I{\bm P}^{(l)}),l{\bm P}^{(l)})\to (l^*(\bm\alpha-d_I\bm P^{(l^*)}), l^*\bm P^{(l^*)})$ as $d_S\to 0^{+}$. Setting $ \bm S^*=l^*(\bm\alpha-d_I\bm P^{(l^*)})$, then 
$$
({\rm diag}(\hat{\bm G})-\mathcal{L})\bm S^*={\hat{\bm G}}\circ\bm \zeta\circ\bm r/(\bm 1-\bm r ),
$$
where ${\hat{\bm G}:=l^*(\bm 1-\bm r)\circ\bm \beta\circ\bm P^{l^*}/(\bm\zeta+\bm S^*)\gg \bm0.}$
 Hence, noting that  $\bar{\bm S}^*:=\|\bm \zeta\circ\bm r/((\bm 1-\bm r )\circ\bm\alpha)\|_{\infty}\bm \alpha$ satisfies 
$$
({\rm diag}(\hat{\bm G})-\mathcal{L})\bar{\bm S}^*={\hat{\bm G}}\circ\bar{\bm S}^*\ge \hat{\bm G}\circ\bm \zeta\circ\bm r/(\bm 1-\bm r ),
$$
 $\sigma_*(\mathcal{L}-{\rm diag}({\hat{\bm G}}))<\sigma_*(\mathcal{L})=0$, and $\mathcal{L}-{\rm diag}({\hat{\bm G}})$ is quasipositive and irreducible, then by the comparison principle, we have that $\bm S^*\le \bar{\bm S}^*$ with a strict inequality if $\bm \zeta\circ\bm r/((\bm 1-\bm r )\circ\bm\alpha) \notin {\rm span}(\bm 1)$. Therefore, $N=\|\bm S^*\|_1\le \|\bar{\bm S}^*\|_1= \|\bm \zeta\circ\bm r/((\bm 1-\bm r )\circ\bm\alpha)\|_{\infty}$. This contradicts our initial assumption on $N$ and the fact $\bm S^*\ll \bar{\bm S}^*$ if $\bm \zeta\circ\bm r/((\bm 1-\bm r )\circ\bm\alpha)\notin{\rm span}(\bm 1)$.  Therefore,  $M_*>0$. This rules out {\rm (ii-2)},  hence {\rm (ii-1)}  holds.
    
\end{proof}

\begin{proof}[Proof of Theorem \ref{TH8}] Fix $d_S>0$, $N>0$ and  suppose that $\|N\bm\alpha/(\bm r\circ(\bm \zeta+N\bm\alpha) )\|_{\infty}>1$. Then, by \eqref{R-0-limit}, there is $d_1>0$ such that $\mathcal{R}_0>1$ for all $0<d_I<d_1$. It then follows from Theorem \ref{TH5}-{\rm(i)} that system \eqref{model} has a unique EE solution $(\bm S, \bm I)$ for every $0<d_I<d_0:=\min\{d_1,d_S\}$. Now,  for every $0<d_I<d_0$, by Lemma \ref{lem5}-{\rm (i)}  there is $\kappa>0$ such that \eqref{Eqc-1} holds. By the similar arguments in \eqref{Eqd-1}, we have 
\begin{align}\label{Eqd-3}
   \|\bm S-(\sum_{j\in\Omega}S_j)\bm \alpha\|_1\le& \Big\|\bm S-\frac{\kappa}{d_S}\bm\alpha\Big\|_1+\Big\|(\sum_{j\in\Omega}S_j-\frac{\kappa}{d_S})\bm\alpha\Big\|_1\cr
   =& \Big\|\bm S-\frac{\kappa}{d_S}\bm\alpha\Big\|_1+\Big|\|\bm S\|_1-\|\frac{\kappa}{d_S}\bm\alpha\|_1\Big|\le 2\Big\|\bm S-\frac{\kappa}{d_S}\bm\alpha\Big\|_1\le \frac{2d_I}{d_S}N\to 0\ \text{as}\ d_I\to 0^+.
\end{align}
Next, from the second equation of \eqref{EE-system}, using the quadratic formula and the positivity of $\bm I$,  we have 
\begin{equation}\label{Eqd-5}
    I_j=\frac{\Big(\frac{d_I}{\beta_j}B_j+A_j\Big)+\sqrt{\Big(\frac{d_I}{\beta_j}B_j+A_j\Big)^2+4\frac{d_I}{\beta_j}(r_j-\frac{d_I}{\beta_j}L_{jj})(\zeta_j+S_j)\sum_{i\in\Omega\setminus\{j\}}L_{ji}I_i}}{2\Big(r_j-\frac{d_I}{\beta_j}L_{jj}\Big)} \quad j\in\Omega,
\end{equation}
where  $B_j:=\sum_{i\ne j}L_{ji}I_i+L_{jj}(\zeta_j+S_j)$ and $A_j:=(S_j-r_j(\zeta_j+S_j))$ for all $ j\in\Omega.$  Since $N=\|\bm S\|_1+\|\bm I\|_1$ for all $d_I>0$, then thanks to \eqref{Eqd-3}, possibly after passing to a subsequence, we may suppose that $\bm S \to m\bm\alpha$ as $d_I\to 0^+$ for some $m\in[0,N]$. It then follows from \eqref{Eqd-5} that
\begin{equation}\label{Eqd-6}
I_j\to \frac{(m(1-r_j)\alpha_j-r_j\zeta_j)_+}{r_j} \quad \text{as} \ d_I\to0^+, \quad \forall\ j\in\Omega.
\end{equation}
Hence, we must have that 
\begin{equation}\label{Eqd-4}
    N=\sum_{j\in\Omega}m\alpha_j+\sum_{j\in\Omega}\frac{(m(1-r_j)\alpha_j-r_j\zeta_j)_+}{r_j}=m+\sum_{j\in\Omega}\frac{(m(1-r_j)\alpha_j-r_j\zeta_j)_+}{r_j}=m+\sum_{j\in\tilde{\Omega}}\frac{(m(1-r_j)\alpha_j-r_j\zeta_j)_+}{r_j},
\end{equation}
where $\tilde{\Omega}:=\{j\in\Omega : r_j<1\}$. Note that $\tilde{\Omega}\ne\emptyset$ since $\|N\bm\alpha/(\bm r\circ(\bm\zeta+N\bm\alpha))\|_{\infty}>1$. It is easy to see that the function 
$$
(0,\infty)\ni m\mapsto m+\sum_{j\in\tilde{\Omega}}\frac{(m(1-r_j)\alpha_j-r_j\zeta_j)_+}{r_j} 
$$
is strictly increasing, continuous,  
$$
\lim_{m\to 0^+}\Big(m+\sum_{j\in\tilde{\Omega}}\frac{(m(1-r_j)\alpha_j-r_j\zeta_j)_+}{r_j}\Big)=0\quad \text{and}\quad \lim_{m\to\infty}\Big(m+\sum_{j\in\tilde{\Omega}}\frac{(m(1-r_j)\alpha_j-r_j\zeta_j)_+}{r_j}\Big)=\infty.
$$
It then follows from the implicit function theorem that the algebraic equation \eqref{Eqd-4} has a unique root. This shows $m\in[0,N]$ is independent of the chosen subsequence, and hence $\bm S\to m\bm \alpha$ as $d_I\to 0^+$, where $m\in[0,N]$ is the unique root of \eqref{Eqd-4}. It is clear from \eqref{Eqd-4} that $m>0$ since $N>0$. Next, if $m=N$, then we must have that  $\sum_{j\in\Omega}\frac{(N(1-r_j)\alpha_j-r_j\zeta_j)_+}{r_j}=0$, from which it follows that $N(1-r_j)\alpha_j\le r_j\zeta_j$ for all $j\in\Omega.$ Equivalently, $N\alpha_j/(r_j(\zeta_j+N\alpha_j))\le 1 $ for all $j\in\Omega$. This contradicts our initial assumption $\|N\bm \alpha/(\bm r\circ(\bm\zeta+N\bm\alpha))\|_{\infty}>1$. Therefore, we must also have that $m<N$. Recalling that \eqref{Eqd-6} holds, $\bm S\to m\bm \alpha$ as $d_I\to0^+$, and $0<m<N$ satisfies \eqref{Eqd-4}, the result follows.
    
\end{proof}

 \subsection*{Declarations}
{\bf Ethical Approval:} Not applicable for this study.

\noindent{\bf Competing interests:} The authors declare  that there is no competing interest.

\noindent{\bf Authors' contributions:} All authors contributed equally in designing and conducting the study.



\noindent{\bf Availability of data and materials:} Not applicable.


\end{document}